%% file: main.tex
\documentclass[a4paper,UKenglish,cleveref, numberwithinsect]{lipics-v2021}
\hideLIPIcs
\usepackage[ruled, vlined, linesnumbered]{algorithm2e}
\usepackage{colonequals}
\usepackage{stmaryrd}
\usepackage{tabularx}
\usepackage{booktabs}
\usepackage{bm}
\usepackage{tikz}
\usepackage{nicefrac}
\usepackage{mathabx} %
\usetikzlibrary{decorations.pathreplacing,calligraphy} %
\usetikzlibrary{arrows.meta}
\usepackage[bibliography=common]{apxproof}

\usepackage{relsize}
\usepackage{stackengine}
\stackMath
\usepackage{nameref}
\usepackage{mathtools}

\newcommand{\TreeDP}{\texttt{TreeDP}}
\newcommand{\JoinNode}{\textup{\textsc{Join}}}
\newcommand{\IntroVertNode}{\textup{\textsc{Introduce Vertex}}}
\newcommand{\ForgetVertNode}{\textup{\textsc{Forget Vertex}}}
\newcommand{\IntroEdgeNode}{\textup{\textsc{Introduce Edge}}}
\newcommand{\LeafNode}{\textup{\textsc{Leaf}}}

\crefname{claim}{Claim}{Claims}

\newtheoremrep{theorem}{Theorem}[section]

\newtheoremrep{proposition}[theorem]{Proposition}
\newtheoremrep{fact}[theorem]{Fact}
\newtheoremrep{lemma}[theorem]{Lemma}
\newtheoremrep{claim}[theorem]{Claim}
\newtheoremrep{example}[theorem]{Example}
\newtheoremrep{definition}[theorem]{Definition}

\input{macros}
\title{Path Length Existence Problems in Bounded-Treewidth Probabilistic Graphs} %

\author{Antoine Amarilli}{LTCI, Télécom Paris, Institut Polytechnique de Paris,
  Palaiseau, France \and Univ. Lille, Inria, CNRS, Centrale Lille, UMR 9189
CRIStAL, F-59000 Lille, France  \and \url{https://a3nm.net/}}{a3nm@a3nm.net}{https://orcid.org/0000-0002-7977-4441}{}
\author{Mikaël Monet}{Univ. Lille, Inria, CNRS, Centrale Lille, UMR 9189
CRIStAL, F-59000 Lille, France \and
\url{https://mikael-monet.net/}}{mikael.monet@inria.fr}{https://orcid.org/0000-0002-6158-4607}{}

\author{Harshul Sagar}{University of Illinois at Urbana-Champaign, Champaign, Illinois}{harshul4@illinois.edu}{https://orcid.org/0009-0003-3802-872X}{}

\authorrunning{A.\ Amarilli, M.\ Monet, and H.\ Sagar}

\Copyright{Antoine Amarilli, Mikaël Monet, and H.\ Sagar}



\nolinenumbers %

\EventEditors{John Q. Open and Joan R. Access}
\EventNoEds{2}
\EventLongTitle{42nd Conference on Very Important Topics (CVIT 2016)}
\EventShortTitle{CVIT 2016}
\EventAcronym{CVIT}
\EventYear{2016}
\EventDate{December 24--27, 2016}
\EventLocation{Little Whinging, United Kingdom}
\EventLogo{}
\SeriesVolume{42}
\ArticleNo{23}

\begin{document}

\maketitle

\begin{abstract}
We consider the problem of testing the existence of a source-to-target path satisfying an input-specified length constraint in an input probabilistic graph, as a simplified model of query evaluation on a probabilistic database. More precisely, we wish to determine the probability that a suitable walk or simple path exists, and consider the cases of a length constraint requiring the length to be at least $\ell$, at most $\ell$, or exactly $\ell$, for $\ell$ a target length given as input and written in binary. 
The probabilistic graphs that we consider are directed graphs where each edge is annotated with an independent probability of existence. 
Our goal is to determine, given such a graph, endpoints, and length constraint, what is the probability that a requisite path/walk exists, i.e., the total probability of the outcomes that satisfy this requirement.
As this task is generally intractable (\#P-hard), we aim for algorithms applying to input graphs that are assumed to be of bounded treewidth.

We show that, in contrast to many other tasks on probabilistic graphs (e.g., computing the probability of MSO formulas), bounding the treewidth does not always imply tractability of our problems, and we characterise the resulting complexity landscape.
Specifically: 
  \begin{itemize}
    \item We give an XP algorithm parameterised by treewidth to compute the probability that an input probabilistic graph contains a source-to-target path/walk of length $\geq \ell$, and do the same for a path/walk of length $\leq \ell$. 
  This implies in particular that these problems are in PTIME on graphs whose treewidth is bounded by a constant. 
    \item We show that XP is likely the best one can hope for the four problems mentioned in the last item, as we show that these are W[1]-hard (under Turing reductions) parameterised by pathwidth (hence not FPT under the standard assumption that FPT$\neq$W[1]). This is in contrast to the case of non-probabilistic graphs, where the corresponding problems were known to be FPT parameterised by treewidth; and to the connectivity problem (and more generally MSO evaluation) on probabilistic graphs, where the same was known.
    \item Surprisingly, we show that it is \#P-hard to compute the probability that an input graph contains a source-to-target path/walk of length \emph{exactly}~$\ell$, even on graphs of pathwidth~$2$.
      By contrast, all our problems are tractable on graphs of treewidth~$1$.
  \item
    Last, we investigate the complexity of our problems when the edges can carry weights. This causes all our problems to become \#P-hard already on pathwidth-2
    graphs, and causes specifically the ``exactly $\ell$'' walk problem to be \#P-hard even on pathwidth-1 graphs.
  \end{itemize}
\end{abstract}

\vskip\topmattervskip\baselineskip\noindent\noindent\textcolor{lipicsGray}{\fontsize{9}{12}\sffamily\bfseries Usage of AI\enskip} {\fontsize{9}{12}\selectfont We used the free version of LLMs for two purposes. First, to find related work we might have missed and for minor uses (proofreading, LaTeX help, TikZ, etc.). Second, the proof of the W[1]-hardness result (\cref{thm:w1hard}) was found with LLM assistance.

\vfill
\pagebreak

\section{Introduction}
\label{sec:introduction}
\input{introduction}

\section{Preliminaries and Problem Statement}
\label{sec:prelims}
\input{preliminaries}
\section{Treewidth}
\label{sec:treewidth}
\input{treewidth}

\section{Introducing Probabilities}
\label{sec:probabilities}
\input{proba}
\section{Tractable Cases on Bounded-Treewidth Probabilistic Graphs}
\label{sec:tract}
\input{tract}
\section{Intractable Cases on Treewidth-2 Probabilistic Graphs}
\label{sec:untract}
\input{untract}

\section{Weighted Graphs}
\label{sec:variants}

\input{variants}

\section{Conclusion and Future Work}
\label{sec:conclusion}
\input{conclusion}
\bibliography{main}

\end{document}

%% file: macros.tex
\AddToHook{env/proposition/begin}{\crefalias{theorem}{proposition}}
\AddToHook{env/lemma/begin}{\crefalias{theorem}{lemma}}
\AddToHook{env/example/begin}{\crefalias{theorem}{example}}
\AddToHook{env/definition/begin}{\crefalias{theorem}{definition}}
\AddToHook{env/corollary/begin}{\crefalias{theorem}{corollary}}
\AddToHook{env/claim/begin}{\crefalias{theorem}{claim}}
\AddToHook{env/remark/begin}{\crefalias{theorem}{remark}}
\AddToHook{env/conjecture/begin}{\crefalias{theorem}{conjecture}}

\newcommand\probname{\mathsf{PLEP}}
\newcommand\genprob{\probname^{X}_{\bowtie}}
\newcommand\prob[2]{\probname^{#1}_{#2}}
\renewcommand\Pr{\mathrm{Pr}}

\newcommand{\spl}{\mathrm{SPL}}

\newcommand\apsp{\mathrm{ALG\!-\!\!APSP}}

\newcommand\psp{\mathrm{PSP}}
\newcommand\slpdf{\mathrm{\texttt{SLPDC}}}
\newcommand\llpdf{\mathrm{\texttt{LLPDC}}}
\newcommand\lee{\mathrm{\texttt{$\#$LE}}}

\newcommand\pptdnf{\mathrm{PP2DNF}}
\newcommand\spptdnf{\mathrm{\#PP2DNF}}

\newcommand\knapsack{\mathrm{\#KNAPSACK}}
\newcommand\multsubsm{\mathrm{\#UBK}}

\newcommand\calS{\mathcal{S}}
\newcommand\calO{\mathcal{O}}
\newcommand\calP{\mathcal{P}}
\newcommand\proj{\mathrm{Proj}}
\newcommand\length{\mathrm{length}}
\newcommand\calG{\mathcal{G}}
\newcommand\G{\mathcal{G}}

\newcommand{\NN}{\mathbb{N}}

\newcounter{sqindex}

\newcommand{\Walk}{\mathsf{Walk}}
\newcommand{\Path}{\mathsf{Path}}
\newcommand{\lppe}{\probname^{\Path}_=}
\newcommand{\lppg}{\probname^{\Path}_\geq}
\newcommand{\lppl}{\probname^{\Path}_\leq}
\newcommand{\lpwe}{\probname^{\Walk}_=}
\newcommand{\lpwl}{\probname^{\Walk}_\leq}
\newcommand{\lpwg}{\probname^{\Walk}_\geq}
\newcommand\treewidth{\mathsf{tw}}
\newcommand\ldpl{\mathrm{LDPL}}
\newcommand\cpsp{\mathrm{CPSP}}
\newcommand\CPSP{\cpsp}

\newcommand{\myparagraph}[1]{\subparagraph*{#1.}}

%% file: introduction.tex
\emph{Probabilistic graphs} have been studied in computer science both as a theoretical object of study and as an applicative model~\cite{saha2021shortest}, e.g.,
describing possible failures in a computer network, closures on road
networks, uncertain connections in social networks, etc. We consider the simplest such model:
every edge is
annotated with a probability, and the edge is then present with this
probability, independently from all other edges.
Such graphs are an arity-two case of 
\emph{tuple-independent probabilistic databases}~\cite{suciu2011probabilistic},
which have also been studied in their own
right~\cite{amarilli2017conjunctive,maniu2017indexing,amarilli2022dichotomy}.

We study the question of efficiently finding paths in such graphs, which is
related to the \emph{two-terminal network reliability} problem in
probabilistic graphs (ST-CON). This problem asks,
given a probabilistic graph $\mathcal{G}$, a source $s$, and a target $t$, what is the
probability that there is some path from~$s$ to~$t$, i.e., what is the total
probability of the graphs having such a path in the distribution represented by
the graph~$\mathcal{G}$. We can of course always solve this problem exactly by
considering every possible realisation of~$\mathcal{G}$, solving the
(non-probabilistic) $st$-reachability problem, and summing the probabilities of
the realisations where a path exists. However, in general the number of
realisations is exponential in the size of~$\mathcal{G}$, so this algorithm is
inefficient. It has in fact been shown~\cite[Theorem~1]{V79}
that ST-CON is \emph{\#P-complete}, 
so it is unlikely that substantially better algorithms exist. 

However, the tractability of ST-CON is regained when the probabilistic
graph is assumed to have bounded
\emph{treewidth}~\cite{arnborg1991easy,wolle2002framework,goharshady2020efficient}, i.e., 
to be similar to a tree. More precisely, for any
treewidth bound $k \in \mathbb{N}$, the ST-CON problem can be solved in
polynomial time when the input probabilistic graphs are restricted to have
treewidth at most~$k$; in fact the problem is fixed-parameter tractable (FPT)
when parameterised by treewidth.
This observation generalises to the fixed-parameter tractability of the class of \emph{monadic
second-order queries} (MSO) on tuple-independent probabilistic
databases parameterized by treewidth, namely, the problem is in linear time in the database up to the cost of arithmetic operations~\cite{amarilli2015provenance,amarilli2016tractable}.
While the motivation of treewidth is theoretical, its
amenability to dynamic programming techniques makes it potentially attractive
for practical applications, and the experimental treewidth of practical data has
for instance been investigated in earlier work~\cite{maniu2019experimental}.

However, these results do not address the task of finding paths of
\emph{specified lengths} in probabilistic graphs, a query which is not
generally expressible in MSO. For instance, 
we can ask for the probability that the shortest path from~$s$ to~$t$ has length less than some
threshold, e.g., to check if $t$ can be reached from~$s$ in a
road network with edge probabilities
subject to a time budget. Such questions have been
investigated~\cite{saha2021shortest,yuan2010efficiently}, in particular in
connection to the computation of  
\emph{most probable shortest paths} that form a kind of generalisation of
shortest paths in probabilistic graphs. However, while these tasks are generally
intractable on unrestricted probabilistic graphs, to our knowledge it had not
been investigated whether tractability could also be regained in this
setting by bounding the treewidth, ideally with an FPT algorithm.
The one exception we are aware of is the very
recent preprint of Ando~\cite{ando2026complexity}, which gives a \#P-hardness proof and an XP algorithm
parameterised by treewidth to compute the distribution of the longest and
shortest simple paths in a probabilistic graph. However their setting is
different because they consider
\emph{edge lengths} that are probabilistic (according to a continuous
distribution), whereas we consider graphs with edges
annotated by a probability of existence. It turns out however that some aspects of the \#P-hardness proof of~\cite{ando2026complexity} are related to the W[1]-hardness proof that we present later.

\subparagraph*{Contributions.}
In this paper, we consider the question of computing paths satisfying length
constraints in probabilistic graphs. More precisely, we consider three types of
length requirements: asking for paths of length $\leq \ell$, or $= \ell$, or
$\geq \ell$, for an integer threshold $\ell$ 
given as input (and written in binary). We consider the search for simple paths, and 
the search for walks (which may revisit a given vertex): while this choice makes no difference for shortest paths (or
paths of length $\leq \ell$), it may make a difference in general because
revisiting vertices may be useful to achieve higher lengths. This gives six
problem combinations: $\lppl$, $\lppe$, and $\lppg$ for paths; and $\lpwl$,
$\lpwe$, and $\lpwg$ for walks.

For each of these six problems, the input consists of a probabilistic graph
$\mathcal{G}$, a source vertex~$s$ and target vertex $t$, and the
threshold~$\ell$. Our task is to compute the total probability of the
realisations of~$\mathcal{G}$ where
a suitable path from $s$ to $t$ exists. Our
goal is to study the complexity of these problems when the input graph
$\mathcal{G}$ has bounded treewidth.

To our surprise, the complexity landscape turns out to be more complicated than
expected. We initially expected these problems to behave like ST-CON,
and to be fixed-parameter tractable on bounded-treewidth data by being captured
by a suitable logic
(e.g., an extension of monadic second-order logic, in the  spirit of
quantitative MSO~\cite{kreutzer2013quantitative,knop2019simplified}  or other
extensions~\cite{knop2019simplified}). However, it turns out that tractability
does not always hold, 
that they generally do not enjoy FPT algorithms (conditionally),
and that in particular the six problems do not have the same complexity.
We now present the results that we obtain, and the structure of the paper. 

After
formally defining our problems in
\cref{sec:prelims}, we first review the complexity of the six problems on
non-probabilistic graphs in \cref{sec:treewidth}. Indeed, the problems $\lppe$ and $\lppg$, which ask
for long simple paths, are NP-hard already on non-probabilistic graphs because
of an immediate reduction from the Hamiltonian path problem. 
So we introduce treewidth as a way to regain tractability, and
show that all six variants of
our problems are FPT when parameterised by treewidth.
This result is not
surprising and follows known constructions, e.g., for the  Hamiltonian path
problem on bounded-treewidth graphs~\cite{ziobro2019finding}, for which
efficient FPT algorithms are
known~\cite{bodlaender2015deterministic,cygan2018fast,cygan2022solving}
achieving $c^{\calO(k)} n^{\calO(1)}$ for some constant $c$ and with $k$ the
input treewidth. We will sketch in the section a self-contained proof of these
results, introducing notions that will be lifted later to study $\lppg$ on
bounded-treewidth probabilistic graphs.
The situation  for non-probabilistic graphs is summarised in \cref{tab:nonproba}.

We then move to probabilistic graphs in \cref{sec:probabilities}. Like ST-CON,
all six problems become \#P-hard when the input probabilistic graphs are
arbitrary; so we study the impact of bounding the treewidth. We show that,
for treewidth-1 graphs (i.e., probabilistic forests), the
six problems indeed become tractable, by an ad-hoc dynamic programming argument.

The key technical results of the paper are given in \cref{sec:tract}, where we
investigate whether tractability extends to arbitrary treewidth bounds. 
We show tractability on bounded-treewidth probabilistic graphs of four of
our six problems: $\lppl$, $\lpwl$, $\lpwg$, and $\lppg$.
Namely, we show that these algorithms are XP parameterised by treewidth, and
that they are likely not FPT because they are W[1]-hard already on acyclic
graphs and when parameterising by pathwidth. Note that this contrasts both with the
FPT results for these problems on non-probabilistic graphs, and with the FPT
results for ST-CON on probabilistic graphs.
Our XP algorithms are shown by standard dynamic programming
arguments but with significant complications to account for the probabilistic
computations: we first cover the problems $\lppl$ and $\lpwl$, then the problem
$\lpwg$ which adds the computation of handling ``arbitrarily high'' walk
lengths, then the problem $\lppg$ which asks for long simple paths. This latter
problem is the most complicated one to cover, as we must both perform
probability computations and ensure that the witnessing paths are simple,
generalising the technique used on non-probabilistic graphs.
In all these results, we use the ``monotonicity'' of the $\leq \ell$ and $\geq
\ell$ requirements: they intuitively make it possible to concisely summarise
realisations in a bottom-up dynamic programming algorithm on the tree
decomposition, by considering the best possible choices of paths that can be
achieved below each bag and using this as a way to group realisations together.
The W[1]-hardness proof is by a reduction from the problem of counting the number of linear extensions of a partially ordered set, like in \cite{ando2026complexity}.

The natural question is whether we can similarly design tractable algorithms for $\lppe$
and $\lpwe$ on bounded-treewidth probabilistic graphs. Surprisingly, we show in
\cref{sec:untract} that the answer is negative (under standard complexity
assumptions). Namely, we show that the two problems $\lppe$ and $\lpwe$ are \#P-hard already on
pathwidth-2 probabilistic graphs. 
Note that, because we are talking about the probability of existence of paths of
a certain length, 
the probability of 
existence of a path of length exactly $\ell$ cannot be obtained simply by
subtracting, e.g., the probability of existence of paths of length $\leq \ell-1$  from the
probability of existence of paths of length $\leq \ell$ (as the corresponding
outcomes are not mutually exclusive).
The situation is summarised in 
\cref{tab:proba}.

We then study in \cref{sec:variants} the variants of our problems on \emph{weighted graphs}, i.e., the edges of the graph carry
weights (written in binary) in addition to their probabilities, and the length
requirements now pertain to the total weights of paths. The use of weights can
be relevant, e.g., to model lengths or durations in routing applications. We
show that adding weights can have a significant impact on the complexity of our
problems: it makes $\lpwe$ and $\lppe$ NP-hard already on pathwidth-2
non-probabilistic graphs, and in the presence of both probabilities and weights
all our problem variants become \#P-hard on pathwidth-2
graphs. We also 
show \#P-hardness on pathwidth-1 graphs for $\lpwe$ specifically when we have
weights and probabilities, remembering that the problem $\lpwe$ was tractable on
treewidth-1 unweighted probabilistic graphs.

We conclude in \cref{sec:conclusion}. Most detailed proofs are deferred to the
appendix.

\input{tables-new}

%% file: tables-new.tex
\begin{table}[t]
\centering
\begin{subtable}[b]{0.43\linewidth}
\centering
\label{table:results:nonproba}
\begin{tabular}{l@{~}ccc}
\toprule
$\downarrow\! X$ $\bowtie \rightarrow$ & $\leq \ell$ & $= \ell$ & $\geq \ell$\\
\midrule
$\Walk$ & PTIME & PTIME & PTIME\\
\midrule
\shortstack{$\Path$\\~} & \shortstack{PTIME\\~} & \shortstack{NP-c.\\FPT} &
\shortstack{NP-c.\\FPT}\\
\bottomrule
\end{tabular}
\caption{On non-probabilistic graphs}
\label{tab:nonproba}
\end{subtable}%
\hfill
\begin{subtable}[b]{0.57\linewidth}
\centering
\label{table:results:proba}
\begin{tabular}{l@{}c@{~}c@{~}c}
\toprule
$\downarrow\! X$ $\bowtie \rightarrow$ & $\leq \ell$ & $= \ell$ & $\geq \ell$\\
\midrule
\shortstack{$\Walk$\\~\\~\\~} & \shortstack{\#P-h.\\XP (tw)\\W[1]-h.\ (pw)} & \shortstack{\#P-h. (pw
2)\\FP (tw 1)\\~} & \shortstack{\#P-h.\\XP\\W[1]-h.}\\
\midrule
\shortstack{$\Path$\\~\\~\\~} & \shortstack{\#P-h.\\XP (tw)\\W[1]-h.\ (pw)} & \shortstack{\#P-h. (pw
2)\\FP (tw 1)\\~} & \shortstack{\#P-h.\\XP (tw)\\W[1]-h.\ (pw)}\\
\bottomrule
\end{tabular}
\caption{On probabilistic graphs}
\label{tab:proba}
\end{subtable}
  \caption{Summary of the complexity of the problems $\genprob$}
\end{table}

%% file: preliminaries.tex
\subparagraph*{Graphs, walks, (simple) paths.} A (directed) graph
$G=(V,E)$ consists of a finite set $V$ of \emph{vertices} and of a
set $E\subseteq V\times V$ of \emph{edges}. For $E'\subseteq E$, we
write $G[E']$ the subgraph defined by $G[E'] = (V,E')$.
A \emph{walk} $w$ in $G$ is a non-empty sequence of vertices $w = v_1,\ldots,v_n$ of~$V$
such that $(v_i,v_{i+1}) \in E$ for $1 \leq i \leq n-1$, and its \emph{length}, denoted $|w|$, is $n-1$. 
A
\emph{simple path}, or simply \emph{path}, is a walk in which all
vertices are distinct. 
The \emph{endpoints} of a walk are the first and last vertices, and its
\emph{internal vertices} are all other vertices.
For $G$ a graph, we sometimes write $V(G)$ for its set of vertices and $E(G)$ for its set of edges.

\subparagraph*{Path length problems over non-probabilistic graphs.}
The problems studied in this paper deal with the existence of walks or paths from one vertex to another in 
graphs, where the walks/paths satisfy a length constraint.
We focus on the following constraints:
the walk/path must have length \emph{at most}, \emph{exactly}, or \emph{at
least}~$\ell$, where $\ell$ is an integer given as input in binary. This
gives rise to $6$ computational problems, which we denote
$\probname^{X}_{\bowtie}$ for ${\bowtie} \in \{{\leq}, {=},
{\geq}\}$ and $X \in \{\Walk, \Path\}$, defined as
follows (“$\probname$” stands for \emph{Path Length Existence Problem}):

\begin{center}
\fbox{\begin{tabular}{lp{10cm}}
    \small{PROBLEM} : & $\probname^{X}_{\bowtie}$ for ${\bowtie} \in \{{\leq}, {=}, {\geq}\}$ and for $X \in \{\Walk, \Path\}$ %
\\
    {\small INPUT} : & A graph $G = (V,E)$, two vertices $s,t\in V$, and an integer $\ell$ (given in binary)
\\
{\small OUTPUT} : & Decide whether a path/walk of length $\bowtie \ell$ exists from $s$ to $t$
\end{tabular}}
\end{center}

For instance, the problem $\lppl$ asks 
whether there is a path from $s$ to $t$ of length at
most $\ell$. Equivalently, it asks whether the
length of a shortest path from~$s$ to~$t$ is at
most $\ell$.
Further note that $\lppl $ and
$\lpwl$ are identical, because shortest
walks are always paths; but the same is not true of
$\lppg$ and $\lpwg$, or
of $\lppe $ and $\lpwe $.
Before moving on to the probabilistic setting, we recall known results on the complexity of these (rather standard) problems
in the \emph{non-probabilistic setting}
(see \cref{tab:nonproba}).
Problem $\prob{\Walk}{\leq} = \prob{\Path}{\leq}$ simply amounts
to computing the length of a shortest path from~$s$ to~$t$,
which is PTIME. Problems $\prob{\Path}{=}$ and $\prob{\Path}{\geq}$
are (obviously in NP and) NP-hard by reduction from Hamiltonian
path with prescribed endpoints (even when the target length is given in unary).
The problem
$\prob{\Walk}{=}$ is PTIME, using fast
exponentiation of the graph's adjacency matrix. The problem
$\prob{\Walk}{\geq}$ is PTIME too, by creating $\ell$ copies of
the graph and performing a reachability test (assuming
without loss of generality that $\ell$ is at most $|V|$).

\subparagraph*{Probabilistic graphs.}  A \emph{probabilistic graph}
is a triple $\G = (V,E,\pi)$, where $G=(V,E)$ is a graph, called the \emph{underlying non-probabilistic graph}, and
$\pi\colon E \to [0,1]$ assigns a probability value to every edge.
A \emph{realisation} of $\mathcal{G}$ is a graph
of the form $G[E']$ for $E'\subseteq E$. The probabilistic graph
$\calG$ defines a probability distribution on the set of its
realisations, by keeping or discarding each edge $e\in E$ with
probability $\pi(e)$, independently across edges. Formally, for a
given realisation $G[E']$ of~$\calG$, its probability, denoted
$\Pr_\G(E')$, is the quantity $\Pr_\G(E') = \prod_{e \in E'} \pi(e)
\cdot \prod_{e \in E \setminus E'} (1 - \pi(e))$.
Note that realisations may have probability~0.

\subparagraph*{Problem statement.}
We extend the $\probname$ problems to probabilistic graphs in the
natural way by asking for
the probability that the input graph 
satisfies the required conditions.
More precisely, for 
${\bowtie} \in \{{\leq}, {=}, {\geq}\}$ and for $X \in \{\Walk, \Path\}$, the
$\probname^{X}_{\bowtie}$ problem is 
as follows. We receive as input a
probabilistic graph $\G = (V, E, \pi)$, vertices $s$ and $t$, and an
integer $\ell\geq 0$ written in binary. 
Our task is to compute the sum over all $E' \subseteq E$
of the probabilities $\Pr_\G(E')$ of the
realisations of $\G$ that satisfy $\probname^{X}_{\bowtie}$ for $s$, $t$, and
$\ell$.
Notice that the non-probabilistic setting is a special case of the probabilistic setting in which all edges have probability $1$.

\subparagraph*{Computational model.}
Recall that \#P is the class of counting problems that can be expressed as counting the number of accepting paths of a nondeterministic polynomial-time Turing machine. 
Following Valiant~\cite{valiant1979complexity,V79}, we 
define \#P-hardness using Turing reductions (i.e., with oracles).
In all our problems, unless stated otherwise, probability values are given as rational numbers $\frac{p}{q}$ with
$(p,q) \in \mathbb{N} \times (\mathbb{N}\setminus \{0\})$, where~$p$ and~$q$ are encoded in binary. 
All probability computation problems that we study are in $\text{FP}^{\#P}$, i.e., they can be solved in polynomial time with access to a \#P oracle. 

%% file: treewidth.tex
We have seen that already in the non-probabilistic setting, the $\probname$ problems may be  
intractable: this is specifically the case of 
$\lppe$ and $\lppg$.
One common way to regain tractability of problems on graphs
is to bound their \emph{treewidth}.
In this section, we recall the definition of treewidth, defining
tree decompositions and mentioning the \emph{nice tree
decompositions} that we will use in our proofs. We then review the
folklore result that all $\probname$ problems are tractable on
non-probabilistic graphs when the input is assumed to have bounded
treewidth. This only needs to be explained for 
$\prob{\Path}{=}$ and $\prob{\Path}{\geq}$ since the other problems are
tractable on arbitrary graphs. Namely, we will review the following
result:

\begin{theorem}[Folklore]
\label{thm:nonprobatw}
    For non-probabilistic graphs, the problems $\lppe$ and $\lppg$ are FPT parameterised by treewidth.
\end{theorem}

Once this is done, our goal in the rest of the paper will be to study whether
this tractability result extends to the probabilistic setting. 

\myparagraph{Treewidth and nice tree decompositions}
A \textit{tree decomposition} of a graph \( G = (V,E)\) is a pair $
\mathcal{S} = (S, \{S_b\}_{b \in V(S)})$
where \( S \) is a rooted, unordered tree
and where every node $b$ 
is assigned a vertex subset $S_b \subseteq V(G)$, called a \emph{bag}, such that the following three conditions hold:
(i) \( \bigcup_{b \in V(S)} S_b = V(G) \), in other words, every vertex of \( G \) is in at least one bag;
(ii) For every $ (u,v) \in E$, there exists a node~$b$ of~$S$ such that the bag $S_b$ contains both $u$ and $v$;
(iii) For every $ u \in V$, the set $ T_u = \{ b \in V(S) : u \in S_b \} $ induces a connected subtree of~$S$.
The \textit{width} of \( \mathcal{S} = (S, \{S_b\}_{b \in V(S)}) \) is \( \max_{b \in V(S)} |S_b| - 1 \), that is, the maximum size of its bags minus 1. The \textit{treewidth} of a graph \( G \), denoted by \( \treewidth(G) \), is the minimum possible width of a tree decomposition of \( G \).
Likewise, we define the treewidth of  a probabilistic graph $\G$ to mean the treewidth of its underlying non-probabilistic graph.

We will also use \emph{path decompositions} and \emph{pathwidth} in some
hardness results. A \emph{path decomposition} is a tree decomposition whose tree
is a (rooted) path, and the \emph{pathwidth} of a graph~$G$ (resp., a
probabilistic graph $\G$) is the minimum
possible width of a path decomposition of~$G$ (resp., of the underlying
non-probabilistic graph of~$\G$). We will only use pathwidth to show hardness
results (when parameterising by it, or on pathwidth-1 or pathwidth-2 graphs),
which allows for crisper results than the corresponding results for treewidth
given that pathwidth is always no smaller than treewidth.

By contrast, all our tractability proofs will use treewidth, and for them it will be convenient 
to use so-called
\emph{nice tree
decompositions}~\cite{PA}, which
satisfy some additional properties. 
Specifically, every edge of $E$ is \emph{introduced} in
exactly one bag of the decomposition. Further, beyond introduction nodes
(which introduce a vertex or an edge), forget nodes (that remove a vertex), and leaf nodes,
other nodes $b$ 
have exactly two children, the bags
of~$b$ and its children are identical, and $b$ does not introduce any edge (we call such a node a \emph{join node}). We
define these more formally in Appendix~\ref{apx:nice} (but the precise
definition should not be needed to understand the main text of our article).

On such tree decompositions, we will use some additional notation.
Given a graph $G=(V,E)$ and a nice tree decomposition $\calS = (S,
\{S_b\}_{b \in V(S)})$ of~$G$, for $b$ a node in the tree~$S$, we
define $V_b$ as the set of vertices included in any of the bags
contained in the subtree rooted at $b$. We further define $G_b$ as
the graph whose vertices are $V_b$ and whose edges are exactly
those that are introduced in some bag of the subtree rooted at $b$. 
In particular, if $b$ is a join node with children $b_1,b_2$, then
$G_b$ is the disjoint union (in terms of edges) of $G_{b_1}$ and $G_{b_2}$.

\begin{toappendix}
  \label{apx:nice}
  In this section we describe what are nice tree decompositions.
  Intuitively, these are tree decompositions in which 
  vertices and edges are introduced one-by-one, to simplify the design of
  dynamic programming algorithms. 
  A \textit{nice tree decomposition} of width $k$ of a graph
  $G=(V,E)$ is a tree decomposition $\mathcal{S} = (S, \{S_b\}_{b
  \in V(S)})$ of width $k$ of $G$ where each node $b \in S$ is of
  one of the following five types:

\begin{itemize}

    \item \LeafNode{} node: $b$ has no children and $S_b=\emptyset$.
    
    \item \IntroVertNode{} node: $b$ has exactly one child $b'$ and $S_b = S_{b'} \cup \{v\}$ for some $v \in V$.
    
    \item \JoinNode{} node: $b$ has exactly two children ${b_1}$ and ${b_2}$ and $S_b = S_{b_1} = S_{b_2}$.

    \item \ForgetVertNode{} node: $b$ has exactly one child $b'$ and $S_b = S_{b'} \setminus \{v\}$ for some $v \in V$.
    
    \item \IntroEdgeNode{} node: $b$ has exactly one child $b'$ such that $S_b = S_{b'}$, and $b$ has two distinguished vertices $u,v \in S_b$ such that $(u,v) \in E$, demarcating the addition of an edge. We say that the edge $(u,v) \in E$ is \emph{introduced at node~$b$}. We additionally require that for every edge $(u,v)\in E$, the edge is introduced at exactly one bag.
\end{itemize}

Given a graph $G=(V,E)$ and a nice tree decomposition $\calS = (S, \{S_b\}_{b \in V(S)})$ of $G$, for any node $b$
in the tree~$S$, we define $V_b$ as the set of vertices included in any of the bags contained in the subtree rooted at $b$. We further define $G_b$ as the graph whose vertices are $V_b$ and whose edges are exactly those that are introduced in some bag of the subtree rooted at $b$. Formally, $G_b = (V_b, E_b) \text{ where } E_b = \{ e \mid e \text{ is introduced in the subtree rooted at } b \}$.
In particular, for $b_{\text{root}}$ the root bag,
we have $V_{b_{\text{root}}} = V$ and $G_{b_{\text{root}}} = G$.
\end{toappendix}

It follows from standard techniques \cite{bodlaender1988dp,PA} that we can
compute nice tree decompositions in FPT-linear time and further ensure that the
root bag precisely contains the source and target vertices.
Formally, we will always use the following implicitly in our proofs:

\begin{toappendix}
  We will use the following result, from~\cite{bodlaender1988dp}:
\end{toappendix}

\begin{lemmarep}[\cite{bodlaender1988dp}]
  \label{lem:computeNiceTreeDec}
  There is an algorithm and a polynomial $p$ which, given as input a graph $G = (V, E)$ with two distinguished vertices $s, t \in V$
  and an integer $k$, determines if $\treewidth(G) \leq k$, and if yes, computes a nice tree decomposition of $G$ of width at most  $\treewidth(G)+1$
  whose root bag is $\{s, t\}$,
  in time $\mathcal{O}(2^{p(k)}|V|)$.
\end{lemmarep}    

\begin{toappendix}
  Its proof intuitively follows from combining two lemmas: the fact
  from~\cite{bodlaender1988dp} that tree decompositions can be computed in
  FPT-linear time, and the fact from~\cite{PA} that tree decompositions can be
  made nice:

\begin{lemma}[\cite{bodlaender1988dp}]
\label{lem:treewidth}
  There is an algorithm and a polynomial $p$ such that, given as input a graph $G = (V, E)$ and integer $k$, determines if $\treewidth(G) \leq k$, and if yes, computes a tree decomposition of $G$ of width $\treewidth(G)$ in time $\mathcal{O}(2^{p(k)}|V|)$.
\end{lemma}    
\begin{lemma}[\cite{PA}]
\label{lem:make-nice}
  If a graph $G = (V,E)$ admits a tree decomposition of width $k$, then it also admits a nice tree decomposition of width $k$. Moreover, given a tree decomposition $S$ of~$G$ of width $k$, one can in time $\mathcal{O}(k^2 \cdot \max(|S|, |V|))$ compute a nice tree decomposition of~$G$ of width $k$ that has at most $\mathcal{O}(k|V|)$ nodes in $\mathcal{O}(k^2 \cdot \max(|S|, |V|))$ time.
\end{lemma}

And we will also use the following lemma to ensure that the root bag is exactly
$\{s,t\}$, which is also convenient for our algorithms.

\begin{lemma}
\label{lem:nice-st}
  Consider a graph $G=(V,E)$, a nice tree decomposition \( \mathcal{S} = (S, \{S_b\}_{b \in V(S)}) \) of $G$ and two vertices $s,t \in V$. Then we can compute in linear time in $\calS$ another nice tree decomposition $\mathcal{S'}$ where the root bag is exactly $\{s,t\}$ and where the width of $\mathcal{S'}$ is at most one more than the width of~$\mathcal{S}$.
\end{lemma}
\begin{proof}
  We add $t$ to every bag of the tree decomposition. Now there must be at least some node $b$ with a bag $S_b$ containing both $s$ and $t$
  and potentially some other vertices. If $S_b =\{s,t\}$ then we make $b$ the root and we are done.
  Else, create another \ForgetVertNode{} node $b'$ %
  and join it to $b$. Thus, $S_b$ and $S_{b'}$ differ only by a single vertex $v \in S_b$ such that $v \neq s,t$.
  Repeat the process to form a linear chain of nodes until we arrive at a node
  with the corresponding bag containing only $s,t$ and set it to be the root.
  Note that due to the addition of $t$, the bags corresponding to the
  \LeafNode{} nodes may no longer be empty.
  Hence, to preserve the nice tree decomposition property, for every \LeafNode{} node $b_1$, 
  we have $S_{b_1} = \{t\}$ after the addition of $t$, 
  so we create a new node
  $b_2$ with the corresponding bag $S_{b_2} = \emptyset$ and attach it to $b_1$. Thus $b_2$ is now a \LeafNode{} node and $b_1$ becomes an $\IntroVertNode{}$ node.
By following the process mentioned above we obtain another tree decomposition
which is nice, contains only $s$ and $t$ at the root, and has width at most one
more than that of~$S$. It is easy to see that this procedure takes linear time.
\end{proof}

We can now prove Lemma~\ref{lem:computeNiceTreeDec}:

\begin{proof}[Proof of Lemma~\ref{lem:computeNiceTreeDec}]
Given a graph $G=(V,E)$ one may apply \cref{lem:treewidth} to verify if
$\treewidth(G) \leq k$. Then, if indeed $\treewidth(G) \leq k$, we compute a tree decomposition $\mathcal{S}_{1}$ of width $k$. This step takes $\mathcal{O}(2^{p(k)}|V|)$ time, and so $\mathcal{S}_1$ has size bounded by the same expression. Further using Lemma~\ref{lem:make-nice}, we may obtain from this tree decomposition a nice tree decomposition  $\mathcal{S}_{2}$ of width $k$ that has at most $\mathcal{O}(k|V|)$ nodes. Finally using \cref{lem:nice-st} we may compute, in linear time, another nice tree decomposition $\mathcal{S}$ where the root bag is exactly $\{s,t\}$ and where the width of $\mathcal{S}$ is at most one more than the width of~$\mathcal{S}_2$. Overall, up to changing the polynomial $p$ if necessary, this takes $\mathcal{O}(2^{p(k)}|V|)$ time, establishing the result.
\end{proof}
\end{toappendix}

\myparagraph{Sketching the proof of \cref{thm:nonprobatw}}
We now review the high-level ideas to prove \cref{thm:nonprobatw},
as some of our later results will generalise these techniques.
We first notice that the problem $\prob{\Path}{\geq}$ reduces to
$\prob{\Path}{=}$ in polynomial time (using the fact that a simple path cannot
be longer than $|V(G)|$). Hence, we only need to show that 
$\prob{\Path}{=}$ is FPT parameterised by treewidth.

To do this, we adapt the so-called naïve approach described
in~\cite[Section 2.2]{ziobro2019finding} that is used to find Hamiltonian paths.
While this seems to be folklore, we could not find a reference that
explains how to do this for simple paths with a length constraint and specified endpoints.
The closest we found is \cite[Theorem 3.20]{bodlaender2015deterministic}, but there the endpoints are not specified and the
techniques used are much more evolved (i.e., this is not the “naïve approach”).
At a high level, the idea is to solve
the problem by dynamic programming, processing the tree
decomposition bottom-up to determine, for each bag $b$, what are
the “partial solutions” of the problem in the subgraph $G_b$. We
need a few definitions to precisely explain what partial solutions are, which we
sketch here and will formally redefine and extend for algorithms on probabilistic graphs later.
In these definitions, the set of vertices $A$ will intuitively be a bag of the tree decomposition.

\begin{definition}
 For $G=(V,E)$ a graph and $\mathcal{P}$ a set of paths in $G$,
 we call $\mathcal{P}$ \emph{vertex-disjoint} if for any two
 distinct paths $P_1,P_2 \in \mathcal{P}$, no vertex of $P_1$
 occurs in $P_2$. The \emph{length} of $\mathcal{P}$, denoted $\length(\calP)$, is then the
 sum of lengths of the paths that $\calP$ contains. For $A
 \subseteq V$, a set of paths $\mathcal{P}$ is
 \emph{$A$-vertex-disjoint} if it is vertex-disjoint, and for each
 path~$P \in \mathcal{P}$, the endpoints of $P$ are vertices 
 in~$A$.
\end{definition}

\begin{definition}
 Let $G=(V,E)$ be a graph, $A\subseteq V$ a set of vertices of $G$,
 and $P$ a path in $G$ whose endpoints are vertices of $A$. We
 define the \emph{projection of~$P$ relative to~$G$ and~$A$},
 written $\proj_{G,A}(P)$,
 to be the tuple comprising those
 vertices of $P$ that are in $A$, in the order in which they occur in $P$.
  Note that the tuple has at least one element, as it contains at least the
  endpoints of the path (and there is only one endpoint in case $|P|=0$, i.e.,
  $P$ is a single-vertex path), but
  it may have more if some internal vertices of $P$ are in $A$.

 We extend this notion to define the \emph{projection} of an
 $A$-vertex-disjoint set of paths $\mathcal{P}$ of~$G$ as
 $\proj_{G,A}(\calP) = \{\proj_{G,A}(P) \mid P \in \calP\}$; we call
 such a set a \emph{path set projection}.
 Finally, we define the \emph{disjoint path set projection of $A$
 relative to $G$}, denoted by $\psp_G(A)$, to be the set of the path set projections of all
 $A$-vertex-disjoint sets of paths of~$G$, formally,
 $\psp_G(A) \coloneq \{\proj_{G,A}(\mathcal{P}) \mid \mathcal{P} \text{ is an
 $A$-vertex-disjoint set of paths of $G$}\}$. 
\end{definition}

The algorithm that we propose to solve $\lppe$ over non-probabilistic
bounded-treewidth graphs is then as follows. We process the tree
decomposition bottom-up and, for every node~$b$ of the
decomposition, we compute the following information: (1.) the set $\psp_{G_b}(S_b)$,
and (2.) for every
path set projection $\Pi \in \psp_{G_b}(S_b)$, the
set of natural numbers 
\[\{\length(\calP) \mid \calP \text{ is an
$S_b$-vertex-disjoint set of paths of $G_b$ such that
$\proj_{G_b,S_b}(\calP) = \Pi$}\}.\]

We briefly sketch how this information can be computed bottom-up.
If $b$ is a 
leaf node the computation is easy, and if $b$ is a 
join node we can
compute the information from the two children $b_1,b_2$ of~$b$.
This is because, by the properties of nice tree
decompositions, when $\calP$ is an $S_b$-vertex-disjoint set of paths
of $G_b$,
$\calP$ can be decomposed as a concatenation of some $S_{b_1}$-vertex-disjoint sets of
paths $\calP_1$ of $G_{b_1}$ and some $S_{b_2}$-vertex-disjoint sets of
paths $\calP_2$ of~$G_{b_2}$. 
Furthermore, since
$\psp_{G_b}(S_b)$ is a set of sets of tuples of vertices of $S_b$,
its size 
depends only on~$k$, so the total size of
the information that we compute for every node $b$ of the tree
decomposition is bounded\footnote{As the reader might have
understood, our focus in this paper is not to get the best $f(k)$ possible.} by $n \times f(k)$.
When we reach the root bag $b_{\text{root}}$,
with $S_{b_{\text{root}}} = \{s,t\}$ and $G_{b_{\text{root}}} = G$, we have the information we need to
decide if there exists a simple path of length $\ell$ in $G$ from
$s$ to $t$,
namely, consider the path set projection
$\calP = \{(s,t)\}$ (or $\{(s)\}$ if $s=t$) and see whether the length $\ell$ was achieved for $\calP$.
All in all, this gives us an algorithm with complexity $f(k) \times
n^{\calO(1)}$, i.e., FPT parameterised by treewidth, as promised.

%% file: proba.tex
We now move to the study of the $\probname$ problems on probabilistic graphs, presenting in this section their complexity in two extreme regimes. 

First, we show that the $\probname$ problems are all intractable on probabilistic graphs of unbounded treewidth:
this simply follows from the fact that source-to-target reachability (ST-CON) is
\#P-hard~\cite{V79}, and we recall a standard self-contained 
reduction from the \#PP2DNF problem for the reader's convenience and to
introduce the techniques used in our subsequent hardness results.  %
Second, we show that the $\probname$ problems are all tractable on probabilistic graphs of treewidth 1, i.e., directed graphs whose underlying undirected graph is a forest: this follows from a somewhat tedious dynamic programming computation.
We will then show in the next two sections 
that, unexpectedly, this tractability result generalises to higher treewidth
bounds only for some of the $\probname$ problems (\cref{sec:tract}), and that hardness can be shown for the others (\cref{sec:untract}).

\subparagraph*{Intractability on unbounded-treewidth probabilistic graphs.}
We first show that these problems are all intractable when we allow arbitrary graphs (of unbounded treewidth) as input, which justifies that we only study them on bounded-treewidth graphs afterwards:

\begin{theoremrep}
\label{thm:bla}
On probabilistic graphs, the problems $\probname^{X}_{\bowtie}$ for ${\bowtie} \in \{{\leq}, {=},
{\geq}\}$ and $X \in \{\Walk, \Path\}$ are all \#P-hard, even
when the target length is given in unary.
\end{theoremrep}

To show Theorem~\ref{thm:bla}, we reduce from the classical problem $\spptdnf$:

\begin{definition}
  \label{def:pp2dnf}
  Let $X = \{X_1,\ldots, X_{n_1}\}$ and $Y =
\{Y_1,\ldots,Y_{n_2}\}$ be two disjoint sets of variables.
  A \emph{positive, partitioned, 2DNF (PP2DNF) over $X,Y$} is a Boolean formula of the form $\varphi
= \bigvee_{r=1}^m X_{i_r} \land Y_{j_r}$, where $i_r\in
\{1,\ldots,n_1\}$ and $j_r \in \{1,\ldots,n_2\}$ for $1\leq r \leq
m$. 
An \emph{assignment} for $\varphi$ is a function $\nu \colon X \cup Y \to
\{0,1\}$; we call $\nu$ \emph{satisfying} if there is some $1 \leq r \leq m$ for
which we have $\nu(X_{i_r}) = 1$ and $\nu(Y_{j_r}) = 1$.
The problem $\spptdnf$ takes as input such a $\pptdnf$ $\varphi$
and outputs $\#\varphi$, the number of satisfying assignments of~$\varphi$.
\end{definition}

The \#PP2DNF problem is known to be \#P-hard~\cite{provan1983complexity}. We can now show \cref{thm:bla}:

\begin{proof}[Proof sketch of \cref{thm:bla}] 
  We use a standard construction, used for instance to show the hardness of
  some conjunctive queries on probabilistic
  databases~\cite[Theorem~3.2]{suciu2011probabilistic}.
  We code the PP2DNF formula as a
  bipartite graph, with variables $X$ and $Y$ as nodes respectively in the first and second partition, and connecting an
  $X$-variable to a $Y$-variable by an edge of probability $1$ if
  they occur together in a term. We then add a source vertex $s$
  connected to all variables $X$ by edges of probability
  $\frac{1}{2}$, and likewise a target vertex $t$ probabilistically
  connected to all variables $Y$. The probability that there exists
  a walk or path from $s$ to $t$ having length at most, at least, or exactly $3$ is
  then easily seen to be equal to
  $\frac{\#\varphi}{2^{|\mathrm{Vars}(\varphi)|}}$.
\end{proof}

\begin{toappendix}
\begin{proof}
  We first prove this result for the problem $\lppe$.
Let $\varphi$ be a PP2DNF formula as in Definition~\ref{def:pp2dnf}.
Define for
convenience $n = n_1 + n_2$. 
Define $\G_\varphi$ to be the probabilistic graph having as set of
vertices $X \cup Y \cup \{s,t\}$ (with~$s$ and~$t$ being fresh
nodes) and edges $\{(s,X_i) \mid 1\leq i \leq n_1\} \cup \{(Y_i,t)
\mid 1\leq i \leq n_2\} \cup \{(X_{i_r}, Y_{j_r}) \mid 1 \leq r
\leq m\}$, with all edges of the form $(s,X_i)$ and $(Y_i,t)$
having probability $\nicefrac{1}{2}$ and all other edges having
probability $1$. 
 For instance, letting $\varphi = (X_1 \land Y_1)\lor (X_1 \land Y_2) \lor (X_2
  \land Y_2)$, we show $\G_{\varphi}$ in~\cref{fig:exmpl}.
  
Now, let~$q$ be the probability that there exists a
  path from $s$ to $t$ in~$\G_\varphi$; notice that we have
  $q = \frac{\#\varphi}{2^n}$, so that it is enough to argue that we can
  compute $q$.

\begin{figure}
\centering
  \begin{tikzpicture}
\begin{scope}[every node/.style={circle,thick,draw}]
    \node (s) at (0,0) {$s$};
    \node (X1) at (2,1) {$X_1$};
    \node (X2) at (2,-1) {$X_2$};
    \node (Y1) at (4,1) {$Y_1$};
    \node (Y2) at (4,-1) {$Y_2$};
    \node (u) at (6,0) {$t$};
\end{scope}

\begin{scope}[>={Stealth[black]},
              every node/.style={fill=white,circle},
              every edge/.style={draw,thick}]
              \path [->] (s) edge node {$\frac{1}{2}$} (X1);
              \path [->] (s) edge node {$\frac{1}{2}$} (X2);
              \path [->] (X1) edge node {$1$} (Y1);
              \path [->] (X1) edge node {$1$} (Y2);
              \path [->] (X2) edge node {$1$} (Y2);
              \path [->] (Y1) edge node {$\frac{1}{2}$} (u);
              \path [->] (Y2) edge node {$\frac{1}{2}$} (u);
\end{scope}
\end{tikzpicture}
\caption{The probabilistic graph $\calG_{\varphi}$ from the proof of \cref{thm:bla}.
Probabilities are indicated in the middle of edges.}
  \label{fig:exmpl}
\end{figure}

 Let us call our oracle to $\lppe $ to determine if there exists a path of length $3$ from $s$ to $t$. Since all paths from $s$ to $t$ have length $3$, if the output of the oracle is $q$, we may return $q \cdot 2^{n}$ as the number of satisfying solutions to the $\spptdnf$ formula.
 This concludes the proof in the case of $\lppe$.

Now, observe that since all paths from $s$ to $t$ have length $3$ for the graph
constructed in our reduction, the probability that a path/walk of length
exactly $3$ exists is equivalent to the probability that a path/walk of length
at least/at most $3$ exists, from $s$ to $t$. Hence the same reduction shows
that $\lppe, \lppl,\lppg,\lpwe,\lpwl,\lpwg $ are all \#P-hard.
\end{proof}
\end{toappendix}

\subparagraph*{Tractability on treewidth-1 probabilistic graphs.}
As a first step towards studying bounded-treewidth probabilistic graphs, we now
show that all six $\probname$ problems are tractable on treewidth-1
probabilistic graphs (i.e., probabilistic forests). We will be able to generalise this result to higher constant treewidth bounds in \cref{sec:tract} for some problems only; by contrast we will show in \cref{sec:untract} that for the other problems they do not generalise:

\begin{theoremrep}
  \label{thm:forest}
For probabilistic graphs of treewidth $1$, the problems
$\probname^{X}_{\bowtie}$ for ${\bowtie} \in \{{\leq}, {=}, {\geq}\}$ and $X
\in \{\Walk, \Path\}$ can all be solved in linear time, up to the cost of
arithmetic operations.
\end{theoremrep}
\begin{proofsketch}
  As $\calG$ has treewidth one, there is at most one directed path
  from $s$ to $t$. If no such path exists then the answer is zero,
  otherwise call the path $P$. The problems $\lppe$, $\lppl$
  (equivalently $\lpwl$), and $\lppg$ are then easy to solve, by
  first determining the length~$m$ of~$P$, then comparing $m$ to
  the target input length $\ell$, and returning, depending on this
  comparison, either~$0$, or the probability that $P$ exists (i.e.,
  the product of the probabilities of its edges). For
  $\lpwe$, if $m>\ell$ we can return~$0$ and if $m = \ell$ we can
  again return the probability that $P$ exists. 
  Otherwise, we need that all edges of~$P$ exist, but we must also take into account 
  the possibility to have, in a realisation of $\calG$, a
  path from $s$ to $t$ which is not simple and achieves the requisite length. If $m$ and $\ell$ do not have the same parity, this can only be achieved thanks to a self-loop, either on a vertex of~$P$ or on some other vertex reached by a back-and-forth excursion. Whenever this is possible, then we can get arbitrarily long walks, but we must consider whether the excursion is doable without exceeding the target length. If $m$ and $\ell$ have the same parity, it is additionally possible to lengthen the walk by multiples of two by doing a back-and-forth either on an edge dangling from the path~$P$ or using some return edge for an edge of~$P$. These probabilities can all be computed in linear time by dynamic programming if we define the events carefully enough to ensure that they are independent.
The case of $\lpwg$ is similar.
\end{proofsketch}

\begin{toappendix}

\begin{algorithm}
\SetAlgoLined
\KwIn{A graph $\G=(V,E, \pi \colon E \to [0,1])$, vertices $s,t$, and an integer $\ell$ representing the desired walk length }
\KwOut{Probability $p$ that a walk of length $\ell$ exists from $s$ to $t$ in $\G$}

Let $P = (n_0 = s, n_1, \dots, n_k = t)$ be the unique path from $s$ to $t$ (if no such path exists, return $0$ as the probability)\;

\If{$\ell < k$}{
    \Return{$0$}\;
}
\If{$\ell = k$}{
    $p \gets$ $\Bigl(\prod_{i = 0}^{k-1}  \pi(n_i,n_{i+1})\Bigr)$\;
    \Return{$p$}\;
}

    \If{$(\ell - k) \bmod 2 = 0$}{
    $Q = \{u_1,u_2,u_3 \dots u_j\} \gets \text{set of all vertices} \in P \text{
    which have a self-loop in } \mathcal{G}$\;
    \tcp{$p_Q$ is the probability that there no vertex of~$P$ has a self-loop}
    $p_Q \gets \prod_{u \in Q}\Bigl(1-\pi(u,u)\Bigr)$\;
    $R = \{(v_1,w_1),(v_2,w_2) \dots (v_i,w_i)\} \gets \text{set of all
    vertex pairs}\; 
            (v,w) \text{ with } v \in P \text{ and } w \notin P \text{ and
            the edges } (v,w) \text{ and } (w,v) \text{ exist in } \mathcal{G}$\;
    \tcp{$p_R$ is the probability that there is no back-and-forth edge pair
    outside of~$P$ incident to a vertex of~$P$}
    $p_R \gets \prod_{(v,w) \in R}\Bigl(1-\pi(v,w)\pi(w,v)\Bigr)$\;
    $S = \{(v_1,w_1),(v_2,w_2) \dots (v_j,w_j)\} \gets \text{set of all
    vertex pairs}\; 
            (v,w) \text{ with } v,w \in P  \text{ and the edge } (v,w)
            \text{ occurs in } P \text{ and the back edge } (w,v) \text{ exist in }
            \mathcal{G}$\;
    \tcp{$p_S$ is the probability that there is no return edge for any forward
    edge of~$P$}
    $p_S \gets \prod_{(v,w) \in S}\Bigl(1-\pi(w,v)\Bigr)$\;
    \tcp{return probability that P exists and we have one self-loop on a
      vertex of P or back-and-forth on P or on an edge incident to P}
    \Return{$\Bigl(\prod_{i = 0}^{k-1}  \pi(n_i,n_{i+1}) \Bigr)
    \Bigl(1-p_{Q}p_{R}p_{S}\Bigr)$}\; 
    }\Else{
        \tcp{critical length to reach a self-loop and change parity}
    $\mathsf{cl} \gets \frac{\ell - k - 1}{2}$\;

    \ForEach{$v \in P$}{
        \tcp{explore the tree rooted at~$v$ without the edges of~$P$}
        $E_P \gets \text{edges of~$P$}$\;
        $\mathcal{T}_{v} \gets \text{tree rooted at $v$ of the connected
        component of~$v$ in~$\mathcal{G}[E \setminus E_P]$} $\;
        \tcp{compute the probability that we can reach a self-loop from~$v$}
        $p_v \gets \texttt{GetExProb}{(\mathcal{T}_v, v, \mathsf{cl})}$\;
    }
    \tcp {return probability that P exists and there is a vertex of P from which
    we can reach a self-loop}
    \Return{$\left(\prod_{i = 0}^{k-1}  \pi(n_i,n_{i+1}) \right) \times
    \left(1 - \prod_{v \in P} \left(1 - p_v\right)\right)$ }\;
}

\caption{Compute probability that a probabilistic treewidth-$1$ graph contains
an $st$-walk of length~$\ell$}\label{alg:twg1eq}
\end{algorithm}

\begin{proof}

Note that if a graph $\mathcal{G}=(V,E,\pi)$ has treewidth $1$, then it implies that either there is a unique directed path from $s$ to $t$ or $t$ is not reachable from $s$. Indeed, otherwise the underlying undirected graph would contain a cycle, implying that its treewidth is greater than or equal to~$2$.   
Note that if it exists, this unique directed path may be found using a breadth first search in linear time in the size of the input. Thus given an instance of $\lppe $ where we need to find the probability of existence of a path of length $\ell$, we can check if the length of the unique directed path from $s$ to $t$ is $\ell$. If there is no $s$ to $t$ path or if the length of the unique directed path from $s$ to $t$ does not equal $\ell$ we return $0$, else we return the multiplication of the probabilities over all the edges in this path. For $\lppl$ (or equivalently $\lpwl $) where the desired path length threshold is $\ell$, we again check if a path from $s$ to $t$ exists, check whether its length is less than or equal to $\ell$ if it exists, and return $0$ or the multiplication of the probabilities over all the edges in this path accordingly.
A similar argument works for $\lppg$.

The only remaining cases are $\lpwg$ and $\lpwe $. We propose \Cref{alg:twg1eq} for $\lpwe$, which may be adapted to $\lpwg$  as well. 
Throughout the algorithm, the probability $\pi(u,v)$ is always implicitly~$0$ if there is no edge from~$u$ to~$v$.
We give more details here on why the algorithm works.

Let $P$ be the unique directed path from $s$ to $t$. (If no such path exists, we simply return $0$). If the length of $P$ is $k$, and the desired walk length threshold is $\ell$ then the case where $\ell < k$ is easy since there cannot exist a walk from $s$ to $t$ with length less than that of the unique directed path $P$. Likewise when $\ell = k$, the only possible walk from $s$ to $t$ with length $\ell$ is $P$ itself, hence in this case we return the multiplication of the probabilities associated with the edges of $P$. 

We now turn to the case where $\ell > k$. We make two cases on the basis of the parity of $(\ell -k)$. Note that if $(\ell -k)\mod2 =0 $, then a realisation $G$ of $\G$ contains a walk of length $\ell$ from $s$ to $t$ if and only if at least one of these three kinds of edges is present in $G$, in addition to the directed path~$P$.

\begin{itemize}
 \item{There exists a vertex $v \in P$ such that there exists a self-loop at the vertex $v$ in $G$ (i.e., an edge from $v$ to $v$)} 
 \item{There exist two adjacent vertices $v_1,v_2 \in P$ such that there is an edge from $v_1$ to $v_2$ and an edge from $v_2$ to $v_1$} in $G$
 \item{There exist two vertices $v_1,v_2$ such that $v_1 \in P, v_2 \notin P$ and there exists an edge from $v_1$ to $v_2$ and an edge from $v_2$ to $v_1$} in $G$
\end{itemize}

Therefore if $(\ell -k)\mod2 =0 $, then the required probability of a walk of length $\ell$ from $s$ to~$t$ being present in $\G$ is the probability that the path $P$ is present, and at least one of these three conditions outlined above is satisfied. Note that up to the parity of $(\ell-k)$, these conditions are independent of the specific value of $\ell$.

Let us now look at the case when $(\ell-k) \mod 2 =1 $. In such a case, we can
show that given a realisation $G$ of $\G$ there exists a walk of length $\ell$
from $s$ to $t$ in $G$ if and only if there exists a vertex $v \in P$  such that a \emph{retraceable} path of length at most $\frac{\ell-k-1}{2}$ exists from $v$ to another vertex $v'$, possibly not in $P$, such that $v'$ contains a self-loop, i.e., there is an edge from $v'$ to $v'$. By a \emph{retraceable} path we mean a path in $G$, where besides every forward edge $(p,q)$ in the path $P$, the corresponding back edge $(q,p)$ must also be present in $G$. The idea is that we can take the path $P$ from $s$ to $v$, diverge from $v$ to $v'$, take the self-loop attached to $v'$ sufficiently many times, retrace back to the vertex $v$ and then continue along the path $P$ from $v$ to $t$. However, this self-loop cannot be located too far from $v$, in particular, if it is at a distance of greater than $\frac{\ell-k-1}{2}$ from $v$, then visiting $v'$ from $v$ and returning back to $v$ will necessarily cause the length of walk from $s$ to $t$ to exceed $\ell$. So in this case, the natural algorithm is to compute, for every vertex $v$ on the path $P$, the probability that there is a path of length less than or equal to the maximum length $\texttt{cl}$ from $v$ to another vertex with self-loop, where $\texttt{cl}=\frac{\ell-k-1}{2}$. 

We compute this quantity using \cref{alg:recurse}. Intuitively, at every node $w$ reachable from $v$ without going through an edge of~$P$ and without exceeding the maximum length, letting $\pi(w,w)$ be the probability that the current node $w$ has a self-loop, then the probability that we can reach a self-loop from~$w$ and come back within our allotted maximum length can be expressed as follows.
It is $1$ if the self-loop on $w$ is present, which happens with probability $\pi(w,w)$. Otherwise, it is the independent OR of the probabilities that we can reach a self-loop from~$w'$ for each child $w'$ and come back, weighted by the probabilities that the edges to~$w'$ and back are present, and multiplying by the recursively computed quantity on~$w'$ with the maximum length having been decremented. Note that, here, the recursively computed probabilistic events are indeed independent because they pertain to disjoint parts of the tree. The same holds between the probabilistic events across the various choices of~$v$ along the path, thanks to the fact that the underlying probabilistic graph is a forest.

Now, we briefly sketch how this algorithm can be adapted to $\lpwg$. Again, we find the unique directed path $P$ from $s$ to $t$ (if it exists, else we return $0$). Let its length be $k$. If the desired length $\ell$ is such that $\ell \leq k$, then clearly such a walk exists if and only if the directed path $P$ exists.

Otherwise, if $\ell > k$, 
we take the independent disjunction across all possible ways to make the walk
longer than~$P$. (We need not care about parity this time, because we do not need to reach
a precise length.) More precisely, we take the independent disjunction across
all self-loops of $P$, all back-edges corresponding to a forward edge of~$P$,
and all pairs of back-and-forth edges not in~$P$ that lead out of a vertex
of~$P$; and we multiply this by the probability that all edges of~$P$ exist (i.e., we do precisely like the case of $(\ell - k) \bmod 2 = 0$ in \cref{alg:twg1eq}).
\end{proof}

\begin{algorithm}
\SetAlgoLined
\SetKwFunction{ComputeT}{GetExProb}
\KwIn{directed rooted tree $\mathcal{T} =(V_{\mathcal{T}},E_{\mathcal{T}},\pi\colon E_{\mathcal{T}} \to [0,1])$, current node $v$, length budget $\mathsf{cl}$}
\KwOut{Probability that we can reach a self-loop within budget $\mathsf{cl}$ and come back}

\SetKwProg{Fn}{Function}{:}{}
\Fn{\ComputeT{$\mathcal{T}, v, \mathsf{cl}$}}{
    \If{$\mathsf{cl} < 0$}{
        \Return{$0$}\;
    }

    \tcp{$\pi(v,v)$ is implicitly $0$ if $(v,v) \notin E$}
    \Return{
        $\pi(v,v)+(1-\pi(v,v))\times\Bigl(
        1 - \displaystyle\!\!\!\!\!\!\!\!\!\!\prod_{\substack{v' \text{child of } v \text{ in } \mathcal{T} \\ (v,v'),(v',v) \in E}}\!\!\!\!\!\!\!\!\!\!\!\!
        \Bigl(1-\pi(v,v') \cdot\pi(v',v)
    \cdot \text{\emph{\texttt{GetExProb}}}(\mathcal{T}, v', \mathsf{cl}-1)\Bigr)\Bigr)$}
}
\caption{Compute the probability that we can reach a self-loop with the  indicated budget}
\label{alg:recurse}
\end{algorithm}
\end{toappendix}

%% file: tract.tex
This section contains two results. The first one,
whose proof we will sketch in detail in the rest of the
section, is that when the input probabilistic graphs
have bounded treewidth, all $\probname$ problems
except the ``exact-length'' variants are tractable. 

\begin{theorem}
  \label{thm:tract}
    For any constant $k$, the problems $\probname^{X}_{\bowtie}$ for ${\bowtie}
    \in \{{\leq}, {\geq}\}$ and $X \in \{\Walk, \Path\}$ are tractable on
    probabilistic graphs of treewidth at most $k$ (i.e., XP parameterised by
    treewidth).
\end{theorem}

The second result is that, in terms of parameterised complexity, XP is likely
the best one can hope for, as we show that the problem variants from \cref{thm:tract} are all W[1]-hard,
hence not FPT 
assuming the standard hypothesis that FPT$\neq$W[1]. We show that this holds
already for acyclic graphs and already when parameterising by pathwidth:

\begin{toappendix}
  \subsection{Proofs of the W[1]-Hardness Result}
In this section we prove the following result.
\end{toappendix}

\begin{theoremrep}
  \label{thm:w1hard}
  On probabilistic acyclic graphs, the problems $\probname^{X}_{\bowtie}$ for
  ${\bowtie} \in \{{\leq}, {\geq}\}$ and $X \in \{\Walk, \Path\}$ are W[1]-hard
  (under Turing reductions) parameterised by pathwidth (hence also by
  treewidth).
\end{theoremrep}
\begin{proofsketch}
The reduction is done in two steps, starting from the problem of counting
  linear extensions of a partially ordered set (poset) $\calP$, that we write $\lee$,  which is W[1]-hard
  parameterised by pathwidth~\cite{eiben2019counting}.\footnote{The authors of
  \cite{eiben2019counting} show hardness when parameterised by treewidth, but
  in fact the authors of \cite{belmonte2022grundy} (cf.\
  \cite[Section~1.1]{belmonte2022grundy} and text before)
  notice that their proof also works for pathwidth.}
  The main idea of this proof was obtained with an LLM; we digested the proof and
  simplified its presentation. The technique is very reminiscent of the \#P-hardness proof of \cite{ando2026complexity} (in their different context with probabilistic continuous edge lengths). Their proof relies on the hardness of counting linear extensions of posets \emph{of height 2}, and an additional trick (similar to what is done in \cite[Section
  2]{Dittmer_Pak_2020}) allows us to reduce instead from $\lee$.

   We now sketch the reduction. We first reduce $\lee$ to
   the problem of counting the number of \emph{order-preserving maps of a poset
   into a chain}. For a poset $\calP=(A,\leq_\calP)$ and $q \geq 2$, a \emph{$q$-order-preserving map} is a function $f\colon A\to \{1,\ldots,q\}$ such that $f(x)\leq f(y)$ whenever $x\leq_\calP y$, and the \emph{order polynomial} of $\calP$ is the function that counts the number of such $q$-order-preserving maps. This first reduction relies on the known fact that the leading coefficient of the order polynomial of $\calP$ 
  is equal to its number of linear extensions, divided by $|\calP|!$, so we can recover it via polynomial interpolation, through a Turing reduction that uses polynomially many calls to the oracle. 

  Second, we reduce this latter problem to our problem, with a reduction that again uses polynomially many 
  calls to the target problem to recover the number of order-preserving maps. In
  fact, we use a variant of our problems where 
  we instead simply count the number of
  realisations that contain a given number of edges and that satisfy the path length constraint,
  which we show to be reducible to our
  problem under Turing reductions again via standard polynomial interpolation techniques
(see, e.g., \cite[Proof of Proposition 3.1]{deutch2022computing}, or
\cite[Appendix~A]{bienvenu2024shapley}; we prove it as \cref{lem:interpol}).

The second step of the reduction is illustrated in \cref{fig:constr}: we code the poset as a
bipartite graph where every covering relation $x \prec_\calP y$ of $\calP$ gives a
non-probabilistic
edge $(u_x, v_y)$ and further every vertex gives a non-probabilistic edge
$(u_x,v_x)$. We then have probabilistic edges that give paths of every possible
length from~$s$ to each $u_x$  and from each $v_x$ to~$t$. We can show that, on
realisations satisfying the path length requirements, the $(u_x,v_x)$ edges
force us to discard at least some number of probabilistic edges: our restriction to realisations that contain a given number of edges then enforces that we
make consistent choices so that we can assign each element of~$\calP$ to a position.
The $(u_x,v_y)$ edges then enforce that this assignment corresponds to an
order-preserving map.
\end{proofsketch}

  \begin{figure}
  \centering
  \begin{tikzpicture}[xscale=0.6,yscale=0.5,
    >=stealth,
    main/.style={circle, draw, fill=white, inner sep=2pt},
    aux/.style={circle, draw, fill=white, inner sep=1.5pt}
]

\node[main] (s)  at (0,0)   {$s$};

\node[main] (u1) at (6,5)   {$u_z$};
\node[main] (u2) at (6,0)   {$u_y$};
\node[main] (u3) at (6,-5)  {$u_x$};

\node[main] (v1) at (10,5)  {$v_z$};
\node[main] (v2) at (10,0)  {$v_y$};
\node[main] (v3) at (10,-5) {$v_x$};

\node[main] (t)  at (16,0)  {$t$};

\draw[->] (u1) -- (v1);
\draw[->] (u2) -- (v2);
\draw[->] (u3) -- (v3);
\draw[->] (u2) -- (v1);
\draw[->] (u3) -- (v1);

\draw[->, thick, dotted, line cap=round, draw=blue] (s) to[out=70,in=200] (u1);

\node[aux] (a1) at (3,2.5) {};
\draw[->, thick, dotted, line cap=round, draw=blue] (s) -- (a1);
\draw[->] (a1) -- (u1);

\node[aux] (a2) at (2.6,1.2) {};
\node[aux] (a3) at (4.5,2.8) {};
\draw[->,thick,dotted, line cap=round, draw=blue] (s) -- (a2);
\draw[->] (a2) -- (a3);
\draw[->] (a3) -- (u1);

\draw[->, thick, dotted, line cap=round, draw=blue] (s) to[out=20,in=160] (u2);

\node[aux] (b1) at (3,0) {};
\draw[->, thick, dotted, line cap=round, draw=blue] (s) -- (b1);
\draw[->] (b1) -- (u2);

\node[aux] (b2) at (2,-0.5) {};
\node[aux] (b3) at (4,-0.5) {};
\draw[->, thick, dotted, line cap=round, draw=blue] (s) -- (b2);
\draw[->] (b2) -- (b3);
\draw[->] (b3) -- (u2);

\draw[->, thick, dotted, line cap=round, draw=blue] (s) to[out=-20,in=110] (u3);

\node[aux] (c1) at (3,-2.5) {};
\draw[->, thick, dotted, line cap=round, draw=blue] (s) -- (c1);
\draw[->] (c1) -- (u3);

\node[aux] (c2) at (1.5,-3) {};
\node[aux] (c3) at (3,-4.0) {};
\draw[->, thick, dotted, line cap=round, draw=blue] (s) -- (c2);
\draw[->] (c2) -- (c3);
\draw[->] (c3) -- (u3);

\draw[->, thick, dotted, line cap=round, draw=blue] (v1) to[out=-20,in=110] (t);

\node[aux] (d1) at (13,2.5) {};
\draw[->, thick, dotted, line cap=round, draw=blue] (v1) -- (d1);
\draw[->] (d1) -- (t);

\node[aux] (d2) at (11.5,2.8) {};
\node[aux] (d3) at (13.4,1.2) {};
\draw[->, thick, dotted, line cap=round, draw=blue] (v1) -- (d2);
\draw[->] (d2) -- (d3);
\draw[->] (d3) -- (t);

\draw[->, thick, dotted, line cap=round, draw=blue] (v2) to[out=20,in=160] (t);

\node[aux] (e1) at (13,0) {};
\draw[->, thick, dotted, line cap=round, draw=blue] (v2) -- (e1);
\draw[->] (e1) -- (t);

\node[aux] (e2) at (12,-0.5) {};
\node[aux] (e3) at (14,-0.5) {};
\draw[->, thick, dotted, line cap=round, draw=blue] (v2) -- (e2);
\draw[->] (e2) -- (e3);
\draw[->] (e3) -- (t);

\draw[->, thick, dotted, line cap=round, draw=blue] (v3) to[out=70,in=200] (t);

\node[aux] (f1) at (13,-2.5) {};
\draw[->, thick, dotted, line cap=round, draw=blue] (v3) -- (f1);
\draw[->] (f1) -- (t);

\node[aux] (f2) at (13,-4) {};
\node[aux] (f3) at (14.5,-3) {};
\draw[->, thick, dotted, line cap=round, draw=blue] (v3) -- (f2);
\draw[->] (f2) -- (f3);
\draw[->] (f3) -- (t);

\end{tikzpicture}
\caption{Example construction of the graph $G$ from the proof of
\cref{prp:w1best}, for $q=4$ and for the poset $\calP$ containing three elements
$x$, $y$ and $z$, with $x\leq_\calP z$ and $y \leq_\calP z$ and $x$ and $y$
being incomparable. The dotted blue lines are the \emph{coding edges}, which are defined in the proof and intuitively correspond to probabilistic edges.}
\label{fig:constr}
\end{figure}

\begin{toappendix}
  To prove this, we will need the following general
  tool, intuitively saying that if we can efficiently
  compute the probability that a randomly chosen
  subset of edges of a probabilistic graph has a
  given property, then we can also efficiently
  compute the \emph{number} of such sets of a given \emph{cardinality}
  that have the property. Formally a \emph{property}
  in this context is a function that takes as input a
  graph $G=(V,E)$, a subset $S\subseteq E$ of the
  edges, and some parameters $\Pi$, and that is
  either true or false for these inputs. For
  instance, $P$ could take a graph $G=(V,E)$, a
  subset $S\subseteq E$, and as parameters a source $s\in V$ and target
  $t \in V$ and integer $\ell$, and is satisfied only
  if there exists a path of length $\leq \ell$ from
  $s$ to $t$ in $G[S]$. For such a property $P$,
  define the \emph{probabilistic computation problem
  for $P$} to be the problem where the input is a
  probabilistic graph $\calG=(G=(V,E),\pi)$ and some
  parameters $\Pi$, and the output is the sum,
  across all $S\subseteq E$ such that $(G,S,\Pi)$
  has the property $P$, of $\Pr_\calG(S)$. Define
  also the \emph{generalised cardinality computation problem for
  $P$} to take as input a graph $G=(V,E)$, a subset $T\subseteq E$, some
  parameters $\Pi$, and some natural number $c\in
  \mathbb{N}$, and the goal is to compute how many
  subsets $S\subseteq E\setminus T$ of size $c$ are such that
  $(G,T\cup S,\Pi)$ satisfies property $P$. 
  Intuitively, the edges of~$T$ are always present and not accounted for, i.e., they correspond to \emph{exogenous facts}.
  We then recall the following, which is
  folklore (see, e.g., \cite[Proof of Proposition 3.1]{deutch2022computing}, or
  \cite[Appendix~A]{bienvenu2024shapley}):

  \begin{lemma}
    \label{lem:interpol}
    For any property $P$ as above, there is a polynomial-time
    Turing reduction from the generalised cardinality problem for
    $P$ to the probabilistic computation problem for
    $P$.
  \end{lemma}
  \begin{proof}
Let $G=(V,E)$, $T\subseteq E$, $\Pi$ and $c$ the input instance for
the generalised cardinality problem for $P$, and let $n=|E\setminus T|$. For
$r\in \{0,\ldots,n\}$ define $N_r$ to be the number
of subsets $S\subseteq E\setminus T$ of size $r$ such that
$(G,T\cup S,\Pi)$ satisfies property $P$. We will show how to
compute all the $N_r$ values, which is enough since
we only need $N_c$.

For $z$ a natural number, define $\calG_z$ to be the
probabilistic graph $\calG_z=(V,E,\pi_z)$ where
$\pi_z((v,v'))=1$ for every $(v,v')
\in T$ and 
$\pi_z((v,v'))=\frac{z}{1+z}$ for every $(v,v')
\in E\setminus T$. Then, denoting $O_z$ the result of the probabilistic computation problem for $P$ on $\calG_z$ and~$\Pi$, we have:

\begin{align}
  O_z &= \sum_{\substack{S\subseteq E\\ (G,S,\Pi)\text{ has property $P$}}} \Pr_{\calG_z}(S) \nonumber\\
  &= \sum_{\substack{S\subseteq E\setminus T\\ (G,T\cup S,\Pi)\text{ has property $P$}}} \Pr_{\calG_z}(T\cup S) \nonumber\\
  &= \sum_{r=0}^{n} \sum_{\substack{S\subseteq E\setminus T, |S|=r\\ (G,T\cup S,\Pi)\text{ has property $P$}}} \Pr_{\calG_z}(T\cup S) \nonumber\\
                                             &= \sum_{r=0}^{n} \sum_{\substack{S\subseteq E\setminus T, |S|=r\\ (G,T\cup S,\Pi)\text{ has property $P$}}} \left(\prod_{e\in S} \frac{z}{1+z}\right)\left(\prod_{e\in E\setminus (T\cup S)} \left(1- \frac{z}{1+z}\right)\right) \nonumber\\
                                             &= \sum_{r=0}^{n} \sum_{\substack{S\subseteq E\setminus T, |S|=r\\ (G,T\cup S,\Pi)\text{ has property $P$}}} \left(\frac{z}{1+z}\right)^{r}\left(\frac{1}{1+z}\right)^{n-r} \nonumber\\
                                             &= \frac{1}{(1+z)^n} \sum_{r=0}^{n} z^r \sum_{\substack{S\subseteq E\setminus T, |S|=r\\ (G,T\cup S,\Pi)\text{ has property $P$}}} 1 \nonumber\\
                                             &= \frac{1}{(1+z)^n} \sum_{r=0}^{n} z^r N_r \label{eq:vdm}
\end{align}

We then use the oracle to the probabilistic computation problem for $P$ to compute values $(1+z)^n
O_z$ for $n+1$
pairwise different values of $z$. We then solve the system of
    linear equations given by~\cref{eq:vdm} (which
    is invertible because its matrix is a Vandermonde matrix with pairwise different values of~$z$), to compute in polynomial
time all the values $N_r$ for $r\in
\{0,\ldots,n\}$ as we wanted. This concludes the proof.
  \end{proof}

  With this in place, we now focus on proving
  \cref{thm:w1hard}. Remember that the problems
  $\probname^{\Walk}_{\leq}$ and
  $\probname^{\Path}_{\leq}$ are the same
  problems.
  We first provide a proof of a weaker statement that applies only to $\probname^{\Path}_{\leq}$ for general probabilistic graphs (i.e., without restricting to acyclic graphs), and then a second proof (which is the one that we sketched in the maintext) that works for all variants and only uses acyclic graphs.  
Both arguments were originally suggested by an LLM (but written by the authors):
  even though the first proof is subsumed by the second one, we give it nevertheless
  because it is much simpler and the source problem is different. 

  \subsubsection{First Proof of W[1]-hardness for $\probname^{\Path}_{\leq}$}

First, we show that 
  $\probname^{\Path}_{\leq}$ is W[1]-hard when we
  consider \emph{undirected} graphs (as opposed to
  the directed graphs we usually consider in this
  article).  The definition of
  $\probname^{\Path}_{\leq}$ on undirected graphs is
  the same as when we have directed graphs, except
  that the input is now an undirected probabilistic
  graph $\calG=(V,E,\pi)$ (with $E$ containing pairs
  of the form $\{v,v'\}$ for $v,v'\in V$, $v\neq v'$). 
  We show:

  \begin{proposition}
    \label{prp:w1-leq-undirected}
    On probabilistic \emph{undirected} graphs, the problem 
    $\probname^{\Path}_{\leq}$ is W[1]-hard (under Turing reductions)
  parameterised by pathwidth.
  \end{proposition}
  \begin{proof}
  We reduce from the Minimum Length-Bounded Cut
  problem (MLBC). The input to this problem is an
  undirected graph $G=(V,E)$, two vertices $s,t\in V$
  (called source and target), and two integers
  $\ell,g\in \mathbb{N}$. For a subset of edges
  $C\subseteq E$, we say that $C$ is an \emph{$\ell$-cut} if
  every path from~$s$ to~$t$ of length at most $\ell$
  contains an edge in $C$; in other words, if $s$
  and~$t$ are at distance $> \ell$ in the graph where
  we have removed all edges of $C$. The goal of
  the problem is then to determine if there exists an
  $\ell$-cut of size at most $g$. This problem has
  been shown to be W[1]-hard parameterised by
  pathwidth
  in~\cite{dvovrak2018parameterized,bentert2022length}.

  We now explain the reduction. Let $G=(V,E)$, $s,t
  \in V$, $\ell,g \in \mathbb{N}$ be the input instance to
  MLBC, and let $n=|E|$. For $c \in \{0,\ldots,n\}$, define $N_c$ to
  be the number of $\ell$-cuts of size $c$. Notice
  then that the goal is simply to determine if $\sum_{c=0}^g
  N_c >0$. For $c \in \{0,\ldots,n\}$, denote by $\bar{N_c}$ the number of subsets of size $c$ of $E$ that are \emph{not} $\ell$-cuts, i.e., $\bar{N_c} = \binom{n}{c} - N_c$. Then, by what precedes, if we can compute all the numbers $\bar{N_c}$ for $c \in \{0,\ldots,n\}$, we can solve the MLBC problem on this instance. 
  We will do this using our oracle for $\probname^{\Path}_{\leq}$, and more
  specifically we invoke
  \cref{lem:interpol} with the property $P$ being the
  following: it takes a graph $G=(V,E)$, a subset $S\subseteq E$
  and parameters $s,t\in V$ and $\ell\in \mathbb{N}$, and it is
  satisfied if $E\setminus S$ is \emph{not} an $\ell$-cut for $s,t$. The probabilistic
  computation problem for this property $P$ is precisely our target problem
  $\probname^{\Path}_{\leq}$.
  Thus, via the reduction of
  \cref{lem:interpol}, we can 
  use our $\probname^{\Path}_{\leq}$ oracle to solve the generalised cardinality
  computation problem for~$P$, i.e., to
  compute in polynomial time for every
  $r\in \{0,\ldots,n\}$ the number of subsets $S\subseteq E$ of
  size $r$ such that $E\setminus S$ is \emph{not} an $\ell$-cut (we set
  $T\coloneq \emptyset$, i.e., there are no exogenous facts).
  Since the number of subsets $S\subseteq E$ of
  size $r$ that are not an $\ell$-cut is equal to the number of
  subsets $S'\subseteq E$ of size $n-r$ such that $E\setminus S'$
  is not an $\ell$-cut, we can indeed compute in polynomial time
  all the numbers $\bar{N_c}$, as desired.

Last, we notice that the reduction is indeed a valid
FPT (Turing) reduction, since the graph, in
particular its pathwidth, was not modified at all. This
concludes the proof.
\end{proof}

Hence, all that is left to prove \cref{thm:w1hard} is to explain
how to adapt the proof to the \emph{directed} probabilistic graphs
that we study in this article. To do this, we consider essentially the same
reduction from \cref{prp:w1-leq-undirected}, but we replace every
undirected edge $\{v,v'\}$ of the undirected graph $G=(V,E)$ with the gadget shown in
Figure~\ref{fig:dir-undir}, with $a_{\{v,v'\}}$ and $b_{\{v,v'\}}$
being fresh vertices. 
(The gadget in question was proposed by an LLM but appears to be folklore, and is mentioned, e.g., to be a well known 
reduction in \cite[Section~2]{azar2003optimal}.)
Call $\overrightarrow{G}$ the directed graph obtained from~$G$ in this fashion. Let us write $U$ and call the \emph{coding edges} the set of edges of $\overrightarrow{G}$ of the form
$a_{\{v,v'\}} \to b_{\{v,v'\}}$, and $D$ the set of all other edges, that we
call the 
\emph{exogenous} edges.
For a set of edges $S\subseteq E$, let $S'$ be the set of coding edges of
$\overrightarrow{G}$ that they correspond to, i.e., $S' = \{\{a_{\{v,v'\}},
b_{\{v,v'\}}\}\mid \{v,v'\} \in S\}$. We then observe the following: $S$ is not
an $\ell$-cut of $G$ if and only if $S'$ is not a $3\ell$-cut of
$\overrightarrow{G}$. Adapting the reduction is then straightforward, using
$T\coloneq D$ in the generalised cardinality problem to force the exogenous edges not to
be considered in a cut,
and also noting that the transformation from~$G$ to $\overrightarrow{G}$ only
multiplies the treewidth by a constant factor (indeed, the two fresh vertices introduced
in the gadget can be added in any bag witnessing the original undirected edge,
and this bag now witnesses the five directed edges of the gadget).

\begin{figure}
\centering
  \begin{tikzpicture}[xscale=1.1]
    \node (v) at (0,0) {$v$};
    \node (a) at (2,1) {$a_{\{v,v'\}}$};
    \node (b) at (2,-1) {$b_{\{v,v'\}}$};
    \node (v') at (4,0) {$v'$};
    \draw[->] (v) -> (a);
    \draw[->] (v') -> (a);
    \draw[->] (a) -> (b);
    \draw[->] (b) -> (v);
    \draw[->] (b) -> (v');
  \end{tikzpicture}
  \caption{Gadget to adapt the reduction of \cref{prp:w1-leq-undirected} to directed graphs.}
  \label{fig:dir-undir}
\end{figure}

Last, we remark here that this proof could be adapted to show the
W[1]-hardness of the $\geq$-variants, if one could show that the corresponding
cut problems (i.e., we want that there is no path of length
\emph{at least} $\ell$ connecting $s$ and $t$) are W[1]-hard; however to
our knowledge this is not known.

\subsubsection{A Second Proof of W[1]-hardness that Works for All Variants}

In this section we give a full proof of
\cref{thm:w1hard} that works for all variants of the
problems $\probname^{X}_{\bowtie}$ for ${\bowtie} \in
\{{\leq}, {\geq}\}$ and $X \in \{\Walk, \Path\}$, and that in addition shows hardness already
  for the case of \emph{acyclic} probabilistic graphs.
We focus first on the case of $\probname^{\Path}_{\geq}$ and explain after that why the proof also works for the other variants.
For any integer $q\geq 1$, we write $[q]$ for
the set $\{1,\ldots,q\}$. 

The source of hardness will be the problem of
counting the linear extensions of a partially ordered
set (poset). Recall that a \emph{poset}
$\calP=(A,\leq_\calP)$ is a set $A$ together with a
relation ${\leq_\calP} \subseteq A^2$ that is
reflexive, transitive and antisymmetric. For $x,y\in
A$, we say that $y$ \emph{covers} $x$, written
$x\prec_\calP y$, if we have $x \neq y$, we have $x \leq_\calP y$, and
for every $z\in A$ such that $x \leq_\calP z
\leq_\calP y$ it holds that $z=x$ or $z=y$. The
\emph{cover graph} of $\calP$ (also called the Hasse diagram) is the graph
$G_\calP=(A,\prec_\calP)$, i.e., whose vertices are
the elements of $A$, and which contains a directed edge
$x \to y$ whenever $x\prec_\calP y$. A \emph{linear
extension} of $\calP$ is a 
permutation $\pi\colon A \to [|A|]$ which ensures that, for every $x,y \in A$,
if $x \leq_\calP y$ then $\pi(x) \leq \pi(y)$.

Every poset admits linear extensions, which can be computed in linear time using any topological sorting algorithm.
By contrast,
the authors of \cite{eiben2019counting} show that the problem of counting the linear extensions of a (finite) poset $\calP$ is W[1]-hard under Turing reductions parameterised by the treewidth of $G_\calP$, and those of
  \cite{belmonte2022grundy} (cf. Section 1.1 and text before) actually notice that their proof also holds for pathwidth. So we will reduce from the following:

\begin{theorem}[\cite{eiben2019counting, belmonte2022grundy}]
  \label{thm:linexcount}
  Counting the linear extensions of a finite poset is W[1]-hard (under Turing reductions) parameterised by the pathwidth of (the cover graph of) the poset.
\end{theorem}

We will reduce that problem to
$\probname^{\Path}_{\geq}$, in two steps. First, we
reduce to the problem of counting the number of
\emph{order-preserving maps to a chain}, as defined
next. A
\emph{$q$-order-preserving map} from a poset
$\calP=(A,\leq_\calP)$ is a map $f\colon A \to [q]$
which ensures that, for every $x,y\in A$, if $x
\leq_\calP y$ then $f(x)\leq f(y)$. Notice that
$f(x)=f(y)$ is allowed; in fact,
notice that the linear extensions of $\calP$ correspond exactly to the $|A|$-order-preserving maps of $\calP$ that are bijective.
The problem of \emph{counting
order-preserving maps for posets} then takes as input a
poset $\calP$, an integer $q\geq 2$ given in unary, and outputs the
number of $q$-order-preserving maps of $\calP$.
We claim the following, which is the first step of the proof:

\begin{proposition}
  \label{prp:linexttoopm}
  The problem of counting linear extensions of posets has an FPT (Turing) reduction, parameterised by the pathwidth of $G_\calP$, to the problem of counting the number of order-preserving maps of posets.
\end{proposition}
\begin{proof}
  For a poset $\calP=(A,\leq_\calP)$,
  let $\Omega_\calP(q)$ be the number of
  $q$-order-preserving maps of $\calP$. It is known that $\Omega_\calP(q)$ is a
  polynomial of degree $|A|$ in the variable~$q$, and that its leading coefficient is equal to the
  number of linear extensions of $\calP$ divided by $|A|!$: see for
  instance~\cite[paragraph before Example 3.12.2]{stanley2011enumerative}. Thus,
  using the problem of counting order-preserving maps as an oracle to compute
  $\Omega_\calP(q)$ on $|A|+1$ different values of $q$ (which we write in unary), we can recover in polynomial
  time the number of linear extensions of $\calP$ (this is again inverting a
  Vandermonde matrix, or equivalently, performing Lagrange interpolation to
  compute the $|A|+1$ coefficients of the degree-$|A|$ polynomial from
  the evaluation in $|A|+1$ points). 
This is clearly a valid FPT reduction as the poset $\calP$ has not been modified.
\end{proof}

Thus, all that is needed to show the W[1]-hardness of $\probname^{\Path}_{\geq}$ is the following:

\begin{proposition}
  \label{prp:w1best}
  The problem of counting the number of order-preserving maps of
  posets, parameterised by pathwidth of their cover graph, has an FPT (Turing) reduction to
  $\probname^{\Path}_{\geq}$, parameterised by pathwidth, on acyclic probabilistic graphs.
\end{proposition}
\begin{proof}
  Let $\calP=(A,\leq_\calP)$ be the input poset and $q\geq 2$ be an integer
  given as input (in unary); our goal is to compute the number of 
  $q$-order-preserving maps of $\calP$. Let $n\coloneq |A|$.
  Build a directed graph $G=(V,E)$ as follows:

  \begin{enumerate}
    \item for every element $x\in A$, we have a vertex $u_x$ and a vertex $v_x \in V$;
    \item \label{itm:identityedges} for every $x\in A$, we have an edge $(u_x, v_x)$;
    \item \label{itm:internal} for every $x,y\in A$ such that $x\prec_\calP y$, we have an edge $(u_x, v_y)$;
    \item we have $s,t \in V$ for two fresh source and target vertices $s$ and $t$;
    \item \label{itm:from} for every $x\in A$ and every $r\in [q-1]$, we have a fresh path of length $r$ from $s$ to $u_x$;
    \item \label{itm:to} for every $x\in A$ and every $r\in [q-1]$, we have a fresh path of length $r$ from $v_x$ to $t$;
  \end{enumerate}

  Notice that $G$ is indeed a directed acyclic graph.
  We now let $U$ be the set of edges containing
  the first edge of each path from \cref{itm:from}, and the first 
  edge of each path in \cref{itm:to}; we call $U$ the \emph{coding edges}. We
  let $D$ denote all other edges, and call them \emph{exogenous}. 
This construction is illustrated in \cref{fig:constr} on a small poset with $3$ elements.
  For $x\in A$, let $U_x$ be the set of edges of $U$ that are on the path from $s$ to $u_x$ or on the path from $v_x$ to $t$, and call them the \emph{$x$-related coding edges}. Notice that $U = \bigcup_{x\in A} U_x$ with the union being disjoint.
    It is not hard to check that the pathwidth of $G$ is at most $2k+5$,
    for $k$ the pathwidth of the cover graph of $\calP$. Indeed, from a path
    decomposition of the cover graph of~$\calP$, we can simply replace each
    element $x$ of $\calP$ by the two vertices $u_x$ and $v_x$, add $s$ and $t$
    to every bag, and insert one sequence of bags for each path of 
    \cref{itm:from} and of \cref{itm:to} above.

  We now use our oracle for $\probname^{\Path}_{\geq}$, via
  \cref{lem:interpol}, to solve the generalised cardinality computation problem
  for the property $P$ defined as follows: given a graph $G = (V, E)$,
  vertices $s$ and $t$, the length~$q$, and a subset $S \subseteq E$, the
  property $P$ is satisfied if $G[S]$ contains a path of length $> q$. Observe
  that the probabilistic computation problem for this property~$P$ is precisely
  our target problem $\probname^{\Path}_{\geq}$. Now, consider the generalised
  cardinality computation problem with the set of exogenous edges
  $T\coloneq D$, and the target cardinality
  $c \coloneq n(q-1)$ for the sets $S$.
  The oracle for the generalised cardinality computation problem thus allows us to
  compute
  the number $\omega'$ of subsets $S\subseteq U$ of cardinality
  $n(q-1)$ such that $G[D \cup S]$ contains a path from $s$ to $t$
  of length $> q$. By changing variables and counting over the complements
  of~$S$ relative to~$U$ instead, this is equal to the number of subsets $S \subseteq U$
  of cardinality $|U|-n(q-1) = 2n(q-1) - n(q-1) = n(q-1)$ such that $G[D \cup (U
  \setminus S)] = G[E\setminus S]$ contains a path from $s$ to~$t$ of length~$>
  q$. From this, we deduce the number $\omega \coloneq 
  {2n(q-1) \choose n(q-1)} - \omega'$ of subsets $S \subseteq U$ of cardinality
  $n(q-1)$ such that $G[E\setminus S]$ does not contain a path from~$s$
  to~$t$ of length~$>q$.
  We now claim that $\omega$
  is equal to the number of $q$-order-preserving maps of~$\calP$,
  which would conclude the proof because we can then 
  simply return $\omega$, and the overall reduction is indeed a
  valid FPT (Turing) reduction.

  To prove this claim, we show that there is a bijection between
  the subsets $S\subseteq U$ of size $n(q-1)$ such that $G[E\setminus S]$
  contains no path from $s$ to $t$ of length $> q$, and the
  $q$-order-preserving maps of~$\calP$.

  Let $S$ be such a set. For $x\in A$, let $a_x$ be one plus the length of the longest path from $s$ to $u_x$ in $G[E\setminus S]$ if such a path exists, and $1$ otherwise. Notice that $a_x$ is in $[q]$.
  We will show that the function $f_S\colon A\to [q]$ that associates $x\in A$ to $a_x$ is a $q$-order-preserving map. Let
   $b'_x$ be one plus the length of the longest path from $v_x$ to $t$ in $G[E\setminus S]$ if such a path exists, and $1$ otherwise; and define $b_x = q+1-b'_x$. Thanks to the edges from \cref{itm:identityedges}, we have that
  $a_x + b'_x \leq q+1$ (as otherwise we would have a path of length
  $>q$ from $s$ to $t$ in $G[E\setminus S]$), i.e., $a_x \leq
  b_x$. Furthermore, having $a_x +b'_x \leq
  q+1$ implies that $|U_x \cap S| \geq q-1$: this is intuitively because we have paths of
  every possible length in $[q-1]$ from $s$ to $u_x$ and likewise from $v_x$
  to~$t$ so that we
  must remove one edge each time we decrement the sum of the lengths of the
  longest remaining paths, so bringing the sum of lengths down from~$2(q-1)$ to~$q-1$
  requires $q-1$ edge removals. Now, because
  $U=\bigcup_{x\in A} U_x$ with the union being disjoint, and
  because $|S|=n(q-1)$, we must have in fact that  $|U_x \cap S| = q-1$ for each $x
  \in A$.
  This implies that we must have $a_x=b_x$ for $x \in A$. Indeed, we know that
  $a_x \leq b_x$, and if we have
  $a_x < b_x$ then we have $a_x + b_x' < q+1$ which by the decrement argument
  earlier would require $>q-1$ edge removals, exceeding our budget of $|U_x \cap
  S| = q-1$ for this $x$.
  Now, thanks to the edges of \cref{itm:internal}, for every $x,y\in A$ with
  $x\prec_\calP y$, we must also have that $a_x + b'_y \leq q+1$, i.e., $a_x \leq
  b_y$, so $a_x \leq a_y$ because we showed that $b_y = a_y$.
  We have established that $f_S$ is order-preserving on the cover relation
  $\prec_\calP$, which thanks to transitivity is sufficient to establish that we
  have $f_S(x) \leq f_S(y)$ for all $x,y\in A$ with $x \leq_\calP y$, e.g., by
  considering a chain $x \prec_\calP z_1 \prec_\calP \cdots \prec_\calP z_k
  \prec_\calP y$ and combining the inequalities.
  Thus, we have indeed shown
  that the function $f_S$ is a $q$-order-preserving map. 

  Conversely, if $f\colon A \to [q]$ is a $q$-order-preserving map, let us
  define $S_f \subseteq U$ as follows. For $x\in A$, we put in
  $S_f$ all the edges of $U_x$ that lie on a path from $s$ to $u_x$
  of length at least $f(x)$ (there are $q-f(x)$ such
  edges), as well as all those that lie on a path from $v_x$ to $t$
  of length at least $q-f(x)+1$ (there are $f(x)-1$ such edges). It
  is clear then that $|S_f|=n(q-1)$, so we only need to argue why it is
  the case that $G[E\setminus S_f]$ contains no path of length
  $>q$ from $s$ to $t$. Let $p$ be a longest path in
  $G[E\setminus S_f]$ from $s$ to $t$, if it exists (otherwise there is nothing to show). By construction of~$G$, paths from $s$
  to~$t$ must either go through an edge of
  the form $(u_x, v_y)$ with $x \neq y$ (case 1), or through an edge of the form $(u_x, v_x)$ (case 2). For case~1, by definition of $S_f$ and by
  construction of $G$, we have that $1<f(x)$ and $f(y)<q$, and we see that the part of $p$ between $s$ and
  $u_x$ has length $f(x)-1$ and that the part between $v_y$ and $t$ has length
  $q-f(y)$, so that the total length of the path is
  $f(x)-1+q-f(y)+1=q+f(x)-f(y)$. But $f(x)\leq f(y)$ since $f$ is
  order-preserving, hence the length of $p$ is at most
  $q$. For case 2, if a path uses an edge of the form $(u_x, v_x)$, then we have $1<f(x)<q$ and the length of the path is exactly $q$. Therefore, 
  $G[E\setminus S_f]$ indeed contains no path of length
  $>q$ from $s$ to $t$.

  The last two paragraphs establish the existence of the bijection that we
  claimed, so this concludes the proof.
\end{proof}

We now conclude the proof of \cref{thm:w1hard}. Our proof given so far works for
the variants $\probname^{\Walk}_{\geq}$ and equivalently
$\probname^{\Path}_{\geq}$ because the constructed graph is acyclic. Namely, the
W[1]-hardness follows directly from \cref{prp:w1best} together with
\cref{prp:linexttoopm}.
We now explain why this reduction can be adapted to the $\leq$ variants of
$\probname^{X}_{\bowtie}$, i.e., 
to $\probname^{\Path}_{\leq}$ (or $\probname^{\Walk}_{\leq}$ which is the same
problem).
We use oracle calls
to $\probname^{\Path}_{\leq}$ to compute the number of subsets $S\subseteq U$
of size $n(q-1)$ such that $G[E\setminus S]$ contains no path of length $<q+2$. 
It is then routine to check that the rest
of the proof also adapts. Namely, we have a bijection between the subsets
$S\subseteq U$ of size $n(q-1)$ such that $G[E \setminus S]$ contains no path from
$s$ to $t$ of length $<q+2$, and the $q$-order-preserving maps of~$\calP$.
This is shown in the forward direction by defining, for each set $S\subseteq U$, the value $a_x$ is the length of the \emph{shortest} path from $s$
to $u_x$ if such a path exists and $q$ otherwise, and similarly for $b'_x$. We deduce that $a_x + b_x' \geq q+1$ so $a_x
\geq b_x$, and the disjoint union argument with the dual observation that
incrementing the length of shortest paths requires one edge removal per
increment implies that $a_x = b_x$ for each~$x$. The rest of the proof
of the forward direction follows by defining $f_S(x) \coloneq q-a_x+1$, i.e.,
whenever $x \prec_\calP y$ we must have $a_x \geq a_y$ so $f_S$ as we now define
it is an order-preserving map. 
Likewise, the proof of the backward direction
follows, except that we now consider shortest paths and that 
we construct $S_f$
from~$f$ by deleting the edges of the $q-f(x)$ \emph{shortest} paths from $s$ to~$u_x$
and of the $f(x)-1$ \emph{shortest} paths from~$v_x$ to~$t$.
This concludes the proof of \cref{thm:w1hard}.
\end{toappendix}

Note that this is in contrast with the non-probabilistic case, which as we
recalled is FPT (\cref{thm:nonprobatw}), and also with the fact that ST-CON is FPT parameterised by treewidth on probabilistic graphs. We now focus on proving \cref{thm:tract}
in this section, before turning in the next section to the intractability of
the exact-length variants. 

We first deal with the at-most-length variants (for both paths and walks),
which follows a dynamic-programming approach and introduces the key ideas. We then 
adapt this to the at-least-length variant for walks following similar lines. We
last cover the more complicated case of the at-least-length variant for simple
paths: there the dynamic programming algorithm is more elaborate (to handle
both the length requirements and the simple-path requirements), in particular
it generalises the tractability of $\lppg$ (see \cref{thm:nonprobatw}) to
probabilistic graphs.

Before we start, we quickly generalise the notation $G_b$ to
probabilistic graphs in the expected way, as follows. For $\calG$ a
probabilistic graph with underlying non-probabilistic graph $G$,
for $\calS = (S, \{S_b\}_{b \in V(S)})$ a tree decomposition of $G$, and
for $b$ a node of $S$, we write $\calG_b$ the probabilistic graph
whose underlying non-probabilistic graph is $G_b$ and where the edge
probabilities are as indicated by~$\calG$. Notice that, by the properties of nice tree
decompositions, if $b$ is a \textsc{Join} node with children $b_1,b_2$, then the probability distribution of~$\calG_b$
is the independent product of those of $\calG_{b_1}$
and $\calG_{b_2}$.

\begin{toappendix}
\subsection{Proofs for At-Most Length Variants}
\end{toappendix}
\subsection{At-Most Length Variants: $\lppl$ and $\lpwl$}\label{sub:tractless}
Remember that $\lppl$ and $\lpwl$ are different names for the exact same problem; we will use the name $\lppl$ here. Further note that in our subsequent
discussion we shall assume without loss of generality that $\ell < |V|$. This is because if $\ell \geq |V|$, then the probability that a path of length at most
$\ell$ exists is the same as the probability that a path of length at
most $|V|-1$ exists, hence it suffices to look at the instances where
$\ell < |V|$.

To prove \cref{thm:tract} in the case of $\lppl$, we introduce an intermediate problem,
that we call the \emph{shortest length probability distribution
computation} problem ($\slpdf$). The plan is to then reduce $\lppl$
to $\slpdf$, and then to show that $\slpdf$ is tractable on probabilistic graphs of bounded treewidth.

\myparagraph{Shortest Length Probability Distribution Computation}
In the \emph{shortest length probability distribution computation} problem, we are given a probabilistic graph $\G$, two vertices~$s$ and~$t$, and our aim is to compute, for every integer $\ell$ ranging from $0$ to $|V|-1$, the probability that the length of the shortest path from $s$ to $t$ 
is $\ell$. In addition, we will want to compute the probability that no path from $s$ to $t$ exists at all: in this case we say that the shortest path length from $s$ to $t$ is $\bot$. Formally:

\begin{center}
\fbox{\begin{tabular}{lp{10cm}}
    \small{PROBLEM} : & $\slpdf$ \hspace{.5cm}(\emph{Shortest Length Probability Distribution Computation})
\\
{\small INPUT} : & A probabilistic graph $\G = (V,E,\pi)$, two vertices $s,t\in V$\\
{\small OUTPUT} : &  For every $\ell  \in \{0,1,\ldots, |V|-1\} \cup \{\bot\}$ the probability that the shortest path from $s$ to $t$ is of length $\ell$
\end{tabular}}
\end{center}

We now reduce $\lppl$ to $\slpdf$. Let $\G$ be a probabilistic graph. Consider the $\slpdf$ output for $\G$ and the vertices $s$ and $t$, seen as a table $T$ with $T[\ell]$ giving the probability that the shortest path from~$s$ to~$t$ has length $\ell$, for $\ell \in \{0,1, \ldots, |V|-1\} \cup \{\bot\}$.
Since all of these correspond to mutually exclusive events, the probability that there is a path of length at most $\ell$ from~$s$ to~$t$ can simply be computed as
$\sum_{k=0}^{\ell}T[k]$.
Therefore, to prove \cref{thm:tract} in the case of $\lppl$, it suffices to obtain a
polynomial time algorithm for $\slpdf$ on bounded-treewidth
probabilistic graphs. The key idea is to compute the distribution
over \emph{shortest path length functions} for each node of the tree
decomposition; we explain this next.

\myparagraph{Shortest path length functions} 
Given a graph $G=(V,E)$, define the \emph{shortest path length function
for $G$}, written $\spl_G\colon V\times V \to \{0,1, \ldots, |V|-1\} \cup
\{\bot\}$, to be the function such that for every $u,v \in V$,
the value $\spl_G(u,v)$ equals the length of a shortest path from~$u$ to~$v$
in $G$, and $\bot$ if no path from $u$ to~$v$ exists.  
For $A\subseteq V$ a
set of vertices of $G$, we denote by $\spl_{G,A}$ the restriction
of $\spl_G$ to the domain $A\times A$.

Let $\calG$ be the probabilistic graph given as input of $\slpdf$,
with $\calS = (S, \{S_b\}_{b \in V(S)})$ a nice tree decomposition.
The key idea to solve $\slpdf$ is to compute, by bottom-up induction on
$\calS$, for every node $b$ of $\calS$ and for every possible
function~$f \colon S_b\times S_b \to \{0,1, \ldots, |V|-1\} \cup \{\bot\}$,
the probability that, in a random realisation $H_b$ 
of~$\calG_b$,
the function $f$ is precisely equal to $\spl_{H_b,S_b}$.
Note that, because there are $|V|^{\calO(\treewidth(G)^2)}$ such
functions for each node $b$, the resulting algorithm will be in XP
parameterised by treewidth, and not FPT as was the case of
\cref{thm:nonprobatw} for non-probabilistic graphs.
Formally:

\begin{toappendix}
  \label{apx:lppl}
\end{toappendix}

\begin{toappendix}
In this section we first prove \cref{induction1} used for showing the tractability of $\lppl$ for probabilistic graphs of bounded treewidth (refer \cref{sub:tractless}), by providing an algorithm that computes the shortest path length function for each node in a bottom-up manner (\cref{alg:TreeDP}). The computation is fairly simple when the node under consideration is an $\IntroVertNode{}$ node, a $\ForgetVertNode{}$ node, or an $\IntroEdgeNode{}$ node. The case of a $\JoinNode{}$ node is more interesting, as given a realisation $H_b$ of a $\JoinNode{}$ node $b$ with children $b_1$ and $b_2$ which can be obtained by picking independently realisations $H_{b_1}$ and $H_{b_2}$ of $\mathcal{G}_{b_1}$ and $\mathcal{G}_{b_2}$, we need additional machinery to determine the shortest path length function $\spl_{H_b}$ from the shortest path length functions $\spl_{H_{b_1}}$, $\spl_{H_{b_2}}$. In particular we shall use the fact that it is possible to determine, in a weighted directed graph, the length of the shortest path from a given vertex $u$ to another $v$. To that effect, we define the \emph{All Pairs Shortest Path Problem}.  

\myparagraph{All Pairs Shortest Path Problem} In this problem, a weighted
directed graph $G = (V, \delta \colon V \times V \to \NN{} \cup \{\bot\})$ is given and we must determine the length of the shortest path from vertex $u$ to vertex $v$ for all such $(u,v)$ pairs of vertices, where $\bot$ denotes that the edge does not exist, and where the values returned can be $\bot$ when no path exists from one vertex to another. 
(In other words, $\bot$ behaves like $+\infty$, but we do not use $+\infty$ because it is used with a different semantics in \cref{sec:atleastwalk}.)
Here the length of the path refers to the sum of weights of its constituent edges.
The problem is known to be solvable in polynomial time ($\mathcal{O}(|V|^3)$)
using algorithms such as the Floyd-Warshall algorithm. From now on we denote by
$\apsp{}$ an algorithm that solves in polynomial time the following All Pairs Shortest Path Problem:

\begin{center}
\fbox{\begin{tabular}{lp{10cm}}
    \small{Algorithm}: & $\apsp$($G$) :  \hspace{.5cm}(\emph{All Pairs Shortest Path Algorithm})
\\
{\small INPUT}: & A weighted graph $G = (V, \delta \colon V \times V \to \NN{} \cup \{\bot\})$
\\
{\small OUTPUT}: & A table $T$ with values in $\NN \cup \{\bot\}$ such that for any pair of vertices $u,v \in V$, the length of a shortest path from~$u$ to~$v$ in $G$ is given by $T(u,v)$ (with $\bot$ denoting that no such path exists).
\end{tabular}}
\end{center}

\noindent Next, we define a \emph{change-over vertex} of a path $P$ from a vertex $s$ to another vertex $t$. 

\begin{definition}\label{def:changeoververtices}
Let $G_1 = (V_1,E_1)$ and $G_2 = (V_2,E_2)$ be two graphs, with $E_1 \cap E_2 =
\emptyset $ and possibly some vertices in common, and let $G = (V_1 \cup V_2,
E_1 \uplus E_2)$ be the union of these two graphs.  Let $P$ be any path from $s$
to $t$ in $G$. A \emph{change-over vertex} of $P$ is a vertex $v$ satisfying either of the following properties:

\begin{itemize}
    \item $v$ is an endpoint of $P$, namely $v = s$ or $v = t$.
    \item Otherwise, letting $v_1$ and $v_2$ be the vertices immediately before
      and after $v$, then the edge joining $v_1$ and $v$ belongs to $E_1$ and
      the edge joining $v$ and $v_2$ belongs to $E_2$, or vice versa.
\end{itemize}
\end{definition}

Now, we prove the key claim that makes it possible to compute the shortest path
length function for the $\JoinNode{}$ node in \cref{alg:TreeDP}. This essentially depends on the fact that bags of tree decomposition are separators.

\begin{claim}\label{changeclaim}
  Let $\calG = (V,E,\pi)$ be a probabilistic graph with 
  underlying
  non-probabilistic graph $G=(V,E)$, 
  let $\calS$ be a nice tree decomposition of $\calG$, and 
  let $b$ be a \textsc{Join} node with children ${b_1}$ and ${b_2}$.
  Note that $S_b=S_{b_1}=S_{b_2}$. Recall that we denote by $\mathcal{G}_{b_1}$ and $\mathcal{G}_{b_2}$ the probabilistic subgraphs induced by the subtrees rooted at $b_1$ and $b_2$ respectively. Let $G_1=(V_1,E_1)$ and $G_2=(V_2,E_2)$ be realisations of $\mathcal{G}_{b_1}$ and $\mathcal{G}_{b_2}$ respectively. Define $H_b = (V_1 \cup V_2, E_1 \cup E_2)$ to be the graph formed by taking the union of $G_1$ and $G_2$. 
  Consider any path $P$ in $H_b$ from a vertex $s$ to a vertex $t$ where $s,t \in S_b$. Then:

    \begin{enumerate}
        
    \item \label{enum1} Any \emph{change-over vertex} belonging to the path $P$ must necessarily belong to $S_b$.

    \item \label{enum2} Define $H(S_b, \delta)$ to be the weighted directed
      graph with set of vertices $S_b$ and with weights $\delta(v_i,v_j) = \min(f_1(v_i,v_j),f_2(v_i,v_j))$ for all $v_i,v_j \in S_b$, where 
      $f_1(v_i,v_j) \coloneq \spl_{G_1,S_b}(v_i,v_j)$ and 
      $f_2(v_i,v_j) \coloneq \spl_{G_2,S_b}(v_i,v_j)$,
      where $\min$ ignores~$\bot$ in its arguments except when both arguments are~$\bot$ in which case it evaluates to~$\bot$.
      Then we have $\spl_{H_b,S_b} = f$, where $f$ is given as follows: For every $v_i,v_j \in S_b$, $f(v_i,v_j)$ equals the length of a shortest path from $v_i$ to $v_j$ in the weighted directed graph $H(S_b, \delta)$
    \end{enumerate}

\end{claim}

\begin{proof} 
  We first prove \Cref{enum1}. Let $w$ be a change-over vertex. If $w$ is $s$
or $t$ then $w\in S_b$ by hypothesis. Otherwise,
there are $w',w''\in V$ with $(w',w)$ an edge of
$G_1$ and $(w,w'')$ an edge of $G_2$. By definition
of $G_1$ and $G_2$, this implies that $w$ occurs both
in the subtree rooted at $b_1$ and in the subtree
rooted at $b_2$, hence by the connectedness property
of tree decompositions $w$ must also occur in $S_b$.

We now prove \Cref{enum2}. First, we  prove that, for any $s,t \in S_b$, the
shortest path from $s$ to $t$ in $H(S_b,\delta)$ is no longer than the shortest path from $s$ to $t$ in $H_b$. Let $P$ be any shortest path from  $s$ to $t$ in $H_b$.
Let the change-over vertices appearing in $P$, in that order,
be $u_1,u_2, \dots, u_k$ ($s=u_1, t=u_k$); by Item 1 we know that they are all in $S_b$. Moreover by
the definition of change-over vertex, the
edges appearing between any two consecutive
change-over vertices $u_i$ and  $u_{i+1}$ for $i
\in [1, 2, 3, \dots, k-1]$ lie exclusively in either
$E_1$ or $E_2$. But the shortest path from $u_i$
to $u_{i+1}$ lying exclusively in $G_1$ has length
$f_1(u_i,u_{i+1})$ whereas the shortest path lying
exclusively in $G_2$ has length $f_2(u_i,u_{i+1})$,
therefore if we consider the shorter of these two
paths then it has length
$\min(f_1(v_i,v_j),f_2(v_i,v_j))$ and is shorter than
any other path lying exclusively in $G_1$ or $G_2$.
Therefore the length of the section of the path $P$ between
$u_i,u_{i+1}$ must be equal to
$\min(f_1(u_i,u_{i+1}),f_2(u_i,u_{i+1}))$ (otherwise $P$ would not be a shortest path from $s$ to $t$). The total
length of the path $P$ is therefore
$\sum^{k-1}_{i=1}\min(f_1(u_i,u_{i+1}),f_2(u_i,u_{i+1}))$.
Now consider the path $s=u_1, u_2, \ldots, u_k=t$ in $H(S_b, \delta)$. Observe
that the length of this path is $\sum^{k-1}_{i=1}\delta(u_i,u_{i+1})$ which is
precisely the length of $P$, i.e., the length of any shortest path from $s$ to
$t$ in $H_b$.  This establishes that, indeed, the shortest path from~$s$ to~$t$ in~$H(S_b,\delta)$ is no longer than the shortest path from~$s$ to~$t$ in~$H_b$.

Next we show that the path $u_1, u_2, \ldots, u_k$ is indeed the shortest path
in $H(S_b, \delta)$. Suppose by contradiction that there is a shorter path $v_1,
v_2, \dots, v_l$ where $v_1 = s, v_l = t $ and
$\sum^{l-1}_{j=1}\delta(v_j,v_{j+1}) < \sum^{k-1}_{i=1}\delta(u_i,u_{i+1})$.
Then we replace every edge $(v_i,v_{i+1})$ of weight $\delta(v_i,v_{i+1})$ with a path of equal length in $H_b$: if $\delta(v_i,v_{i+1}) = f_1(v_i,v_{i+1})$ then replace the edge with a shortest path from $v_i$ to $v_{i+1}$ in $G_1$, otherwise, if $\delta(v_i,v_{i+1}) = f_2(v_i,v_{i+1})$, then replace it with a shortest path from $v_i$ to $v_{i+1}$ in~$G_2$.
Thus, we now obtain a path from $s$ to $t$ in $H_b$ with length
$\sum^{l-1}_{i=1}\delta(v_i,v_{i+1})$. Further by assumption
$\sum^{l-1}_{i=1}\delta(v_i,v_{i+1}) < \sum^{k-1}_{j=1}\delta(u_j,u_{j+1})$.
However, $\sum^{k-1}_{j=1}\delta(u_j,u_{j+1})$ is the length of $P$, a shortest
$s$-$t$ path in $H_b$. Thus, we arrive at a contradiction. Hence, $u_1, u_2,
\ldots, u_k$ is indeed the shortest path from $s$ to $t$ in $H(S_b, \delta)$,
and its length is the same as that of shortest path from $s$ to $t$ in $H_b
(=\sum^{k-1}_{i=1} \delta(u_i,u_{i+1}))$, thus proving the claim.
\end{proof}

Having established \cref{changeclaim}, we are now in a position to prove \cref{induction1}, whose statement we recall here:
\end{toappendix}
\begin{claimrep}\label{induction1}
  We can compute in XP parameterised by the treewidth of~$\mathcal{G}$,
  for each bag $S_b$ of~$\calS$,
  a table~$T_b$ 
  indexed with all possible functions $f\colon S_b
  \times S_b \to \{0,1, \ldots, |V|-1\} \cup \{\bot\}$ such that
  $T_b[f]$ equals the probability that in a random realisation
  $H_b$ of $\calG_b$, the function
  $\spl_{H_b,S_b}$ is equal to $f$.
\end{claimrep}
\begin{proofsketch}
  We proceed by bottom-up induction,
  and use the fact that $S_b$ is a separator of $G_b$
  in order to be able to multiply the probabilities of events that are
  independent, and the fact that the information we compute is precise enough so
  that we can sum over mutually exclusive events. More precisely, each
  realisation $H_b$ of $\mathcal{G}_b$ contributes to precisely one cell of the
  table $T_b$. Moreover, each realisation $H_b$ of a \textsc{Join} node $b$ with
  children $b_1$ and $b_2$ can be obtained by picking independently
  realisations $H_{b_1}$ and $H_{b_2}$ of $\mathcal{G}_{b_1}$ and
  $\mathcal{G}_{b_2}$, and then the shortest path length function
  $\spl_{H_b,S_b}$ is entirely determined from the shortest path length
  functions $\spl_{H_{b_1},S_{b_1}}$ and $\spl_{H_{b_2},S_{b_2}}$ of the
  children. Indeed, shortest paths in~$H_b$ between vertices of $S_b=S_{b_1}=S_{b_2}$ can be
  decomposed into some combination of shortest paths in $H_{b_1}$ between vertices of
  $S_{b_1}$ and shortest paths in $H_{b_2}$ between vertices of~$S_{b_2}$, all of which
  are shortest paths. Hence, the total probability of the realisations $H_b$
  achieving a specific shortest path length function can be computed by the
  total probability of the realisations $H_{b_1}$ and $H_{b_2}$ achieving
  specific shortest path length functions.
\end{proofsketch}

\begin{proof}
We prove that this is what \Cref{alg:TreeDP} does, by structural induction on the nice tree decomposition. Clearly for a leaf node, the claim
holds, since leaf nodes contain no
vertices and the only function describing the
shortest distance over the induced empty subgraph is the empty function
with probability one, which is
what our algorithm assigns. Therefore the base case
holds. Assume that the claim holds for the
children of~$b$. We now show that it holds for
$b$ as well:

\begin{description}
\item[Case 1: $b$ is a \ForgetVertNode{} node]
Note that our definition of the subgraph induced at a
node $b$ covers all the edges introduced until that
point and all the vertices (in the bag or any of its
descendants). Hence, let the child of $b$ be $b'$; then
$\G_b = \G_{b'}$. Therefore for any realisation $H$ of $\G_{b'}$
such that $\spl_{H,S_{b'}}$ is given by $f$, we can map it to the same
realisation of $\G_b$ occurring with the same probability, the shortest path
length function $f'=\spl_{H,S_b}$ now being given by the restriction of $f$ to
$S_b \times S_b$. Therefore the correctness of $T_b$ follows directly from
the correctness of $T_{b'}$ obtained by induction hypothesis.

  \item[Case 2: $b$ is an \IntroVertNode{} node.]
Nodes of this kind do not introduce any edges, and they
only add an isolated vertex to the associated
subgraph: they also have precisely one child. 
The possible realisations of $\mathcal{G}_{b'}$ are thus in bijection with those of~$\mathcal{G}_b$, where the bijection adds the single isolated vertex $v_0$.
Hence, the distribution on the shortest path length functions of the new node is obtained by considering all the outcomes at the child bag (where the distribution of the shortest path length function was correctly computed by induction hypothesis), and extending each such function to say that there is a path of length~$0$ from $v_0$ to itself, no path from $v_0$ to any other vertex or to $v_0$ from any other vertex, and that the shortest path length among other vertices is unchanged. The correctness of~$T_b$ again follows directly from that of~$T_{b'}$ obtained by induction hypothesis.

\item[Case 3: $b$ is a \JoinNode{} node]
Let $b$ be a \JoinNode{} node with
children $b_1$ and $b_2$.
Consider realisations of $\G_{b_1}$ and $\G_{b_2}$, say
$G_1 = (V_1,E_1)$ and $G_2 = (V_2,E_2)$ respectively, such that the functions
$f_1 = \spl_{G_1,S_{b_1}}$ and $f_2=\spl_{G_2,S_{b_2}}$ are known. Consider the union of $G_1$ and $G_2$ which we define by $H_b=(V_1 \cup V_2, E_1 \cup E_2)$. Now, note that $H_b=(V_1 \cup V_2, E_1 \cup E_2)$ is a realisation of~$\G_{b}$, the probabilistic graph associated with the node $b$. By \Cref{changeclaim}, we know that $\spl_{H_b,S_{b}}$ is given by the function $f$ where for all $u,v \in S_b$, the value $f(u,v)$ equals the length of the shortest path from $u$ to $v$ in the weighted graph $H(S_b, \delta)$ with weight function $\delta(u,v) = \min
(f_1(u,v),f_2(u,v))$, and this holds for every realisation.
Thus we can map pairs of realisations of $\G_{b_1}$
and $\G_{b_2}$ to realisations of $\G_{b}$, with the shortest path function of the realisation of~$\G_b$ being determined from the shortest path functions of~$\G_{b_1}$ and~$\G_{b_2}$. 
The resulting shortest path function is computed 
by \Cref{alg:joinatmost} as prescribed by \Cref{changeclaim}.

\item[Case 4: $b$ is an \IntroEdgeNode{} node]
\Cref{alg:TreeDP} treats the addition of the new edge
from $u_0$ to $v_0$ as taking a union of the graph
under consideration with another graph with a single
edge from $u_0$ to $v_0$ and all other vertices
isolated. For the latter graph, the probability
distribution over shortest path length functions is trivially known and then we
use \Cref{alg:joinatmost} like for the case of a \textsc{Join} node.
\qedhere
\end{description}
\end{proof}

Once we have this, it is then easy to see that we can solve
$\slpdf$ on probabilistic graphs of bounded treewidth in XP as
well. Indeed, take a nice tree decomposition of $\calG$ whose root
bag is precisely $\{s,t\}$, and then we apply the algorithm underlying \cref{induction1} on the root bag.
What we get as a result is, for every possible function $f\colon \{s,t\}^2 \to 
\{0,1,\ldots, |V|-1\} \cup \{\bot\}$, a value $T[f]$ that is the probability that in a random realisation $H$ of $\calG$, $f$ equals $\spl_{H,\{s,t\}}$.
Now, since these are mutually exclusive events, for $\ell \in \{0,1,\ldots, |V|-1\} \cup \{\bot\}$ the probability
that the shortest path from $s$ to $t$ is of length $\ell$ is
simply given by $\sum_g T[g]$, for $g$ ranging over all
functions $g\colon \{s,t\}^2 \to  \{0,1,\ldots, |V|-1\} \cup \{\bot\}$ such
that $g(s,t)=\ell$. This concludes the proof (sketch) of \cref{thm:tract} in the case of $\lppl$.

\begin{toappendix}
{\footnotesize
\begin{algorithm}
\DontPrintSemicolon
\SetKwProg{Fn}{Function}{:}{}
\SetKwFunction{TreeDP}{TreeDP}

\KwIn{A probabilistic graph $G=(V,E,\pi\colon E \to [0,1])$ of treewidth $k$  and a nice tree decomposition $\mathcal{S}$ with root $b_\mathsf{root}$}
\KwOut{A table $T$ indexed with all possible functions $S_{b_\mathsf{root}} \times S_{b_\mathsf{root}} \to \{0,1, \ldots, |V|-1\} \cup \{\bot\}$ , such that $T[f]$ is the probability that for any pair of vertices $(u,v)$ in $S_{b_\mathsf{root}} \times S_{b_\mathsf{root}}$ the length of the shortest path from $u$ to $v$ is $f(u,v)$ (with $\bot$ denoting that no such path exists). }

\Fn{\TreeDP{$b$}}{  
  \text{For each} $f\colon S_b \times S_b \to  \{0,1, \ldots, |V|-1\}\cup\{\bot\}$, \text{set}  $T_b[f] \gets 0$\;
    \If{$b$ is a \LeafNode{} node}{
      $T_b[\text{empty function}] \gets 1$\;
      \Return{$T_b$}
    }

    \ElseIf{$b$ is a \ForgetVertNode{} node with child $b'$}{
      $T_{b'} \gets$ \TreeDP{$b'$}\;
        \ForEach{$f' \colon S_{b'} \times S_{b'} \to \{0,1, \ldots, |V|-1\} \cup \{\bot\} $}{
            Let $f$ be the restriction of $f'$ to $S_b \times S_b$\;
            $T_b[f] \gets T_b[f] + T_{b'}[f']$\;
        }
        \Return{$T_b$}
    }

    \ElseIf{$b$ is an \IntroVertNode{}  node with child $b'$ introducing vertex $v_0$}{
        $T_{b'} \gets$ \TreeDP{$b'$}\;
        \ForEach{$f'\colon S_{b'} \times S_{b'} \to  \{0,1, \ldots, |V|-1\}\cup\{\bot\}$}{
            Define $f$ by
            $
            f(u,v) \coloneq
            \begin{cases}
              f'(u,v) & \text{if } u \ne v_0 \text{ and } v \ne v_0 \\
            0 & \text{if } u = v_0 \text{ and } v = v_0 \\
            \bot & \text{otherwise}
            \end{cases}
            $
            $T_b[f] \gets T_{b}[f]+T_{b'}[f']$\;
        }
        \Return{$T_b$}
    }

    \ElseIf{$b$ is a \JoinNode{} node with children $b_1, b_2$}{
        $T_{b_1} \gets$ \TreeDP{$b_1$}\;
        $T_{b_2} \gets$ \TreeDP{$b_2$}\;
        $T_b \gets \texttt{JoinComputation$\lppl$}(T_{b_1}, T_{b_2}, S_b)$\;
        \Return{$T_b$}
    }
    \ElseIf{$b$ is an \IntroEdgeNode{} node with child $b'$ introducing edge $(u_0,v_0)$}{
        $T_{b'} \gets$ \TreeDP{$b'$}\;
        Define $f_1$ by
        $
        f_1(u,v) \coloneq
        \begin{cases}
          1 & \text{if } u = u_0 \text{ and } v = v_0 \text{ and } u \neq v \\
        0 & \text{if } u = v \\
        \bot & \text{otherwise}
        \end{cases}
        $\\
      $f_2 \gets \{(u,u) \mapsto 0 \mid u \in S_b\} \cup \{(u,v) \mapsto \bot \mid u, v \in S_b, u \neq v\}$\;
        Define $T$ by
$
T[f] \coloneq
\pi(u_0,v_0)\,\mathbf{1}_{\{f=f_1\}}
+
\bigl(1-\pi(u_0,v_0)\bigr)\,\mathbf{1}_{\{f=f_2\}}
$\\
        $T_b \gets \texttt{JoinComputation$\lppl$}(T_{b'}, T, S_b)$\;
        \Return{$T_b$}
    }

}
\Return{$\TreeDP(b_\text{root})$}
\caption{\texttt{DPComputation$\lppl$($G=(V,E,\pi),\mathcal{S})$}: Compute the
probability distribution on the shortest path length function along a tree decomposition}
\label{alg:TreeDP}
\end{algorithm}
}

\begin{algorithm}[ht]

  \KwIn{A set of vertices $S_b$ and two tables $T_1$, $T_2$ indexed with all
  possible functions $S_b \times S_b \to \{0,1,\ldots, |V|-1\} \cup \{\bot\}$}
  \KwOut{A table $T$ such that for any function $f\colon S_b \times S_b \to  \{0,1,\ldots, |V|-1\}\cup\{\bot\}$,
  the value $T[f]$ is the probability that the shortest distance from vertex $u$ to $v$ in a realization of~$\calG_b$ is given by $f(u,v)$}

\text{For each} $f\colon S_b \times S_b \to  \{0,1,\ldots, |V|-1\}\cup\{\bot\}$, \text{set}  $T[f] \gets 0$\;

\For{\text{each pair} $(f_1,f_2)$ \text{where} $f_1,f_2\colon S_b \times S_b \to  \{0,1,\ldots, |V|-1\}\cup\{\bot\}$}{
$H \gets$ Weighted graph with $\delta(v_i,v_j) \coloneq \min(f_1(v_i,v_j) , f_2(v_i,v_j) )  \text{ }\forall v_i,v_j \in S_b$\;
$f \gets \apsp(H)$\\
$T[f] \gets T[f] + T_{1}[f_1] \times T_{2}[f_2]$
}
\Return($T$)
\caption{{\texttt{JoinComputation$\lppl$$(T_1,T_2,S_b)$}}}\label{alg:joinatmost}
\end{algorithm}

\end{toappendix}

\subsection{At-Least Length Walk Variant: $\lpwg$}
\label{sec:atleastwalk}

For $\lpwg$  we can modify the algorithm of the previous section for $\lppl$ such that, for every node $b$ in the tree decomposition,
we compute for every function $f$ from $S_b \times S_b \to \{0,1,\dots,|V|-1\} \cup \{\infty,\bot\}$, the probability that $f$ represents the \emph{longest walk} function, i.e., for every pair $u,v$ of vertices, $f(u,v)$ represents the longest walk length from vertex $u$ to $v$ in a random realisation of $\calG_b$. 
Note the additional possible value $f(u,v) = \infty$, which means that there are arbitrarily long walks from~$u$ to~$v$, i.e., there is some cycle which is reachable from~$u$ and from which $v$ can be reached.
Thus, the main difference, apart from going from shortest paths to longest walks, is that we have to take care of detecting cycles (which then corresponds to reporting $\infty$).

\begin{toappendix}
\subsection{Proofs for At-Least Length Simple Path Variants}
\label{apx:proofatleastlengthsimplepaths}
\end{toappendix}
\subsection{At-Least Length Simple Path Variant: $\lppg$}\label{sub:tractmore}

We now explain how to show \cref{thm:tract} in the case of $\lppg$, whose proof
is of a different flavour than the last two.
We reuse for this the notions of path set projections that we
introduced for the proof of \cref{thm:nonprobatw}, and combine them with the
ideas used for the two other problems to correctly compute probabilities.
We first note that we can again assume without
loss of generality that $\ell < |V|$ since we are looking for simple paths (if
$\ell \geq |V|$ we can safely return $0$).

We start by defining the \emph{longest length probability distribution computation}
($\llpdf$) problem and reduce $\lppg$ to $\llpdf$.

\myparagraph{Longest Length Probability Distribution Computation} In the
\emph{longest length probability distribution computation} problem, we are
given a probabilistic graph $\G$, two vertices $s$ and $t$, and our aim is to
compute, for every integer $\ell$ ranging from $0$ to $|V|-1$, the probability
that the length of the longest (simple) path from $s$ to $t$ is $\ell$. In
addition, we also seek to compute the probability that there does not exist a
path from $s$ to $t$, denoted $\bot$.
Formally:

\begin{center}

\fbox{\begin{tabular}{lp{10cm}}
    \small{PROBLEM} : & $\llpdf$ \hspace{.5cm}(\emph{Longest Length Probability Distribution Computation})
\\
{\small INPUT} : & A probabilistic graph $\G = (V,E,\pi)$, two vertices $s,t\in V$\\
{\small OUTPUT} : &  For every $\ell  \in \{0,1,\dots, |V|-1\} \cup \{\bot\}$
the probability that the longest simple path from $s$ to $t$ is of length $\ell$
\end{tabular}}

\end{center}

It is then clear that we can reduce $\lppg$ to $\llpdf$, in the
same way as we reduced $\lppl$ to $\slpdf$.
Therefore it suffices to explain how to solve $\llpdf$ in XP for probabilistic graphs of bounded treewidth.
To this end, remember the notion of path set projection from
\cref{sec:treewidth}, in particular the set of tuples of vertices $\psp_G(A)$
which we defined for a graph $G=(V,E)$ and set of vertices $A\subseteq V$ (where $A$ will be a bag).
Building from this, we define the notions of \emph{candidate path set projection} and of \emph{longest disjoint
path length function}:

\begin{definition}\label{def:cpsp}
  For $A$ a set of vertices, let $C_A$ be the full 
  clique over~$A$, i.e., the vertices of $C_A$ are $A$ and
  its edges are $\{(u,v) \mid u,v \in A, u\neq v\}$. We then define the set
  of \emph{candidate path set projections of $A$}, written $\cpsp_A$, to be
  $\psp_{C_A}(A)$. 
\end{definition}

In other words, an element of $\cpsp_A$ is simply a set of 
tuples of elements of $A$ such that every $v\in A$ occurs at most once in these tuples; and $\cpsp_A$ contains all such sets.

\begin{definition}\label{def:ldpl}
  Let $G=(V,E)$ be a graph and $A\subseteq V$. The \emph{longest
  disjoint path length function relative to $G$ and $A$}, denoted
  $\ldpl_{G,A}\colon \cpsp_A \to \{0,1,\ldots,|V|-1\}\cup\{\bot\}$, is
  the function such that, for every candidate path set projection
  $\Pi \in \cpsp_A$, the value $\ldpl_{G,A}(\Pi)$ is equal to the
  maximum length of an $A$-vertex-disjoint set of paths $\calP$ of
  $G$ such that $\proj_{G,A}(\calP) = \Pi$ if such a $\calP$
  exists, or to $\bot$ otherwise.
\end{definition}

To solve $\llpdf$ we will then compute, for every bag $S_b$ of the
tree decomposition and every possible function $f\colon \cpsp_{S_b} \to
\{0,1,\ldots,|V|-1\}\cup\{\bot\}$, the probability that in a random
realisation~$H_b$ of~$\calG_b$, the function $f$ is precisely equal
to $\ldpl_{H_b,S_b}$. Before we explain this, let us
show that this suffices to solve $\llpdf$. By applying
this to the root bag $S_{b_{\text{root}}}=\{s,t\}$, and noticing that
$G_{b_{\text{root}}} = G$, we
obtain, for every function $f\colon \cpsp_{S_{b_{\text{root}}}} \to
\{0,1,\ldots,|V|-1\}\cup\{\bot\}$, the probability that in a random
realisation~$H$ of $G$, the function $f$ equals $\ldpl_{H,S_{b_{\text{root}}}}$;
let us call this value $T[f]$. Now, let $\Pi$ be the specific path
set projection $\Pi = \{(s,t)\}$ (or $\Pi = \{(s)\}$ in case $s=t$), which is in
$\cpsp_{S_{b_{\text{root}}}}$. Then,
for $\ell \in \{0,1,\ldots,|V|-1\} \cup \{\bot\}$, the probability
that the longest path from $s$ to $t$ in $\calG$ has length exactly~$\ell$
is equal to $\sum_g T[g]$, where $g$ ranges over all functions that
map $\Pi$ to $\ell$.
This allows us to solve $\llpdf$, which we already argued
is enough to solve $\lppg$.
Thus, all that is needed to finish our XP algorithm for
$\lppg$ over bounded-treewidth probabilistic graphs 
is to show:

\begin{toappendix}

Having shown the correctness of \cref{induction1}, we now turn to the proof of
\cref{claim:induction2}, which was established in \cref{sub:tractmore} for
showing the tractability of $\lppg$ in probabilistic graphs of bounded
treewidth. Like in the case of $\lppl$ we provide an algorithm (\Cref{alg:TreeDPlongpath}) which computes, for every bag $S_b$ in the tree decomposition, the probability distribution over the longest disjoint path length functions of $S_b$ in $\G_b$. 
In other words, it computes, for every $f\colon \CPSP_{S_b} \to \{0,1, \ldots, |V|-1\}\cup\{\bot\}$,
the probability that for a random realisation $H_b$ of $\G_b$, 
we have $f=\ldpl_{H_b,S_b}$. Further, like in the case of $\lppl$, the most interesting case of \Cref{alg:TreeDPlongpath} is that of the $\JoinNode{}$ node for which it uses the notion of an \emph{extension} of a candidate path set projection~$\Pi$, which we define next.

\begin{definition}
Let $A$ be a finite set of vertices. Consider any $\Pi \in \CPSP_A$, where
$\CPSP_A$ is the set of candidate path set projections of $A$ (recall
\cref{def:cpsp}). Let $\Pi= \{(v_{11},v_{12},\ldots, v_{1k_{1}}),\allowbreak
(v_{21},v_{22},\ldots, v_{2k_{2}}), \ldots, (v_{\lambda 1},v_{\lambda 2},\ldots,
v_{\lambda k_{\lambda}})\}$, thus $\Pi$ is a collection of $\lambda$ disjoint tuples. Let $\Pi'$ be a set such that $\Pi'$ satisfies the following property:

   For each non-singleton tuple $(v_{1},v_{2}, \dots, v_{k}) \in \Pi$, the set $\Pi'$ contains a collection of tuples formed by partitioning the index set $[k]$ into consecutive intervals and restricting $(v_{1},v_{2}, \dots, v_{k}) \in \Pi$ to each tuple, with the endpoints overlapping.

More formally, we keep each singleton tuple as-is, and we replace every non-singleton tuple $(v_1, v_2 \dots, v_k)$ by a finite sequence of tuples\\
$(\mathbf v^{(1)}, \mathbf v^{(2)}, \dots, \mathbf v^{(m)})$
for which there exist indices
\[
1 = a_1 < b_1 = a_2 < b_2 = a_3 < \cdots < b_{m-1} = a_m < b_m = k
\]
such that
\[
\mathbf v^{(j)} = (v_{a_j}, v_{a_j+1}, \dots, v_{b_j})
\quad \text{for all } j = 1,\dots,m.
\]

 Then we say that $\Pi'$ is an extension of $\Pi$. Further, we define $E(\Pi)
 \coloneq \{ \Pi' \mid \Pi' \text{ is an extension}\\ \text{ of } \Pi \}$ as the set of all possible extensions of $\Pi$.
\end{definition} 

\begin{example}
 For example let us consider $\Pi_0 = \{(v_1,v_2, v_3, v_4, v_5), (v_6,
 v_7,v_8)\}$. Then, by replacing $(v_1,v_2, v_3, v_4, v_5)$ with
 $(v_1,v_2),(v_2,v_3,v_4),(v_4,v_5)$ and replacing $(v_6, v_7,v_8)$ with
 $(v_6,v_7),(v_7,v_8)$, we get that
 $\Pi_0'=\{(v_1,v_2),(v_2,v_3,v_4),(v_4,v_5),(v_6,v_7),(v_7,v_8)\}$ is an
 extension of $\Pi_0$.
\end{example}

We now show the claim corresponding to the case of \textsc{Join} nodes (note the similarity with \cref{changeclaim}):
 
\begin{claim}\label{changeclaimlong}
  Let $\calG = (V,E,\pi)$ be a probabilistic graph with 
  underlying
  non-probabilistic graph $G=(V,E)$,
  let $\calS$ be a nice tree decomposition of $\calG$, and 
  let $b$ be a \textsc{Join} node with children $b_1$ and $b_2$.
  Note that $S_b=S_{b_1}=S_{b_2}$. Recall that we denote by $\mathcal{G}_{b_1}$ and $\mathcal{G}_{b_2}$ the probabilistic subgraphs induced by the subtrees rooted at $b_1$ and $b_2$ respectively. Let $G_1=(V_1,E_1)$ and $G_2=(V_2,E_2)$ be realisations of $\mathcal{G}_{b_1}$ and $\mathcal{G}_{b_2}$ respectively. Define $H_b = (V_1 \cup V_2, E_1 \cup E_2)$ to be the graph formed by taking the union of $G_1$ and $G_2$.
  Let $f_1=\ldpl_{G_1,S_{b_1}}$ and $f_2=\ldpl_{G_2,S_{b_2}}$(recall \cref{def:ldpl}). Then the longest disjoint path length function for $S_b$ in $H_b$, $\ldpl_{H_b,S_{b}}$ is given by $f\colon \CPSP_{S_b} \to \{0,1, \ldots, |V|-1\}\cup\{\bot\}$ where $f(\Pi)\coloneq \max_{\substack{\Pi_1, \Pi_2 \in \cpsp_{S_b} \\ \Pi_1 \cup \Pi_2 \in E(\Pi),\; \Pi_1 \cap \Pi_2 = \emptyset}} \left[ f_1(\Pi_1) + f_2(\Pi_2) \right]$ for each $\Pi \in \CPSP_{S_b}$,
  where $\bot$ in arithmetic expressions is evaluated as $-\infty$ (i.e., $\bot+x = \bot$ for each $x \in \NN \cup \{\bot\}$, and $\bot$ is smaller than all values of~$\NN$ so it ignored in the $\max$ unless all arguments are $\bot$ in which case the $\max$ evaluates to~$\bot$).
\end{claim}

\begin{proof}

  Consider any set $\Pi \in \CPSP_{S_b}$ and let $\mathcal{P}$ be a set of $S_b$-vertex-disjoint paths in $H_b$, such that $\proj_{H_b,S_b}(\calP) = \Pi$.

  Next for each path $P \in \mathcal{P}$ we break it into multiple subpaths which consist of contiguous edges which lie entirely in $G_{1}$ or $G_{2}$. By doing so we obtain two sets $\mathcal{P}_1$ and $\mathcal{P}_2$ such that the set $\mathcal{P}_1$ corresponds to a $S_b$-vertex-disjoint set of paths in $G_1$ and the set $\mathcal{P}_2$ corresponds to a $S_b$-vertex-disjoint set of paths in $G_{2}$. Now we construct two sets $\Pi_1^{*}$ and $\Pi_2^{*}$ such that the projections of $P_1$ and $P_2$ are $\Pi_{1}^{*}$ and $\Pi_2^{*}$ respectively, i.e., $\proj_{G_1,S_b}(\calP_{1}) = \Pi_{1}^{*}$ and $\proj_{G_2,S_b}(\calP_{2}) = \Pi_{2}^{*}$.

  First we construct a set $\Pi'$ which is an extension of $\Pi$. For each $P \in \mathcal{P}$ consider $\proj_{H_b,S_b}(P)$, the projection of $P$ on $S_b$. For every $P \in \calP$, we obtain a projection $\proj_{H_b,S_b}(P) \in \Pi$ which we replace with a collection of tuples in $\Pi'$ as follows: Suppose that $\proj_{H_b,S_b}(P)$ is given by  $(v_{1},v_{2}, \dots, v_{k})$. Let $v_{a_1}=v_{1}, v_{a_2} \dots, v_{a_{m-1}}, v_{a_m},v_{k}$ denote the change-over vertices (see \cref{def:changeoververtices})
  of $P$ in $H_b$, then in $\Pi'$ we break each tuple at the locations of the change-over vertices. (Note that all change-over vertices must lie in $S_b$ due to \cref{changeclaim}, hence they must occur in the projection of $P$ on $S_b$). Therefore, $\Pi'$ contains a collection of disjoint tuples formed by partitioning the index set $[k]$ into consecutive intervals and restricting $(v_{1},v_{2}, \dots, v_{k}) \in \Pi$ to each tuple, with the endpoints overlapping. By repeating this procedure for all $P \in \mathcal{P}$, i.e., for all tuples in $\Pi$, we obtain $\Pi'$.

Now note that since $v_{a_1}=v_{1}, v_{a_2} \dots, v_{a_{m-1}}, v_{a_m},v_{k}$ are the change-over vertices of $P$ in $H_b$, every tuple $v^{(j)} \text{ for all } j = 1,\dots,m$ corresponds to the projection of a path in either $G_{1}$ or $G_{2}$, and concatenation of these paths gives $P$.

 Thus we partition $\Pi'$ into two sets such that for every $j$, if the tuple $v^{(j)}$  corresponds to a path in $G_{1}$ then we add $v^{(j)}$ to $\Pi_1^{*}$, else if the tuple $v^{(j)}$  corresponds to a path in $G_{2}$ then we add $v^{(j)}$ to $\Pi_2^{*}$. By construction, note that $\Pi'$ is an extension of $\Pi$, and that we have $\Pi_1^{*} \cup \Pi_2^{*}= \Pi'$, further it is easy to verify that $\proj_{G_1,S_b}(\calP_{1}) = \Pi_{1}^{*}$ and $\proj_{G_2,S_b}(\calP_{2}) = \Pi_{2}^{*}$.
  
  Note that the length of $\mathcal{P}$ in $H_b$ is same as the sum of the lengths of $\mathcal{P}_1$ in $G_{1}$ and $\mathcal{P}_2$ in $G_{2}$ which are bounded above by $f_1(\Pi_1^{*})$ and $f_2(\Pi_2^{*})$ respectively.

  Hence the following inequality holds: \\
  Length of $\mathcal{P}$ $\leq f_1(\Pi_1^{*})+f_2(\Pi_2^{*}) \leq \max_{\substack{\Pi_1, \Pi_2 \in \cpsp_{S_b} \\ \Pi_1 \cup \Pi_2 \in E(\Pi),\; \Pi_1 \cap \Pi_2 = \emptyset}} \left[ f_1(\Pi_1) + f_2(\Pi_2) \right]$.

  In the case that $\mathcal{P}$ corresponds to the longest disjoint path length of $\Pi$, i.e., it has the greatest length among all possible sets of $S_b$-vertex disjoint paths in $H_b$ which project to $\Pi$, the length of $\mathcal{P}$ must equal $f(\Pi)$ by definition, further, $\mathcal{P}_1$ and $\mathcal{P}_2$ must also have the maximum length among all $S_b$-vertex disjoint paths which project to $\Pi_1$ and $\Pi_2$ in the graphs $G_{1}$ and $G_{2}$ respectively, so their lengths must equal $f_1(\Pi_1^{*})$ and $f_2(\Pi_2^{*})$. Therefore, the first inequality is tight and is satisfied with equality. Further among all partitions $(\Pi_1,\Pi_2)$ of an extension of $\Pi$, setting $\Pi_1\coloneq \Pi_1^{*}$ and $\Pi_2\coloneq \Pi_2^{*}$ (their union being the extension $\Pi'$) must maximise $f_1(\Pi_1)+f_2(\Pi_2)$, else it would contradict the fact that $\mathcal{P}$ corresponds to the longest disjoint path length of $\Pi$. Hence the second inequality must hold with equality as well.

  Hence, we obtain $f(\Pi)=\max_{\substack{\Pi_1, \Pi_2 \in \cpsp_{S_b} \\ \Pi_1 \cup \Pi_2 \in E(\Pi),\; \Pi_1 \cap \Pi_2 = \emptyset}} \left[ f_1(\Pi_1) + f_2(\Pi_2) \right]$.
\end{proof}

\begin{algorithm}
\DontPrintSemicolon
\SetKwProg{Fn}{Function}{:}{}
\SetKwFunction{TreeDP}{TreeDP}

\KwIn{A probabilistic graph $G=(V,E,\pi \colon E \to [0,1])$ of treewidth $k$  and a nice tree decomposition $\mathcal{S}$ with root $b_\mathsf{root}$ }
\KwOut{A table $T$ indexed with all possible functions $
\CPSP_{S_{b_\mathsf{root}}} \to \{0,1, \ldots, |V|-1\} \cup \{\bot\}$ , such that $T[f]$ is the probability that for any random realisation $H_b$ of $\G$, we have $f=\ldpl_{H_b,S_{b_\text{root}}}$}

\Fn{\TreeDP{$b$}}{  
  $T_b \gets$ \text{for each} $f\colon \CPSP_{S_b} \to  \{0,1, \ldots, |V|-1\}\cup\{\bot\}$, \text{set}  $T_b[f] =0$\;
    \If{$b$ is a \LeafNode{} node}{
    Let $f_{\emptyset}$ be the function on
$\CPSP_{\emptyset}=\{\emptyset\}$ defined by
$f_{\emptyset}(\emptyset)=0$\;
$T_b[f_{\emptyset}] \gets 1$\;
      
      \Return{$T_b$}
    }

    \ElseIf{$b$ is a \ForgetVertNode{} node with child $b'$}{
      $T_{b'} \gets$ \TreeDP{$b'$}\;
      \ForEach{$f' \colon \CPSP_{S_{b'}} \to \{0,1, \ldots, |V|-1\} \cup \{\bot\}$}{
        Define $f$ by $f(\Pi) \coloneq \max_{\substack{\Pi' \in \CPSP_{S_{b'}}\\\Pi' \text{ projects to }
        \Pi}} f'(\Pi')$\;
            $T_b[f] \gets T_b[f] + T_{b'}[f']$\;
        }
        \Return{$T_b$}
    }

    \ElseIf{$b$ is an \IntroVertNode{} node with child $b'$ introducing vertex $v_0$}{
        $T_{b'} \gets$ \TreeDP{$b'$}\;
        \ForEach{$f'\colon \CPSP_{S_{b'}} \to  \{0,1, \ldots, |V|-1\}\cup \{\bot\}$}{
          Define a function $f \colon \CPSP_{S_{b}} \to  \{0,1, \ldots, |V|-1\}\cup \{\bot\}$\;
          \ForEach{$\Pi \in \CPSP_{S_b}$}{
          $f(\Pi) \coloneq 
    \begin{cases}
    \bot & \text{if $\exists$ a 
    non-singleton 
    tuple } t \in \Pi \text{ with } v_0 \in t\\
    f'(\Pi') & \text{otherwise, for $\Pi'$ the restriction of $\Pi$ to $S_{b'}$}
    \end{cases}$\;

    }
    $T_b[f] \gets T_b[f] + T_{b'}[f']$\;
        }
        \Return{$T_b$}
    }

    \ElseIf{$b$ is a \JoinNode{} node with children $S_{b_1}, S_{b_2}$}{
        $T_{b_1} \gets$ \TreeDP{$b_1$}\;
        $T_{b_2} \gets$ \TreeDP{$b_2$}\;
        $T_b \gets \texttt{JoinComputation$\lppg$}(T_{b_1}, T_{b_2}, S_b)$\;
        \Return{$T_b$}
    }

    \ElseIf{$b$ is an \IntroEdgeNode{} node with child $S_{b'}$ introducing edge $(u_0,v_0)$}{
        $T_{b'} \gets$ \TreeDP{$b'$}\;
        Set $f_1$ as
        $f_1(\Pi) \!\coloneq\!
        \begin{cases}
          1\!\! & \!\!\!\text{if } \Pi \text{ contains $(u_0,v_0)$ and no other 
          non-singleton tuple } \\
            0\!\! & \!\!\!\text{if } \Pi \text{ contains only singleton tuples } \\
        \bot\!\! & \!\!\!\text{ otherwise}
        \end{cases}$\\

        Define $f_2$ by 
        $f_2(\Pi) \coloneq
        \begin{cases}
          0 & \text{if } \Pi \text{ contains only singleton tuples } \\
        \bot & \text{otherwise}
        \end{cases}$\\
        \[
T[f] :=
\pi(u_0,v_0)\,\mathbf{1}_{\{f=f_1\}}
+
\bigl(1-\pi(u_0,v_0)\bigr)\,\mathbf{1}_{\{f=f_2\}}.
\]
        $T_b \gets \texttt{JoinComputation$\lppg$}(T_{b'}, T, S_b)$\;
        \Return{$T_b$}
    }

}
\Return{$\TreeDP(b_\text{root})$}
\caption{\texttt{DPComputation$\lppg$($G=(V,E,\pi),\mathcal{S})$}: Compute the
probability distribution on the longest disjoint path length functions along a tree decomposition}
\label{alg:TreeDPlongpath}
\end{algorithm}

\begin{algorithm}
  \KwIn{A set of vertices $S_b$ and two tables $T_1$, $T_2$ indexed with all possible functions $\CPSP_{S_b} \to \{0,1, \ldots, |V|-1\} \cup \{\bot\} $ 
  }
  \KwOut{A table $T$ such that for any function $f\colon \CPSP_{S_b} \to  \{0,1, \ldots, |V|-1\}\cup \{\bot\}$, the value $T[f]$ is the probability that in a random realisation $H_b$ of $\G_{b}$, we have $\ldpl_{H_b,S_b}= f$}

\text{For each} $f\colon \CPSP_{S_b} \to  \{0,1, \ldots, |V|-1\}\cup \{\bot\}$, \text{set}  $T[f] \gets 0$\;

\For{\text{each pair of functions} $(f_1,f_2)$, $f_1,f_2 \colon \CPSP_{S_b} \to  \{0,1, \ldots, |V|-1\}\cup\{\bot\}$}{

\ForEach{$\Pi \in \CPSP_{S_b} $}{
            Define $f \colon \CPSP_{S_b} \to \{0,1, \ldots, |V|-1\}\cup\{\bot\}$ as follows:\\
            $ f(\Pi) =\max_{\substack{\Pi_1, \Pi_2 \in \cpsp_{S_b} \\ \Pi_1 \cup \Pi_2 \in E(\Pi),\; \Pi_1 \cap \Pi_2 = \emptyset}} \left[ f_1(\Pi_1) + f_2(\Pi_2) \right]$\;
        }
$T[f] \gets T[f] + T_1[f_1] \times T_2[f_2]$
}
\Return($T$)
\caption{{\texttt{JoinComputation$\lppg$$(T_1,T_2,S_b)$}}}\label{alg:joinlongpath}
\end{algorithm}

We are now ready to prove \cref{claim:induction2}, whose statement we recall here:
\end{toappendix}

\begin{claimrep}
  \label{claim:induction2} We can compute in XP parameterised by
  the treewidth of $\calG$, for every bag~$S_b$, a table~$T_b$
  indexed with all possible functions $f\colon \cpsp_{S_b} \to
  \{0,1,\ldots,|V|-1\}\cup\{\bot\}$, such that $T_b[f]$ equals the
  probability that in a random realisation $H_b$ of $\calG_b$, $f$ is precisely $\ldpl_{H_b,S_b}$. 
\end{claimrep}
\begin{proof}

We prove that this is what \Cref{alg:TreeDPlongpath} does, by structural
induction on the nice tree decomposition. We show that the algorithm applied to
a bag~$b$ indeed computes, for every bag~$S_b$, a table $T_b$
  indexed with all possible functions $f\colon \cpsp_{S_b} \to \{0,1,\ldots,|V|-1\}\cup\{\bot\}$, such that $T_b[f]$ equals the
  probability that in a random realisation $H_b$ of $\calG_b$, 
  the function $f$ is precisely $\ldpl_{H_b,S_b}$. 
Clearly for a leaf node, the claim
holds, since leaf nodes contain no
vertices and
longest disjoint path length function over the induced empty subgraph is
the function $f_{\emptyset}$ with $f_{\emptyset}(\emptyset)=0$, which is
what our algorithm assigns. Therefore the base case
holds. Assume that the claim holds for the
children of $b$. We now show that it holds for
$b$ as well:

\begin{description}
\item[Case 1: $b$ is a \ForgetVertNode{} node]

Note that our definition of the subgraph induced at a
node $b$ covers all the edges introduced until that
point and all the vertices (in the bag or any of its
descendants). Hence if we let the child of $b$ be $b'$, then
$\G_b = \G_{b'}$. Therefore for any realisation $H$ of $\G_{b'}$
such that the $\ldpl_{H,S_{b'}}$ is given by~$f$,
we have that in $\ldpl_{H,S_b}$ we have the choice for every candidate path
set projection $\Pi$ over~$S_b$ of either achieving the same candidate path
set projection $\Pi$ over~$S_{b'}$ or of achieving a candidate path set projection 
$\Pi'$ obtained by inserting the projected-away vertex of
$S_{b'} \setminus S_b$ as a singleton tuple or as an internal vertex of one tuple
(i.e., inside one tuple but not as the first or last element). This is what we
mean in the algorithm when we say that $\Pi'$ \emph{projects} to~$\Pi$; in
particular $\Pi$ projects to itself. The maximal
possible length of an $S_b$-vertex-disjoint set of paths realising $\Pi$ is then
the maximum taken across these possibilities.
Therefore 
the correctness of $T_b$ follows directly from the correctness of $T_{b'}$ obtained by induction hypothesis.

  \item[Case 2: $b$ is an \IntroVertNode{} node.]
Nodes of this kind do not introduce any edges, and they
only add an isolated vertex to the associated
subgraph: they also have precisely one child. Let the child of $b$ be $b'$; then
$\G_b$ differs from $\G_{b'}$ only by the addition of an isolated vertex, say $v_0$. 
We are given that the probability that a random realisation of $\G_{b'}$ admits the longest disjoint path length function $f$
is $T_{b'}[f]$, where $T_{b'}$ is the table obtained as a result of the function call $\texttt{TreeDP}(b')$, the correctness of which is guaranteed using our induction hypothesis.
Therefore for any
realisation $H_{b'}$
of $\G_{b'}$ such that $\ldpl_{H_{b'},S_{b'}}=f'$, we can map it to a 
realisation $H_b$ of $\G_b$ occurring with the same probability, differing only by the addition of a single isolated vertex $v_0$. Note that such a map from realisations of $\G_{b'}$ to realisations of $\G_b$ is surjective. Given that $\ldpl_{H_{b'},S_{b'}}=f'$, we may construct the function $f$ such that $\ldpl_{H_{b},S_{b}}=f$ in the following fashion: If the candidate path set projection $\Pi$ contains a tuple $(v_1,v_2, \dots, v_k)$ 
which is a non-singleton tuple (i.e., $k>1$) and where we have $v_i=v_0$ for some $i \in \{1,2, \dots, k\}$
then the longest disjoint path length of $\Pi$ is $\bot$,
since $v_0$ is isolated. 
Otherwise, we use the image of the function $f'$ on the candidate path set projection $\Pi'$ obtained from $\Pi$ by restricting it to vertices of $S_{b'}$, i.e., removing the singleton tuple $v_0$ if it exists (noting that this has no effect on the maximal length of paths realising the candidate path set projection).

\item[Case 3: $b$ is a \JoinNode{} node]
Let $b$ be a \JoinNode{} node with children $b_1$ and $b_2$.
Consider realisations of $\G_{b_1}$ and $\G_{b_2}$, say
$G_1 = (V_1,E_1)$ and $G_2 = (V_2,E_2)$ respectively, such that the
functions $f_1 = \ldpl_{G_1,S_{b_1}}$ and $f_2=\ldpl_{G_2,S_{b_2}}$
are known.
Consider the union of $G_1$ and $G_2$ which we define by $H_b=(V_1 \cup V_2, E_1 \cup E_2)$. Now, note that $H_b=(V_1 \cup V_2, E_1 \cup E_2)$ is a realisation of $\G_{b}$, the probabilistic graph associated with the node $b$. By \Cref{changeclaimlong}, 
we know how to express the function 
$\ldpl_{H_{b},S_{b}}$, the longest disjoint path length function of $S_b$ in the
graph $H_b=(V_1 \cup V_2, E_1 \cup E_2)$,
from the functions 
$\ldpl_{H_{b_1},S_b}$
and
$\ldpl_{H_{b_2},S_b}$.
Thus we can map pairs of realisations of $\G_{b_1}$ and $\G_{b_2}$ to
realisations of $\G_{b}$,
with the LDPL of the realisation of $\G_b$ being determined from that of
$\G_{b_1}$ and $\G_{b_2}$.
The resulting LDPL is computed by \Cref{alg:joinlongpath} as prescribed
by \Cref{changeclaimlong}.

\item[Case 4: $b$ is an \IntroEdgeNode{} node]
\Cref{alg:TreeDPlongpath} treats the addition of the new edge
from $u_0$ to $v_0$ as taking a union of the graph
under consideration with another graph with the same vertex set and with a single
edge from $u_0$ to $v_0$, all other vertices being
isolated. For the latter graph, the probability
distribution over longest disjoint path length functions of $S_b$ is
trivially known and then we use
\Cref{alg:joinlongpath} 
like for the case of a \JoinNode{} node.
\qedhere
\end{description}
\end{proof}

We show this in Appendix~\ref{apx:proofatleastlengthsimplepaths}, relying again on the properties of nice tree
decompositions to allow us to sum over mutually exclusive events and multiply
independent events.

%% file: untract.tex
In the previous section, we have shown that the problems $\genprob$
for ${\bowtie} \in \{{\leq},
{\geq}\}$ and $X \in \{\Walk, \Path\}$ are all tractable on probabilistic
graphs of bounded treewidth. This was shown by designing dynamic
programming algorithms that deal with partial solutions for every
bag of the tree decomposition, for various notions of partial
solution. At first glance then, one might think that using the
same kind of techniques would get us the tractability
of $\lppe$ and $\lpwe$. We show in this section that, 
somewhat unexpectedly, these two remaining problems can be 
\#P-hard on bounded-treewidth probabilistic graphs, in fact already on
pathwidth-2 probabilistic graphs. (Remember that, by contrast, they are
tractable on treewidth-1 graphs by \cref{thm:forest}.) 
To prove this, we will reduce again from \#PP2DNF, recalling its definition from
earlier (\cref{def:pp2dnf}). Formally:

\begin{toappendix}
  \label{apx:untract}
\end{toappendix}
\begin{theoremrep}
\label{thm:flpp}
  $\lppe$ and $\lpwe$ are $\#P$-hard on 
  probabilistic graphs of pathwidth~$2$.
\end{theoremrep}
\begin{proofsketch}
We show this for $\lppe$ and observe that the same reduction works for $\lpwe$. 
Given an instance $\varphi$ of $\spptdnf$ 
on variables $X = \{X_1,\ldots, X_{p}\}$, $Y = \{Y_1,\ldots, Y_{q}\}$ with $k$ terms of the type $X_{i_r} \land Y_{j_r}$, where $1 \leq r \leq k$,
  we 
  construct the probabilistic graph $\G=(V,E,\pi)$ described below
  (and illustrated in \Cref{fig:example-graph2} in the appendix).

\begin{itemize}
    \item Create a start vertex $s$.
    \item Create vertices
      $u_1,u_2\ldots, u_{p}$ which will respectively correspond to the variables
      $X_1,X_2\ldots, X_{p}$, and create vertices $v_1, \ldots, v_q$ which will
      respectively correspond to the variables $Y_1, \ldots, Y_q$.
  \item  Create edges with probability~$1$ 
    from $u_1$ to $u_2$, from $u_2$ to~$u_3$, and so on, with the last such edge going from $u_{p-1}$ to~$u_p$.
    \item 
      Create a vertex $u$ and an edge of probability~$1$ from~$u_p$ to~$u$.
    \item   Create paths of $2p$ edges with probability 1 
      from $v_1$ to $v_2$, from $v_2$ to~$v_3$, and so on, 
      with the last such path going from $v_{q-1}$ to $v_q$.
    
    \item For each vertex $u_i$ with $i \in \{1,\dots p\}$, create an edge from
      $s$ to $u_i$ having probability $\frac{1}{2}$.
    \item For each vertex $v_j$ with $j \in \{1,\dots q\}$, create
      an edge from $u$ to $v_j$ having probability $\frac{1}{2}$.

    \item Last, create two vertices $v$ and $t$, and an edge from $v_q$
      to~$v$
      having probability $1$. Now, for each term $X_{i_r} \land Y_{j_r}$
      of the formula with $1 \leq r \leq k$, construct
a path from~$v$ to~$t$ of edges with probability~$1$ having length $(2p)j_r+i_r$. \end{itemize}

Clearly $\calG$ can be constructed in polynomial time.
It is not hard to show that it has treewidth~$2$, and a small modification of
the construction given in the appendix ensures that it has pathwidth~$2$.
We now use the oracle for the $\lppe$ problem on $\calG$, to compute the probability $\alpha$ that there is a path of length
exactly $2pq+p+4$ from~$s$ to~$t$.
We then return
the value $\alpha \times 2^n$, for $n\coloneq p+q$ the number of variables
of~$\varphi$.  The proof that this reduction is correct, i.e.,
that $\alpha= \frac{\#\varphi}{2^{n}}$ for~$\#\varphi$ the number of satisfying
assignments of~$\varphi$, can be found in Appendix~\ref{apx:untract}.
\end{proofsketch}

\begin{toappendix}
\begin{figure}
\centering
  \begin{tikzpicture}[xscale=1.1]
\begin{scope}[every node/.style={circle,thick,draw}]
    \node (s) at (0,0) {$s$};
    
    \node (t) at (12,0) {$u_p$};
    \node (n2) at (4,0) {$u_1$};
    \node (n3) at (6,0) {$u_2$};
    \node (n4) at (8,0) {$u_3$};
    \node[draw=none] (ucdots) at (9,0) {$\cdots$};
    \node (n5) at (10,0) {$u_{p-1}$};

    \node (t1) at (12,-5) {$u$};
    \node (t5) at (2,-5) {$v_{q-1}$};
    \node[draw=none] (vcdots) at (3,-5) {$\cdots$};
    \node (t4) at (4,-5) {$v_3$};
    \node (t3) at (6,-5) {$v_2$};
    \node (t2) at (8,-5) {$v_1$};
    \node (t6) at (0,-5) {$v_{q}$};

    \node (v1) at (0,-10) {$v$};
    \node (v2) at (12,-10) {$t$};
    \node[draw=none] (clausedots) at (4,-9.5) {$\vdots$};
    \node[draw=none] (clausedots2) at (8,-9.5) {$\vdots$};
\end{scope}
\begin{scope}[>={Stealth[black]},
              every node/.style={fill=white,circle},
              every edge/.style={draw,thick}]
              
              \path [->, thick, dotted, line cap=round, draw=blue] (n2) edge node {$1$} (n3);
              \path [->, thick, dotted, line cap=round, draw=blue] (n2) edge node {$1$} (n3);
              \path [->, thick, dotted, line cap=round, draw=blue] (n3) edge node {$1$} (n4);
              
              \path [->] (n5) edge node {$1$} (t);
              \path [->, thick, dotted, line cap=round, draw= blue] (s) edge[bend left=60] node {$1$} (n2);
              \path [->, thick, dotted, line cap=round, draw= blue] (s) edge[bend left=60] node {$1$} (n3);
              \path [->, , thick, dotted, line cap=round, draw= blue] (s) edge[bend left=60] node {$1$} (n3);
              \path [->, , thick, dotted, line cap=round, draw= blue] (s) edge[bend left=60] node {$1$} (n4);
              \path [->, , thick, dotted, line cap=round, draw= blue] (s) edge[bend left=60] node {$1$} (n5);
              \path [->, , thick, dotted, line cap=round, draw= blue] (s) edge[bend left=60] node {$1$} (t);
              \path [->] (n2) edge node {$1$} (n3);
              \path [->] (n3) edge node {$1$} (n4);

              \path [->] (t) edge node {$1$} (t1);
              \path [->] (t2) edge node {$2p$} (t3);
              \path [->] (t3) edge node {$2p$} (t4);
              \path [->] (t5) edge node {$2p$} (t6);

              \path [->, thick, dotted, line cap=round, draw= blue] (t1) edge[bend right=60] node {$1$} (t2);
              \path [->, thick, dotted, line cap=round, draw= blue] (t1) edge[bend right=60] node {$1$} (t3);
              \path [->, thick, dotted, line cap=round, draw= blue] (t1) edge[bend right=60] node {$1$} (t4);
              \path [->, thick, dotted, line cap=round, draw= blue] (t1) edge[bend right=60] node {$1$} (t5);
              \path [->, thick, dotted, line cap=round, draw= blue] (t1) edge[bend right=60] node {$1$} (t6);

               \path [->, thick, dotted, line cap=round, draw= blue] (t1) edge[bend right=60] node {$1$} (t6);

              \path [->] (v1) edge[bend left =50] node
              {\shortstack{$2pj_1+i_1$ \\ \tiny $X_{i_1} \land Y_{j_1}$}} (v2);
              \path [->] (v1) edge[bend left =25] node {\shortstack{$2pj_2+i_2$ \\ \tiny $X_{i_2}\land Y_{j_2}$}} (v2);
              \path [->] (v1) edge[bend right = 15] node
              {\shortstack{$2pj_{k}+i_k$ \\ \tiny $X_{i_k}\land Y_{j_k}$}}  (v2);

              \path [->] (t6) edge node {\shortstack{$1$ \\ \tiny ${}$}} (v1);
\end{scope}
\end{tikzpicture}
\caption{Mapping an instance of $\spptdnf$ to a probabilistic
(unweighted) graph. The edges stand for paths whose intermediate vertices are
omitted for simplicity, and the number indicated in
the annotation stands for the length of that path (so edges labelled 1
correspond to a single edge).
For the solid edges in black, all edges of the path occur
with a probability of $1$; the blue dotted edges are always labelled
with~$1$ and they occur with a probability of $\frac{1}{2}$.}
    \label{fig:example-graph2}
\end{figure}
\end{toappendix}
\begin{proof}
We show this for $\lppe$ and observe that the same reduction extends to $\lpwe$.
We will give the argument with a graph of treewidth~$2$ (because the graph is simpler to
construct and to draw), and explain at the end of the proof how the graph can be modified
slightly to ensure that it has pathwidth~$2$.

We are given an instance $\varphi$ of $\spptdnf$ built on the set of variables
$X = \{X_1,\ldots, X_{p}\}$, $Y = \{Y_1,\ldots, Y_{q}\}$ with $k$ terms of the
form $X_{i_r} \land Y_{j_r}$, where $1 \leq r \leq k$.
We let $n \coloneq p+q$ be the number of variables of~$\varphi$, and write
$\#\varphi$ for the number of satisfying assignments of~$\varphi$, which is what
we wish to compute.
  We construct the probabilistic graph $\G=(V,E,\pi)$ described below
  (and illustrated in \Cref{fig:example-graph2}).

\begin{itemize}
    \item Create a start vertex $s$.
    \item Create vertices
      $u_1,u_2\ldots, u_{p}$ which will respectively correspond to the variables
      $X_1,X_2\ldots, X_{p}$, and create vertices $v_1, \ldots, v_q$ which will
      respectively correspond to the variables $Y_1, \ldots, Y_q$.
  \item  Create edges with probability~$1$ 
    from $u_1$ to $u_2$, from $u_2$ to~$u_3$, and so on, with the last such edge going from $u_{p-1}$ to~$u_p$.
    \item 
      Create a vertex $u$ and an edge of probability~$1$ from~$u_p$ to~$u$.
    \item   Create paths of $2p$ edges with probability 1 
      from $v_1$ to $v_2$, from $v_2$ to~$v_3$, and so on, 
      with the last such path going from $v_{q-1}$ to $v_q$.
    
    \item For each vertex $u_i$ with $i \in \{1,\dots p\}$, create an edge from
      $s$ to $u_i$ having probability $\frac{1}{2}$.
    \item For each vertex $v_j$ with $j \in \{1,\dots q\}$, create
      an edge from $u$ to $v_j$ having probability $\frac{1}{2}$.

    \item Last, create two vertices $v$ and $t$, and an edge from $v_q$ 
      to~$v$
      having probability $1$. Now, for each term $X_{i_r} \land Y_{j_r}$
      of the formula with $1 \leq r \leq k$, construct
a path from~$v$ to~$t$ of edges with probability~$1$
having length $(2p)j_r+i_r$. \end{itemize}

Observe that the graph thus constructed has treewidth at most 2. Indeed, the
part of the graph comprising the vertex $s$ and the vertices $u_i$ 
with $i \in \{1, \ldots, p\}$ admits a path decomposition of width 2, as does
the part of the graph comprising the vertices $u$ and the vertices $v_j$ with $j \in
\{1, \ldots, q\}$. To see this observe that we can define the node $b_i$ with the corresponding bag $S_{b_i}$ to contain the vertices $\{s,u_i,u_{i+1}\}$. Connecting these nodes linearly to form a path $b_1,b_2 \dots b_{p-1}$ gives a valid tree decomposition of the part of $\G$ comprising $s$ and $u_i$
with $i \in \{1,\ldots, p\}$. A similar construction works for the  part of the graph comprising the vertex $u$ and the vertices $v_j$ with $j \in \{1, \ldots, q\}$. Likewise, the part of the graph comprising $v$, $t$ that encodes the conjunctive terms has a tree decomposition of width $2$. We construct a tree decomposition as follows. Create a central node with bag
$\{v,t\}$. For every $v$--$t$ path of the form
$v, v_1, v_2, \dots, v_m, t$ we build a path decomposition by introducing nodes $b_1,\dots,b_m$, where for $i=1,\dots,m-1$ the bag of $b_i$ is
$S_{b_i} = \{v, v_i, v_{i+1}\}$
and the bag of $b_m$ is $S_{b_m} = \{v, v_m, t\}$.
These nodes are connected linearly to form a path, and the node $b_m$ is attached to the central node with bag $\{v,t\}$. Repeating this construction for every $v$--$t$ path yields a tree decomposition in which every bag has size at most $3$. Therefore, it may be concluded that the constructed graph has treewidth at most $2$.

We now use the oracle for the $\lppe$ problem for the graph $G=(V,E,\pi)$ thus
constructed, to compute the probability $\alpha$ that there is a path of length
exactly $2pq+p+4$ from~$s$ to~$t$.
We will return
the value $\alpha 2^n$. To establish the correctness of this reduction we must
show that $\#\varphi = \alpha2^n$, or in other words $\alpha=
\frac{\#\varphi}{2^{n}}$.

For this we define the
function $f$ which maps every assignment of variables $\theta$ to a realisation
$f(\theta)$ of the probabilistic graph $\G = (V,E,\pi)$, where the edge from
$s$ to $u_i$ is present if and only if $X_i$ is true, and the edge from $u$ to
$v_j$ is present if and only if $Y_j$ is true, $\forall X_i \in X$ and  $\forall
Y_j \in Y$. Clearly $f$ is a bijection. We say that a realisation of~$\G$ is
\emph{good} if it contains a path of length exactly $2pq+p+4$ from~$s$ to~$t$,
and \emph{bad} otherwise.

We then prove the following claim, which we denote by $(\dagger)$: the
realisation $f(\theta)$
is good if and only if $\theta$ is a satisfying assignment.

To prove the forward implication, suppose first that $\theta$ is a satisfying
assignment. This implies that for some term of the $\spptdnf$, say $X_i \land
Y_j$,  both variables $X_i$ and $Y_j$ are made true by~$\theta$. Therefore the
edges from $s$ to $u_{i}$ and from $u$ to $v_{j}$ must be present in
$f(\theta)$. Therefore consider the path $s,u_{i},u_{i+1} \dots u_p$, which
has length $1+(p-i)$. Concatenate this with the edge joining $u_p$ and $u$. Next
concatenate it with the path $u,v_{j},v_{j+1} \dots v_q$, which has length $1
+ 2p(q-j)$. Lastly concatenate it with the edge from $v_q$ to $v$, and
then with the path of length $2pj+i$ from $v$ to $t$ corresponding to  $X_{i} \land Y_{j}$ in the $\spptdnf$ formula.

The length of the path thus formed from $s$ to $t$ is then $1 + (p-i) + 1 +
1 + 2p(q-j) + 1 + 2pj+i$, i.e., $2pq + p + 4$.

We now show the reverse implication. Suppose a path of length $2pq+p+4$
exists. It must use exactly one path from $v$ to $t$. Suppose it uses the
path of length $2pj+i$ corresponding to a term $X_i \land Y_j$.
The path must take exactly one blue edge of the form $(s, u_{i'})$ and
exactly one blue edge of the form $(u,v_{j'})$, in addition to the edges
$(u_p,u)$ and $(v_q,v)$. One can see that the path length is then $1 + (p-i') +
1 + 2p(q-j') + 1 + 1 + 2pj+i$, that is, $2p(q-j'+j) + (p-i'+i) + 4$. By
assumption this must be equal to $2pq+p+4$. Now, subtracting $4$ from both
sides, and using the fact that $0 < p-i'+i < 2p$ we can consider the Euclidean
division of each side of the equation by $2p$ to conclude that $(q-j'+j) = q$
and $p-i'+i=p$, so that we must have $j=j'$ and $i=i'$.
This implies that $\theta$ is a satisfying assignment since the term $X_i \land Y_j$ is satisfied. 

Hence, this proves $(\dagger)$ and concludes the proof except that the
graph constructed so far has treewidth~$2$ and not pathwidth~$2$.
We last explain how we can change the
construction of the graph to ensure that it has pathwidth~$2$. The modification
is that the parallel paths from $v$ to~$t$ having different possible lengths are
replaced by one path (for the longest length) together with shortcut edges to
achieve the other lengths. Formally, we create one path from $v$ to~$t$ of
length $M \coloneq \max_{1 \leq r \leq k} (2p)j_r + i_r$, and for each other desired
length $(2p)j_{r'} + i_{r'}$ for some $r'$ we
create one edge from $v$ to the right vertex of this path to also ensure that
there is a path of length $(2p)j_{r'} + i_{r'}$ from~$v$ to~$t$. This
modification has no effect on the path lengths, more precisely the obvious
bijection between the realisations of the original graph and of the modified
graph ensures that there are precisely the same path lengths from~$s$ to~$t$, so
the rest of the correctness proof applies to the modified graph as well. (Note
in particular that all edges in the modified part of the graph are
deterministic.) Last, to explain why the modified graph has pathwidth~$2$, we
have in fact already explained how to build a path decomposition of the vertices
reachable from~$s$ and having a path to~$v$, so all that remains is to continue
the decomposition of the remainder of the graph. For this, we simply have bags
of the form $\{v, w_i, w_{i+1}\}$, for $w_1, \ldots, w_{M-1}$ the intermediary
vertices created on the long path, and finish by one bag $\{v, w_{M-1}, t\}$.
This concludes the proof.
\end{proof}

%% file: variants.tex
We now complete our picture of the complexity landscape of $\probname$ problems by considering 
\emph{weighted probabilistic graphs}, i.e., directed graphs where each
edge is annotated both with a positive length (written in binary) and with a probability. 
The length requirements in $\probname$ problems now ask about the existence of simple
paths or walks whose total weight (i.e., the sum of the weights of all traversed
edges) matches the requirement.

\subparagraph*{Non-probabilistic weighted graphs.}
The presence of weights written in binary has an effect on the complexity of the
$\probname$ problems, 
already in the non-probabilistic setting.
For instance, whereas $\lpwe$ was tractable on
arbitrary graphs (see \cref{sec:prelims})
and $\lppe$ was tractable on bounded-treewidth
graphs (\cref{thm:nonprobatw}),
it is easy to reduce the SUBSET-SUM problem to the exactly-$\ell$ problems,
already on pathwidth-2 graphs:

\begin{proposition}
  \label{prp:weightedhard}
  The $\lpwe$ and $\lppe$ problems are NP-hard already on pathwidth-2 non-probabilistic
  weighted graphs.
\end{proposition}

\begin{figure}
  \begin{tikzpicture}[xscale=1.3,yscale=.6]
    \centering
\begin{scope}[every node/.style={circle,thick,draw}]
    \node (s) at (0,0) {$v_0$};
    \node (n1) at (2,0) {$v_1$};
    \node (t) at (10,0) {$v_k$};
    \node (n2) at (4,0) {$v_2$};
    \node (n3) at (6,0) {$v_3$};
    \node[draw=none] (cdots) at (7,0) {$\cdots$};
    \node (n5) at (8,0) {$\!\!v_{k-1}\!\!$};

    \node (t1) at (1,-2) {$u_1$};
    \node (t2) at (3,-2) {$u_2$};
    \node (t3) at (5,-2) {$u_3$};
    \node (t5) at (9,-2) {$u_k$};
\end{scope}

\begin{scope}[>={Stealth[black]},
              every node/.style={fill=white,circle},
              every edge/.style={draw,thick}]
              \path [->] (s) edge[bend left= 60] node {$n_1+2$} (n1);
              \path [->] (n1) edge[bend left=60] node {$n_2+2$} (n2);
              \path [->] (n2) edge[bend left=60] node {$n_3+2$} (n3);
              \path [->] (n5) edge[bend left=60] node {$n_k+2$} (t);
              \path [->,thick,  line cap=round] (s) edge[bend left=10]  (t1);
              \path [->,thick, dotted, line cap=round, draw = blue] (t1)
                edge[bend left=10]  (n1);
              \path [->,thick,  line cap=round] (n1) edge[bend left= 10]
                 (t2);
              \path [->,thick, dotted, line cap=round, draw= blue] (t2)
                edge[bend left= 10]  (n2);
              \path [->,thick,  line cap=round] (n2) edge[bend left= 10]
                 (t3);
              \path [->,thick, dotted, line cap=round, draw= blue] (t3)
                edge[bend left= 10]  (n3);
              \path [->,thick, line cap=round] (n5) edge[bend left= 10]
                 (t5);
              \path [->,thick, dotted, line cap=round, draw= blue] (t5)
                edge[bend left= 10]  (t);
\end{scope}
\end{tikzpicture}
\caption{The weighted graph constructed from an instance of SUBSET-SUM
  in the proof of \cref{prp:weightedhard} (ignoring the colours and
  dashing of edges). The edges are annotated by
  their weights: the weight is 1 if not indicated. The same graph is reused as a
  probabilistic weighted graph for the proof of \cref{thm:weighthard}, with the
  solid edges in black occurring
with a probability of $1$ whereas the blue dotted edges occur with a probability of $\frac{1}{2}$.}
    \label{fig:example-graph}
\end{figure}

\begin{proof}
  We follow an easy reduction from the NP-hard SUBSET-SUM problem. This problem
  asks, given a set $S = \{n_1, \ldots, n_k\}$ of natural numbers and a threshold $\tau$ (all written in binary), whether there exists a subset $S'$ of $S$ such
  that $\sum_{i\in S'} i = \tau$.
  We construct the pathwidth-2 graph $G$ shown in
\cref{fig:example-graph} (ignoring for now the edge dashing and edge colours): formally, we
have vertices $v_0, \ldots, v_k$ and $u_1, \ldots, u_k$, and for each
$1 \leq i \leq k$, we have edges of weight~1 from $v_{i-1}$ to $u_i$ and from
$u_i$ to~$v_i$, and an edge of weight $n_i+2$ from $v_{i-1}$ to $v_i$.
The graph obviously has pathwidth~$2$, with bags $\{v_0, u_1, v_1\}, \ldots,
\{v_{k-1}, u_k, v_k\}$.
Further, there is
a clear bijection between the paths from~$v_0$ to~$v_k$ and the subsets $S'$
of~$S$, given by the choice of whether to go from $v_{i-1}$ to $v_i$ for each
edge by the length-2 path of weight 2 (we say that the path \emph{does not
select} $n_i$) or the single edge of weight $n_i+2$ (we say that the path
\emph{selects} it). Hence,
the weight of such a path is $2k + \sum_{i\in S'} i$ for the subset $S'$ of the
integers that it selects.
Via the bijection, this  establishes that
there is a path from $v_0$ to~$v_k$ of total weight $\tau +
2k$ iff there is a subset $S'$ of $S$ that sums to~$\tau$. Thus,
the problem $\lppe$ is NP-hard. As the graph is acyclic, the same immediately
holds for
$\lpwe$.
\end{proof}

Note that the same reduction (up to contracting the $u_i$ vertices)
shows that hardness also holds on pathwidth-1 multigraphs, i.e., graphs where we may have multiple parallel edges with the same ordered pair of endpoints but different weights.

To avoid lengthening the paper, we leave to future work a complete exploration
of the complexity landscape of $\probname$ problems on weighted graphs of bounded and
unbounded treewidth and pathwidth. We note that the question is already intriguing in the
setting of non-probabilistic graphs. Another interesting question would be to
compare the setting of non-probabilistic weighted graphs to that of
probabilistic non-weighted graphs, but there one can already remark that the tractable
cases are not identical. For instance, the problems $\lppl$ and $\lpwl$
are intractable on non-weighted probabilistic
graphs (with or without weights; see \cref{thm:bla}), but they can be easily
solved on weighted non-probabilistic graphs by computing shortest paths (e.g.,
via Dijkstra's algorithm).

Our goal in the remainder of this section is
simply to showcase how our previous tractability proofs on bounded-treewidth
probabilistic graphs rely on the absence of edge weights, by giving some
illustrative results and leaving a more complete study to future work. We will first show
that all our problem variants are \#P-hard on probabilistic weighted graphs, even when the
input is assumed to have pathwidth at most~$2$. Then, we will show that the problem $\lpwe$ is \#P-hard
already on probabilistic weighted graphs of pathwidth~$1$.

\subparagraph*{Hardness on pathwidth-2 graphs.}
We show the following result:

\begin{propositionrep}
  \label{thm:weighthard}
  All $\probname$ variants 
  are \#P-hard on probabilistic weighted graphs, even
  when the input graphs are required to have pathwidth at most~$2$.
\end{propositionrep}

By \cref{thm:flpp}, we know that $\lppe$ and $\lpwe$ are
\#P-hard, so it suffices to show this for the $\leq$-variants and the
$\geq$-variants. We reduce from the \#P-hard $\knapsack$ problem~\cite{barceló2020modelinterpretabilitylenscomputational}:

\begin{definition}
  \label{def:subssum2}
  The problem $\knapsack$ asks, given a set of natural numbers $S= \{n_1,
  \allowbreak \ldots, n_k\}$ and a threshold $\tau$, all written in binary, how many subsets $S' \subseteq S$ satisfy $\sum_{i
  \in S'} i \leq \tau$.
\end{definition}

We can now show \cref{thm:weighthard}:

\begin{proofsketch}
  We focus on $\lppl$ and $\lpwl$.
  Given an instance $S, \tau$ of $\knapsack$,
  we build a graph $G$ 
  similarly to the proof of \cref{prp:weightedhard}, but
  some edges are made probabilistic -- see \cref{fig:example-graph} for an
  illustration. The
  realisations of $G$ then correspond to subsets of $S$, and we can ensure that the subset
  sums up to at most~$\tau$ iff the corresponding realisation has some path of
  weight at most~$\tau + 2k$ from $v_0$ to~$v_k$. For $\lppg$ and $\lpwg$, we do a
  variant of this reduction.
\end{proofsketch}
\begin{toappendix}
\begin{proof}[Proof of \cref{thm:weighthard}]
We first show the claim for $\lpwl$, which implies the same for $\lppl$ because both problems are equivalent.

    We reduce from the $\knapsack$ problem. Consider an instance of $\knapsack$
    with 
    $S= \{n_1, \ldots, n_k\}$
    and a threshold $\tau$.
    Let $\phi$ be the number of subsets  $S' \subseteq S$ such that $\sum_{i \in S'} i \leq \tau$, which is the number that we want to compute.
    From the instance $S, \tau$, we construct the following probabilistic weighted graph
    $\G = (V,E,\pi,\delta)$ as shown in \Cref{fig:example-graph}:

    \begin{itemize}
        \item Create a source vertex $v_0$ and vertices $v_1,\dots, v_k$
          corresponding to the integers $n_1, \dots, n_k$.

        \item For every $i \in \{1, \dots, k\}$, add an edge from $v_{i-1}$
          to~$v_i$ of weight $n_i+2$ and probability $1$.
        \item For every $i \in \{1, \dots, k\}$, 
          create a vertex $u_i$, create an edge of weight~$1$ and
          probability~$1$ from $v_{i-1}$ to $u_i$,
          and 
          create an edge of weight~$1$ and probability $1/2$ from~$u_i$
          to~$v_{i}$.

    \end{itemize}

    It is easy to observe that $\G$ has pathwidth~$2$, with a path decomposition featuring the bags $\{v_0, u_1,
    v_1\}, \{v_1, u_2, v_2\}, \ldots, \{v_{k-1}, u_k, v_k\}$. Further, we
    consider the following 
    bijection $f$ between the set of subsets $S'$ of $S$ and the set of
    realisations of the graph $\G$:
    each subset $S' \subseteq S$ is mapped to the
realisation of $\G$ where the edge from $u_{i}$ to $v_i$ (shown in blue in
\Cref{fig:example-graph}) does \textit{not} appear if and only if $n_i \in S'$.
Let us say that a subset $S' \subseteq S$ is \emph{good} if its sum is no greater than
$\tau$, and \emph{bad} otherwise. Let us
further say that a realisation of $\G$ is \emph{good} if it has a path from $v_0$ to~$v_k$ which is no longer than $\tau+2k$, and \emph{bad} otherwise.
We will now show the following claim~$(\dagger)$:
The set $S'$ is good iff the realisation $f(S')$ of $\G$ is good.
As the probability of each realisation of~$\G$ is $1/2^k$, it is clear that $(\dagger)$ implies that we can obtain the answer to the $\knapsack$ instance (i.e., the number of good subsets $S' \subseteq S$)
from the total probability of the good realisations of~$\G$, simply by multiplying by~$2^k$. Hence, we show $(\dagger)$ in what follows.

First, we show that for any $S' \subseteq S$, if $S'$ is good then $f(S')$ is good.
Consider any set $S'$ such that $\sum_{i : n_i \in S'} n_i \leq \tau$. Let the realisation $G = (V,E')$ be the image of $S'$ under $f$. Now consider the path $P$ described as follows: $P$ contains the vertices $v_0, v_1, \dots v_k$ in that order (in addition to possibly some intermediate vertices). If $n_i \in S'$ then $P$ contains the edge of weight $n_i+2$ from $v_{i-1}$ to $v_i$, else $P$ contains the path of length~2 via~$u_i$ from $v_{i-1}$ to $v_{i}$, edges of which exist in~$G=(V,E')$ because $n_i \notin S'$. Observe now that the total weight of $P$ is $\sum_{i : n_i \in S'} (n_i + 2) + \sum_{i : n_i \notin S'} 2$, i.e., it is $2k + \sum_{i : n_i \in S'} n_i$.
Since $\sum_{i : n_i \in S'}n_i \leq \tau $, we have $\sum_{i : n_i \in S'}n_i +2k  \leq \tau + 2k$. Thus, $f(S')$ is a good realisation.

Second, we show that for any $S' \subseteq S$, if $f(S')$ is good then~$S'$ is good.
So let $S'\subseteq S$ and assume that the realisation $(V, E') = f(S')$ of~$\G$ is good, i.e., that it has a path $P$ from $v_0$ to~$v_k$ which has length less than or equal to $\tau +2k$. We assume without loss of generality that $P$ is a shortest path from $v_0$ to~$v_k$,
so that it uses the path of length~$2$ between any two adjacent vertices $v_{i-1}$ and $v_i$ whenever it is present.
Thus, the length of~$P$ is 
$\sum_{i: n_{i} \in S'}(n_i+2) + \sum_{i: n_{i} \notin S'}{2} = \sum_{n_{i} \in S'}{n_i}+2k$.
However, we know that $P$ has length $\leq \tau+2k$.
Thus, we have $\sum_{n_{i} \in S'}{n_i} \leq \tau$, so $S'$ is good. Thus, we have established both directions of $(\dagger)$.

Therefore if we have an oracle for the $\lpwl$ problem, we use the oracle to compute the probability $p$ that there exists a path of length less than or equal to $\tau+2k$. Then $2^kp$ gives the number of subsets of $S$ which sum to a value less than or equal to $\tau$, for the given instance of $\knapsack$, which concludes the reduction
and establishes the result for $\lpwl$ and $\lppl$.

\bigskip

We last explain how the proof can be modified to show the result for $\lppg$ and
$\lpwg$. Given an instance of $\knapsack$ we modify the graph in the proof slightly while maintaining the same structure. In particular, the only change that we enforce is that all edges of weight $1$ now appear with probability $1$, whereas the weighted edges of weight $n_i+2$ where $i \in \{1,2,3, \dots, k\}$ now appear with a probability $\frac{1}{2}$. 

     Consider the map from subsets of the given set $S$ to realisations of $\G$
     in the proof above. Observe that a realisation of $\G$ has a path of length
     greater than $\tau+2k$ if and only if the preimage corresponding to it is a
     subset of $S$ which sums to a value greater than~$\tau$. In other words a subset which sums to a value less than or equal to $\tau$ maps to a realisation of $\G$ which does not have a path of length at least $\tau+2k+1$.

     Thus we call the oracle for the $\lpwg$ problem to obtain the probability that there exists a path of length at least $\tau+2k+1$. Let $p$ be this probability. Then $(1-p)2^{k}$ gives us the number of subsets of $S$ which sum to a value less than or equal to $\tau$, concluding the reduction for $\lppg$ and $\lpwg$. 
\end{proof}
\end{toappendix}
\begin{figure}
\centering
  \begin{tikzpicture}[xscale=1.3,yscale=.6]
\begin{scope}[every node/.style={circle,thick,draw}]
    \node (s) at (0,0) {$s$};
    \node (n1) at (-3,0) {$v_1$};
    \node (n2) at (-1,3) {$v_2$};
    \node[draw=none] (n3) at (1,3) {$\cdots$};
    \node (n4) at (3,0) {$v_k$};
\end{scope}
\begin{scope}[>={Stealth[black]},
              every node/.style={fill=white,circle},
              every edge/.style={draw,thick}]
              \path [->] (s) edge[bend left=10,draw = blue,line cap=round,dotted] node[below] {$n_1$} (n1);
              \path [->] (s) edge[bend left=10, draw = blue,line cap=round,dotted] node[left]
              {$n_2$} (n2);
              \path [->] (s) edge[bend left=10, draw = blue,line cap=round,dotted] (n3);
              \path [->] (n1) edge[bend left=10] node[above] {$n_1$} (s);
              \path [->] (n2) edge[bend left=10] node[above] {$n_2$} (s);
              \path [->] (n3) edge[bend left=10]  (s);
              \path [->] (s) edge[bend left=10, thick, dotted, line cap=round, draw = blue
              ] node[above] {$n_k$} (n4);
              \path [->] (n4) edge[bend left=10] node[below] {$n_k$} (s);
\end{scope}
\end{tikzpicture}
\caption{The probabilistic weighted graph constructed from an instance of
  $\multsubsm$ in the proof of \cref{thm:weighthardforest}. The edges are annotated by their weights. The solid edges in black occur 
with a probability of $1$ whereas the blue dotted edges occur with a probability of $\frac{1}{2}$}

\label{fig:weighthardforests}
\end{figure}
\subparagraph*{Hardness on pathwidth-1 graphs.}
We now show that the $\lpwe$ problem, which was PTIME on pathwidth-1 probabilistic graphs
(\cref{thm:forest}), becomes \#P-hard when weights are allowed:

\begin{propositionrep}
\label{thm:weighthardforest}
    For probabilistic weighted graphs of pathwidth $1$, $\lpwe$ is $\#P$-hard.
\end{propositionrep}

\begin{proof}
  We reduce from the $\multsubsm$ problem (\cref{def:multsubsm}), which is \#P-hard by
  \cref{prp:multsubsm}.
    Given an instance of $\multsubsm$ with $S =\{n_1,n_2,\dots n_k\}$ and target
    $\tau$, we construct the graph $\G=(V,E,\pi,\delta)$ as follows, outlined in
    \Cref{fig:weighthardforests}.

    \begin{itemize}
        \item Create a start vertex $s$
        \item For every $n_i \in S$, create a vertex $v_i$
        \item For each vertex $v_i$, create an edge of weight $n_i$ and probability $\frac{1}{2}$ from $s$ to $v_i$
        \item For each vertex $v_i$, create an edge of weight $n_i$ and probability $1$ from $v_i$ to $s$
    \end{itemize}

 Observe that $\G=(V,E,\pi,\delta)$ has pathwidth $1$.
 Now consider the problem of determining if a walk of length $2\tau$ exists from
 $s$ to $s$ in this graph. Consider the obvious bijection $f$ between subsets
 of~$S$ and realisations of~$\G$, i.e., the edge from $s$ to $v_i$ is present in
 $f(S)$ if and only if $n_i \in S$. We can immediately observe that $f(S)$ 
  contains a walk of length $2\tau$ from~$s$ to~$s$ iff~$S$ is feasible. Thus, the
  number of feasible subsets can be computed from the probability that $\G$
  contains a walk of length~$2\tau$ from~$s$ to~$s$, by multiplying by~$2^k$. This
  concludes the proof.
\end{proof}

We will show this by reduction from a counting problem that we dub $\multsubsm$. It is derived from the UBK problem, which is known to be NP-hard~\cite{ubk}:

\begin{definition}
  The \emph{UBK problem} is the computational problem where we are given
  a set $S = \{n_1, \ldots, n_k\}$ of integers and an integer $\tau$ (all given in binary), and we must decide whether
  there is a linear combination of the $n_i$ with nonnegative integer
  coefficients that sums to~$\tau$, i.e., whether there exist nonnegative
  integers $u_1, \ldots, u_k$ such that $\sum_i u_i n_i = \tau$.
\end{definition}

In \#UBK, we are given a UBK instance $S, \tau$ and we must count the number of
subsets $S' \subseteq S$ such that $S', \tau$ is a positive UBK instance. In
other words:
\begin{definition}
  \label{def:multsubsm}
The problem $\multsubsm$ takes as input a set of natural numbers $S= \{n_1,
  \ldots, n_k\}$ and a threshold $\tau$, all written in binary, and it
  outputs the number of subsets $S'$ of $S$ that are \emph{feasible}, i.e., 
  such that there exist nonnegative integers $u_1, \ldots, u_k$ which ensure that $u_i = 0$ whenever $n_i \notin S'$ and that $\sum_{i=1}^{k}n_iu_i=\tau$.
\end{definition}

With standard techniques and some effort, we can 
show the hardness of $\multsubsm$ by reducing from a variant of~$\spptdnf$~\cite{pp2dnfvar}:

\begin{propositionrep}
  \label{prp:multsubsm}
  The $\multsubsm$ problem is \#P-hard.
\end{propositionrep}

\begin{proof}
  We reduce from the problem \#WPP2DNF, for \emph{weighted \#PP2DNF},
  defined as the following variant of \#PP2DNF: given two sets of variables $X = \{X_1, \ldots, X_p\}$ and $Y = \{Y_1, \ldots, Y_q\}$ 
  and a set $\phi$ of terms each of the form $X_i \land Y_j$, we want to know the number of pairs $(\phi', \theta)$ formed of a subset of terms $\phi' \subseteq \phi$ and of valuations $\theta$ of the variables $X_i \cup Y_j$ such that $\phi'$ evaluates to true under~$\theta$. Intuitively, this is like \#PP2DNF except every satisfying Boolean valuation $\theta$ is counted as many times as there are  subsets of terms 
  whose disjunction satisfies~$\theta$. The problem 
  \#WPP2DNF was introduced in~\cite{pp2dnfvar} and is shown hard (under a different but equivalent terminology) 
  in~\cite[Theorem~4.3]{amarilli2022uniform}.

  Now, from an instance $\phi$ of \#WPP2DNF 
  with $m$ terms built on the set of variables $X$ = $\{X_1,\ldots, X_{p}\}$, $Y$ = $\{Y_1,\ldots, Y_{q}\}$,
  we construct an instance of $\multsubsm$ as follows. 
We let $L= \{n_1,n_2, \ldots, n_{|L|}\}$ be a list of $p+q+m$ integers that we define below. All integers are formed of $1+6+p+q$ bits decomposed in four blocks (from most to least significant digits): a $1$, a block of 6 bits describing whether the integer codes an $X$-variable or $Y$-variable or a term, then a block of $p$ bits coding the index of a $X$-variable, then a block of $q$ bits coding the index of a $Y$-variable. The $p+q+m$ integers are as follows:

\begin{itemize}
    \item For every variable $X_i$, we introduce in $L$ an integer $a_i$. The binary representation of $a_i$ is given as follows, from most to least significant bit. The first bit is $1$, and the next six bits are $000001$. The next $p$ bits code the index of the variable in~$X$, i.e., the $i$-th bit among these $p$ bits is $1$ and all other bits are $0$. The next $q$ bits are all $0$ (they intuitively correspond to~$Y$).
    
    \item For every variable $Y_j$, we introduce in $L$ an integer $b_j$ with binary representation given as follows. The first bit is $1$, and the next six bits are $000100$. The next $p$ bits (intuitively corresponding to~$X$) are all $0$. The next $q$ bits intuitively correspond to the set $Y$ and code the index of the variable in~$Y$, i.e., the $j$-th bit among these $q$ bits is~$1$ and all other bits are~$0$.

    \item For every term $X_i \land Y_j$ in $\phi$ we introduce in $L$ an integer $c_{ij}$ with binary representation given as follows. The first bit is $1$, and the next six bits are $010000$. Then the $i$-th bit of the next $p$ bits is~$0$ (coding $X_i$) and the others are~$1$, and then the $j$-th bit of the next $q$ bits is~$0$ (coding $Y_j$) and the others are~$1$.
\end{itemize}
It is easy to check that this instance can be constructed in PTIME and that all resulting integers are pairwise distinct.

    Further we let the target value be $\tau$, whose binary representation comprises $p+q+8$ bits given as follows: the first two bits are $1$, the next six bits are $010101$ and the next $p+q$ bits are all $1$.

Now we claim that the number of solutions to the \#WPP2DNF instance is the same as the number of feasible solutions to the instance of $\multsubsm$ constructed above.
To establish this, we define the
function $f$ which maps every pair $(\theta, \phi')$ of an assignment of variables $\theta$ and a subset $\phi' \subseteq \phi$ of terms
to a set $f(\theta,\phi') \subseteq L$. The map is defined as follows. We have
$a_i \in f(\theta,\phi')$ precisely when the variable $X_i$ is true in~$\theta$,
we have $b_j \in f(\theta,\phi')$ precisely when the variable $Y_j$ is true
in~$\theta$, and we have $c_{ij} \in f(\theta,\phi')$ precisely when $X_i \land Y_j$
is in~$\phi'$.
It is immediate that $f$ is a bijection between subsets of~$L$ and pairs of a variable assignment
and a subset of~$\phi$.

Next we show the following claim ($\dagger$): we have that $\theta$ is a satisfying assignment of a subset $\phi'$
of~$\phi$ if and only if $f(\theta, \phi')$ is a feasible instance of the
$\multsubsm$ instance. The claim ($\dagger$) ensures that the answer to the $\multsubsm$ instance
is precisely the answer to \#WPP2DNF that we need to compute. Hence, all that
remains is to show ($\dagger$).

For this, in the forward direction,
consider a pair $(\theta,\phi')$ counted for the \#WPP2DNF instance, formed of 
a subset $\phi'$ of terms of~$\phi$ and an assignment $\theta$ which satisfies~$\phi'$.
Thus there exist $X_i$ and $Y_j$ which are true under~$\theta$ 
and are such that $X_i \land Y_j$ occurs in~$\phi'$.
Hence the integers $a_i,b_j,c_{ij}$ occur in $f(\theta,\phi')$.
Now, observe that $a_i+b_j+c_{ij}=\tau$. This shows that $f(\theta,\phi')$
is a feasible subset of $L$.

Now, in the backward direction,
let $f(\theta,\phi')$ be a feasible subset of $L$. We show that $\theta$
must be a satisfying assignment of~$\phi'$.
Since $f(\theta,\phi')$ is a feasible subset of $L$, there exists a tuple of weights $\{x_1,x_2 \dots x_{|L|} \}$, such that $\sum_{i=1}^{|L|}n_ix_i= \tau$, 
and only integers in $f(\theta,\phi')$ have a nonzero weight in the tuple.
We first claim that the tuple must consist of at most three integers,
i.e., $\sum_{i=1}^{|L|} x_{i}\leq 3$. This is because, looking at the two leading
bits $11$ of the target value, we know that every integer in~$L$ has value $01$
on the bits at this position (remembering that the integers in~$L$ are one bit shorter than the target value). Thus, summing four or more of them, we would obtain a sum with 
at least $p+q+9$ bits, which would overshoot the target value~$\tau$.

Now, we know that there exists at most three integers in the set $L$: we now
show that there must be exactly three and that there must be one of the form
$a_i$, one of the form $b_j$, and one of the form $c_{i,j}$. Indeed, 
in the target value~$\tau$,
the first two bits are $11$, and the following six bits are given by $010101$.
Now, the bits at the corresponding position is $010000$ for the $c_{i,j}$, 
$000100$ for the $b_j$, and
$000001$ for the $a_i$.
Thus, if two or three $c_{i,j}$'s are added 
then the first among these six bits in the sum must be~$1$, which does not match
the target.
We can apply the same reasoning to the $b_j$ about the next two bits in the sum,
and so on, so we know that we can take at most one integer of
each type. Now, to achieve the target, it is then clear that we must take in
fact precisely one integer of each type.

Thus, we must end up adding exactly one of the integers corresponding to the set of variables $X$ (say $a_i$ for the variable $X_i$), exactly one of the integers corresponding to the set of variables $Y$ (say $b_j$ for the variable $Y_j$), and exactly one of the integers corresponding to the set of terms. 
We now claim that the integer corresponding to the set of terms must be
$c_{ij}$ (for the term $X_i \land Y_j$). Indeed, 
if we consider the sum of the integers $a_i+b_j$,
only the $i^{th}$ position among the $p$ positions corresponding to the set of variables $X$ and the $j^{th}$ among the $q$ positions corresponding to the set of variables $Y$ are $1$ and the rest are $0$. 
The target requires us to add to $a_i+b_j$ an integer that consists of $1$ at all but these two positions. Thus, the third integer corresponding to a term is uniquely given by $c_{ij}$.
Now since $a_i,b_j,c_{ij} \in f(\theta,\phi')$ we must have that $X_i \land Y_j$
is a term of~$\phi'$ and that $X_i, Y_j$ are true in $\theta$, showing that
$\theta$ is a valid assignment of~$\phi'$. This establishes the backward
direction of ($\dagger$), so ($\dagger$) holds and the result is proven.
\end{proof}

We can then show \cref{thm:weighthardforest} by giving a natural reduction from 
$\multsubsm$ to $\lpwe$ on probabilistic weighted graphs of treewidth~$1$.

\begin{proof}[Proof sketch of \cref{thm:weighthardforest}]
  From a $\multsubsm$ instance $S = \{n_1, \ldots, n_k\}$ with target $\tau$, we build the weighted probabilistic graph of
  \Cref{fig:weighthardforests}, where precisely one edge of each back-and-forth
  path from the centre vertex~$s$ has probability~$1/2$. We then ask for the existence of a walk of length
  exactly~$2\tau$ from~$s$ to itself. It is not hard to see that the realisations of
  the probabilistic weighted graph having such a walk precisely correspond to
  the feasible subsets of~$S$.
\end{proof}

We remark that we expect $\lpwe$ to be the only problem among our variants to be
\#P-hard on treewidth-1 weighted probabilistic graphs, i.e., we do not expect an
analogue of \cref{thm:weighthardforest} for other problem variants, and expect that they are in fact tractable. Indeed,
all simple-path problem variants are tractable because there is at most one
witnessing path in a treewidth-1 instance, and the problem $\lpwl$ is easy
because it suffices to consider shortest paths which are again unique in
treewidth-1 instances. We last expect that $\lpwg$ is in PTIME by an adaptation
of the argument of \cref{thm:forest}.
However, we leave the
detailed study of this to future work.

%% file: conclusion.tex
We conclude the paper by summing up the results that we have achieved,
before mentioning some possible directions for future work.

\myparagraph{Summary} All $\probname$ problems are \#P-hard on arbitrary probabilistic graphs
(\cref{thm:bla}), with $\lppe$ and $\lpwe$ being hard already on probabilistic
graphs of pathwidth $2$ (\cref{thm:flpp}). The problems $\lppl$, $\lpwl$,
$\lppg$ and $\lpwg$ are, on the other hand, all tractable on probabilistic
graphs of bounded treewidth (\cref{thm:tract}): they are in XP parameterised by
treewidth, but W[1]-hard parameterised by pathwidth.
Further, $\lppe$ and $\lpwe$ are
also tractable on probabilistic graphs of treewidth $1$ (\cref{thm:forest}).
When we additionally add binary weights on the edges, all problems become hard
for pathwidth $2$ (\cref{thm:weighthard}), with 
$\lpwe$ being hard
already for pathwidth $1$ (\cref{thm:weighthardforest}).

\myparagraph{Future work}
Our work leaves open many intriguing questions. Can we show, for instance, that
bounding the treewidth of graphs is a necessary condition to obtain the
tractability of $\lppl$, $\lpwl$, $\lppg$ and $\lpwg$? In other words, can we show that
these problems are hard whenever the probabilistic graph family has unbounded
treewidth, as is the case for instance for the problems studied in
\cite{amarilli2016tractable,amarilli2022weighted}?
A related question would be to show hardness in the \emph{uniform reliability}
setting where the input probabilistic graphs only allow edges of probability
$1/2$ (in particular allowing no edges of probability~$1$), in the spirit of 
\cite{amarilli2022uniform}. This should be contrasted with the \emph{symmetric}
setting where the graph is a complete graph with all edges having the same
probability (i.e., we also do not allow the probability~$0$), which sometimes
makes probabilistic query evaluation 
tractable~\cite{beame2015symmetric}.

We also leave to future work the task to complete the picture when we have
binary edge weights in the non-probabilistic setting: we showed in
\cref{prp:weightedhard} that $\lppe$ and $\lpwe$ are hard on treewidth-$2$
graphs, and we know that $\lppl$ and $\lpwl$, as well as $\lpwg$, are tractable
on arbitrary weighted non-probabilistic graphs; this leaves open the $\lppg$
variant, which is NP-hard even without weights but which we expect to be FPT parameterised by treewidth.

Another question left to future work is whether the W[1]-hardness results of
\cref{thm:w1hard} could be strengthened to also establish \#W[1]-hardness.
We note that this question is explicitly left open for the source problem from
which we reduce in the proof of \cref{thm:w1hard} (counting linear
extensions parameterised by pathwidth or treewidth); see
\cite[Section~5]{eiben2019counting}.

Last, one could ask whether there is some quantitative version of MSO that would
capture the tractability of our problems on bounded-treewidth probabilistic
graphs, e.g., MSO with counting quantifiers as in
\cite{kreutzer2013quantitative,knop2019simplified}. However, this seems unlikely because such a logic would have to capture
precisely the tractable variants and not the intractable ones, i.e., it should
not capture $\lppe$ and $\lpwe$ which are hard for pathwidth $2$. 